\documentclass[11pt, letterpaper]{article}
\usepackage[utf8]{inputenc}
\usepackage[USenglish]{babel}
\usepackage{algorithm2e}
\usepackage{amsmath,amssymb,amsfonts,amsthm}
\usepackage{mathabx}
\usepackage{url}
\usepackage{fullpage}
\usepackage[hyperfootnotes=false]{hyperref} % Don't hyperlink footnotes, to avoid conflict with footmisc
\usepackage{cleveref}
\usepackage[usenames,dvipsnames,table]{xcolor}
\usepackage{graphicx}
\usepackage{booktabs}
\usepackage{authblk}
\usepackage{todonotes}
\usepackage{comment}
\usepackage{enumerate}
\usepackage{enumitem}
\usepackage{scalerel}
\usepackage{thm-restate}
\usepackage[title]{appendix}
\usepackage{soul}
\usepackage{nicefrac}
\usepackage{thm-restate}
\usepackage[multiple]{footmisc}
\usepackage{listings}

\newtheorem{theorem}{Theorem}
\newtheorem{lemma}[theorem]{Lemma}
\newtheorem{corollary}[theorem]{Corollary}
\theoremstyle{definition}
\newtheorem{definition}{Definition}
\newtheorem{observation}[theorem]{Observation}
\newtheorem{claim}[theorem]{Claim}

\newcommand{\bbN}{\mathbb N} %the natural numbers
\newcommand{\calI}{\mathcal I}

\newcommand{\calA}{\mathcal A}

\newcommand{\expfac}{\mathcal F}

\newcommand{\calC}{\mathcal C}

\newcommand{\calO}{\mathcal O}

\newcommand{\derive}[1]{\widetilde {#1}}
\newcommand{\toptab}[1]{\widehat {#1}}
\newcommand{\barh}{\overline{h}}

\newcommand{\cond}{\;|\;}

\newcommand{\eps}{\varepsilon}

\newcommand{\E}{\textnormal{E}}
\newcommand{\Var}{\textnormal{Var}}
\newcommand{\sd}{\triangle}
\newcommand{\xor}{\oplus}

\newcommand{\Prp}[1]{\Pr\!\left[{#1} \right]}

\newcommand{\Prpcond}[2]{\Pr\!\left[{#1} \;\middle|\; {#2} \right]}

\newcommand{\Ep}[1]{\E\!\left[{#1} \right]}
\newcommand{\Epcond}[2]{\E\!\left[{#1} \;\middle|\; {#2} \right]}

\newcommand{\Varcond}[2]{\Var\!\left[{#1} \mid {#2} \right]}

\newcommand{\indicator}[1]{\left[{#1}\right]}

\newcommand{\set}[1]{\left \{ #1 \right \}}

\crefname{lemma}{Lemma}{Lemmas}
\Crefname{lemma}{Lemma}{Lemmas}
\crefname{theorem}{Theorem}{Theorems}
\Crefname{theorem}{Theorem}{Theorems}
\Crefname{corollary}{Corollary}{Corollaries}
\crefname{corollary}{Corollary}{Corollaries}
\crefname{observation}{Observation}{Observations}
\Crefname{observation}{Observation}{Observations}
\crefname{definition}{Definition}{Definitions}
\Crefname{definition}{Definition}{Definitions}
\crefname{section}{Section}{Sections}
\Crefname{section}{Section}{Sections}
\crefname{figure}{Figure}{Figures}
\Crefname{figure}{Figure}{Figures}
\crefname{appendix}{Appendix}{Appendices}
\Crefname{appendix}{Appendix}{Appendices}

\def\Symmdiff{\Delta}

\def\cR{\mathcal R}

\def\cI{\mathcal I}

\def\cU{\mathcal U}

\newcommand{\calJ}{\mathcal J}

\DeclareMathOperator*{\bigsd}{\scalerel*{\triangle}{\sum}}
\newcommand{\size}[1]{\ensuremath{\left|#1\right|}}
\newcommand{\ld}{\left}
\newcommand{\rd}{\right}

\newcommand{\parentheses}[1]{\left(#1\right)}
\newcommand{\ssigma}{|\Sigma|}
\newcommand\req[1]{(\ref{#1})}

\DeclareMathOperator{\select}{\mathtt{select}}
\def\DiffKeys{\textnormal{DiffKeys}}

\newcommand{\DP}{7\mu^3(3/|\Sigma|)^{d+1}+1/2^{|\Sigma|/2}}
\newcommand\drop[1]{}
\newcommand{\DPmax}{\ensuremath{24(3/\ssigma)^{d-3} + 1/2^{\ssigma/2}}} % \DP with mu=\ssigma/2

\newcommand{\DPmaxs}{\ensuremath{24(3/s)^b + 1/2^{s/2}}}

\newcommand{\PnegJ}{\Prp{\neg \calJ}}
\newcommand{\PnegJmax}{\Prp{\neg \calJ_{\max}}}
\newcommand{\Jmax}{ \calJ_{\max}}
\newcommand{\Error}{\mathcal{P}_\textit{error}}
\newcommand{\ValueError}{(c+d-2)\ln(\ssigma) \cdot \parentheses{49\parentheses{\frac{3}{\ssigma}}^{d-3} + 3\parentheses{\frac{1}{2}}^{\ssigma/2}}}
\newcommand{\SValueError}{(c+b+1)\ln(s) \cdot \parentheses{49\parentheses{\frac{3}{s}}^{b} + 3\parentheses{\frac{1}{2}}^{s/2}}}
\newcommand{\SValueErrorPlusOne}{(c+b+2)\ln(s) \cdot \parentheses{49\parentheses{\frac{3}{s}}^{b} + 3\parentheses{\frac{1}{2}}^{s/2}}}

\newcommand{\sall}{s_{all}} % was c_{reg}
\newcommand{\ntop}{n_{top}} % was c_{top}
\newcommand{\ssec}{s_{sec}} % was c (sorry!)
\newcommand{\inr}{i_{nr}}   % was i_{reg}
\newcommand{\preg}{p_{reg}} % was t
\newcommand{\ptop}{p_{top}} % always been p_{top}
\newcommand{\iinf}{i_{\infty}} % always been i_\infty
\newcommand{\imax}{i_{\max}} % Highest layer considered

\newcommand{\wmax}{{w_{\max}}} % Largest allowed obstruction for proof of linear independence

\newcommand{\eventI}[1]{\calI\left(\derive h_{<c+d}\left(#1\right)\right)}

\newcommand{\mX}{X_{\alpha, \, r}}
\newcommand{\selBits}{s}

\title{Hashing for Sampling-Based Estimation\thanks{This work was supported by the VILLUM Foundation grants 54451 and 16582.}}

\author[1]{Anders Aamand}
\author[2]{Ioana O. Bercea}
\author[3]{Jakob Bæk Tejs Houen}
\author[1]{Jonas Klausen}
\author[1]{Mikkel~Thorup}

\affil[1]{University of Copenhagen \texttt{\{aa,jokl,mthorup\}@di.ku.dk}}
\affil[2]{Kungliga Tekniska Högskolan \texttt{bercea@kth.se
}} 
\affil[3]{Alipes ApS \texttt{jakn@di.ku.dk
}}

\begin{document}

\maketitle
\begin{abstract}
Hash-based sampling and estimation are common themes in computing. Using hashing for sampling gives us the coordination needed to compare samples from different sets. Hashing is also used when we want to count distinct elements. The quality of the estimator for, say, the Jaccard similarity between two sets, depends on the concentration of the number of sampled elements from their intersection. Often we want to compare one query set against many stored sets to find one of the most similar sets, so we need strong concentration and low error-probability.

In this paper, we provide strong explicit concentration bounds for Tornado Tabulation hashing [Bercea, Beretta, Klausen, Houen, and Thorup, FOCS'23] which is a realistic constant time hashing scheme.
Previous concentration bounds for fast hashing were off by orders of magnitude, in the sample size needed to guarantee the same concentration. The true power of our result appears when applied in the local uniformity framework by [Dahlgaard, Knudsen, Rotenberg, and Thorup, STOC'15].
\end{abstract}

\thispagestyle{empty}
\newpage
\thispagestyle{empty}
\tableofcontents
\newpage
\setcounter{page}1
%\linenumbers

\section{Introduction}
In designing and analyzing randomized algorithms, a common assumption is that we have access to fully random hash functions. These ideal hash functions have beautiful theoretical properties which turn out to be incredibly powerful in obtaining simple and reliable algorithms with strong theoretical guarantees on their performance. Unfortunately, fully random hash functions cannot be implemented in practice, and instead they serve as the objects of aspiration when studying practical hashing schemes. In other words, they motivate the following high level goal.
\begin{center}
\emph{Provide a simple and practical hashing scheme sharing the most powerful probabilistic properties with fully random hashing.}
\end{center}

An ambitious approach towards this goal is a framework introduced by \cite{dahlgaard15k-partitions} which they use to solve a variety of problems that had thus far been out of reach for any realistic hashing scheme. This framework combines two properties of the hash function, concentration bounds for~\emph{selection} and a notion of~\emph{local uniformity}\footnote{The term~\emph{local uniformity} was coined later by \cite{BBKHT23}}. A hash function enjoying both of these properties, provides powerful probabilistic guarantees for a realm of applications.
We consider the Tornado Tabulation hashing scheme from~\cite{BBKHT23}, a realistic, constant-time, hashing scheme, and demonstrate that it fits this framework like a glove. For instance, our work shows how to implement the hashing underlying most hash-based sampling and estimation
schemes. Below, we describe the framework from~\cite{dahlgaard15k-partitions}.

\subsection{Local Uniformity and Concentrated Selection}
The key point in the machinery of \cite{dahlgaard15k-partitions} is to
combine concentration bounds for~\emph{selection} with
a certain \emph{local uniformity}. We proceed to describe these notions. Local uniformity relates to how the hash function $h$ behaves on a subset of selected keys $X$ from a large set of keys $A$ with $|A|=n$.
The selection of keys is done by looking at the binary representation of their hash values.
In particular, the bits of the hash value are partitioned into \emph{free} and \emph{select} bits such that a key is selected if and only if its select bits match some fixed bitmask. As an example, one could consider selecting all the keys whose hash values are strictly smaller than $16$. In this case, the select bits would be all but the rightmost $4$ bits of the hash value and the bitmask would require that they all be $0$. In this work, we will use the same bitmask for each key in $A$.

With $t$ select bits, we expect to
select $\mu=n/2^t$ keys from our set $A$.
Now suppose $\mu\leq s/2$, where $s$ is a space parameter (for now it suffices to know that the hashing scheme uses space $O(s)$ with the $O$-notation hiding a small constant). Within these parameters, local uniformity implies that, with high probability, the free bits of the selected keys will be fully random. That is, if we define $h^f$ to be the function mapping keys to their free bits, then local uniformity would imply that $h^f$ is fully random on $X$. The above statement might appear a bit cryptic, but for understanding it, it is useful think about the generation of the hash function in two phases. The first phase settles the select bits of all keys and hence decides the selected set of keys $X$. The second phase generates the free bits. The point now is that with high probability over the random process of phase one, the free bits generated in phase two will be fully random for the keys in $X$. Importantly, the selection is \emph{not} known when implementing the hash function, but only a tool for the analysis. For example, if we want to count distinct elements, the number of select bits in the analysis will depend on the number of distinct elements in the input.

As observed in~\cite{dahlgaard15k-partitions}, combining  concentration bounds on the number of selected keys with local uniformity provides a powerful analytical framework for getting theoretical guarantees as if we had used fully random hashing for several important applications. Roughly speaking, for many hashing-based sampling and estimation algorithms, it suffices to understand the distribution of keys hashing to a small local region of the full hash range. The region is defined by the bitmask on the select bits, so whether or not a key hashes to this region is determined by whether or not the select bits of the hash value matches the bitmask. To see the similarity to fully random hashing, we view the generation of the fully random function in the same two phases, first generating the select bits and then the free bits. If we have strong concentration on the number of selected keys, we select almost the same number as with the fully random hash function and the remaining free bits are fully random in both cases. Now when the hash function is used in an algorithmic application, the algorithmic behavior within the select keys could be quite complicated. Nonetheless, the local uniformity framework ensures that, in a black box manner, we retain the same guarantees as if the hashing had been fully random.

%For the analysis, this means that we can make a coupling to the~\emph{fully random} setting with only a few more or a few less keys (the exact number depending on the quality of the concentration bound). In particular, the fully random guarantees (with a slight deviation in the number of keys) carries over in a black box fashion. The scope of this technique might not be immediately clear at face value, but~\cite{dahlgaard15k-partitions} show how it can be applied several important problems including including threshold sampling, bottom-$k$ sampling~\cite{CK07}, and within machine learning and similarity search~\cite{li12oneperm,li11minhash,indyk98ann}. We will review these applications in~\cref{sec:applications}. 

Naturally, the above framework is only as strong as its individual components. Weak concentration bounds in the selection step will affect any application. Similarly, in order to be useful, the hashing scheme should provide the local uniformity property with high probability for realistic parameters. In the following two~\cref{sec:locally-uniform-progress,sec:concentration-progress} we discuss progress and challenges in obtaining strong versions of local uniformity and concentrated selection. We further state our main result in~\cref{sec:concentration-progress}

\subsubsection{Local Uniformity}\label{sec:locally-uniform-progress}
In their original paper, \cite{dahlgaard15k-partitions} considered local uniformity with \emph{Mixed Tabulation} hashing.
We will discuss tabulation-based hashing schemes in detail later. For the present section, it suffices to know that in these schemes, keys from the universe $[u]$ are viewed as bit strings of length $\lg u$ partitioned into a small number, $c$, of characters each consisting of $(\lg u)/c$ bits. Defining $s=u^{1/c}$, the hash functions are defined through a small constant number of independent and fully random look-up tables of size $s$. In particular, the total space usage becomes $O(s)$ where the $O$-notation hides a small constant. 

Using $c+2b+2$ such look-up tables, \cite{dahlgaard15k-partitions} obtained local uniformity with probability at least
  \[1 - \left(\frac{O(\log s)^c}{s}\right)^{b} \, .\]
  Above, the $O$-notation hides constants that are exponential in $c$ and $b$.
  For space $s\to\infty$, this is interesting, but for more reasonable values of $s$,
  the above error probability bound may be above 1.

Recently, \cite{BBKHT23} introduced Tornado Tabulation, coining the term local uniformity in passing. Using $c+b+2$ tables of size $s$, 
they obtained local uniformity with
the much more useful \emph{explicit} probability bound of at least
\begin{equation}\label{eq:local-uniform-prob}
1 - (\DPmaxs) \, .
\end{equation}
Note that this is a quite strong guarantee. For example, if $s=2^{16}$, we only need a few extra look-up tables before the error probability becomes extremely small.
In fact, they proved that local uniformity holds even when the selection is done according to the hashes of additional fixed query keys (e.g., select keys that have the same rightmost $4$ bits as some given query key $q$). They also show that their analysis is essentially tight in the sense that the additive $(3/s)^b$ term is necessary. While minor improvements of the constant 24 might be possible, their result arguably provides us with a hashing scheme with a very satisfactory local uniformity guarantee. Naturally, the next question is whether it provides strong concentration bounds for selection.

\subsubsection{Concentration for Selection}\label{sec:concentration-progress}

In getting concentration bounds on the number of selected keys, we again aspire to emulate the fully random setting. In this setting, the most basic tool we have available is the classic Chernoff bound which gives that if $X$ is the set of selected keys, then for any $\delta\in[0,1]$, 
\begin{equation}
\Pr[|X-\mu|\geq \delta\mu]<2\exp(-\mu\delta^2/3).\label{eq:Chernoff0}
\end{equation}
The concentration bounds we discuss with practical hashing schemes will have the form 
\begin{equation}
\Pr[|X-\mu|\geq \delta\mu]<2\exp(-\mu\delta^2/\expfac)+P.
\end{equation}
We shall refer to $\expfac$ as the~\emph{exponent factor} and to $P$ as the \emph{added error probability}. The exponent factor $\expfac$ plays a major role in this paper and is our main measure of the quality of the concentration bound. To see why, let us for a moment ignore the additive error probability $P$. Then $\mathcal{F}$ becomes a linear factor in the expected number of keys we need to select in order to stay within a desired relative error $\delta$.
Flipping the argument on its head, if we want the concentration bound to hold with probability $1-Q$, we obtain the relative error bound $\delta=\sqrt{\frac{\expfac \log (1/Q)}{\mu}}$. If $\mathcal{F}$ is large, the framework in~\cite{dahlgaard15k-partitions} becomes weak since the corresponding fully random experiment has a very different number of keys. As we will see shortly, all past work on practical hashing with general concentration bounds have suffered from the same issue, namely that $\mathcal{F}$ is a large constant.

Surprisingly, as shown in~\cite{BBKHT23}, there is an 'automatic' way of deriving an upper tail Chernoff bound from the property of local uniformity. Namely, if the selection of $X$ is carried out using Tornado Tabulation hashing, then 
\begin{equation}\label{eq:BBKHT23}
\Pr[|X|>(1+\delta)\mu]<\exp(-\mu\delta^2/3)+\DPmaxs \, .
\end{equation}
%The fact that the added error probability is the same as the error probability in~\eqref{eq:local-uniform-prob} is no coincidence.
Thus (so far) Tornado Tabulation provides the peculiar guarantee of local uniformity combined with upper tail bounds. This suffices for many hash table applications where we only worry about hashing too many keys in the same bucket and where our particular concern is the number of keys colliding with a given query key. Indeed, one of their main applications was hash tables with linear probing, for which they proved Tornado Tabulation hashing behaves similarly to fully random hashing.
However, for statistical estimation problems and in many algorithmic applications, we need~\emph{concentration}, not just upper tail bounds. We next discuss past work and state-of-the-art in obtaining such (two-sided) concentration bounds.
%The work of \cite{BBKHT23} provided no lower tail bound, so their two-sided concentration bounds were no better than those from  \cite{HouenT22:chaos} stated in \req{eq:Chernoff-perm}.
\paragraph{Hashing with Strong Concentration}
Designing a practical hashing scheme with strong concentration is an important and well-studied problem within the field of hashing. In the independence framework of Wegman and Carter \cite{wegman81kwise}, we say that the hash function $h$ is $k$-independent if every $k$ keys are mapped independently and uniformly into the hash range. We know from \cite{schmidt95chernoff} that in this case, the number of selected keys is concentrated as
\begin{equation}\label{eq:Chernoff-k-univ}
\Pr[|X-\mu|\geq \delta\mu]<2\exp(-\mu\delta^2/3)+\exp(-k/2).
\end{equation}
In particular, this implies that for some desired error probability $P$, we would need to set $k\geq 2\ln(1/P)$. We could implement $k$ as a degree $k-1$ polynomial, but this gets slow when $k$ is large, and in particular, this is super-constant when $P=o(1)$. 

A completely different approach to obtaining strong concentration bounds is to use tabulation hashing as pioneered by Patrascu and Thorup \cite{patrascu12charhash}. They considered simple tabulation hashing, a scheme dating back to Zobrist~\cite{zobrist70hashing}. Given the space parameter $s$, where $s^c=u$, simple tabulation hashing uses $c$ independent and fully random tables of size $s$ and computes a hash value by doing $c$ table lookups. With tables stored in fast cache, this is faster than computing a degree-2 polynomial. Patrascu and Thorup \cite{patrascu12charhash} proved the following Chernoff-style  concentration bound using simple tabulation hashing. Assuming $\mu\leq n^{1/(2c)}\leq \sqrt s$, for any $\delta\leq 1$:
\begin{equation}\label{eq:Chernoff-simple}
\Pr[|X-\mu|\geq \delta\mu]<2\exp(-\mu\delta^2/\expfac)+1/n^\gamma.
\end{equation}
where $\expfac$ depends exponentially on $c$ and $\gamma$. In the Chernoff bound for fully random hashing \req{eq:Chernoff0}, we had exponent factor 3 and no added error probability.

Getting better concentration bounds has been a main target for research in tabulation-based hashing~\cite{PT13:twist,aamand2020fast,HouenT22:chaos}. For simple tabulation, \cite{aamand2020fast} removed the requirement that $\mu\leq n^{1/(2c)}$ but had an added error probability of $np^\gamma$, so the sampling probability $p$ had to be polynomially small. They also introduced tabulation-permutation hashing, which roughly doubled the number of tables, but removed the restriction on $\mu$ and reduced the added error probability to $1/u^\gamma$. That is, for any $\delta\leq 1$ 
\begin{equation}\label{eq:Chernoff-perm}
\Pr[|X-\mu|\geq \delta\mu]<2\exp(-\mu\delta^2/\expfac)+1/u^\gamma.
\end{equation}
The same bound was achieved for mixed tabulation in \cite{HouenT22:chaos}, which further described the dependence of $\expfac$ on $c$ and $\gamma$ as $\expfac=(c^2\gamma\,\calC )^c$, where $\calC$ is a large unspecified universal constant.
The work of \cite{BBKHT23} provided no lower tail bound, but Tornado Tabulation inherits the two-sided concentration in~\eqref{eq:Chernoff-perm} from~\cite{HouenT22:chaos}.

\paragraph{Our Technical Contribution.}
In this paper, we provide strong explicit lower tail bounds for Tornado Tabulation Hashing with $c+b+3$ tables of size $s$. With this final piece of the puzzle, we get a hashing scheme fitting the powerful framework in~\cite{dahlgaard15k-partitions} with explicit bounds. Below, $X$ is still the set of
selected keys as described above and $\mu=\E[X]$. 
\begin{restatable}{theorem}{prettyOne}\label{thm:pretty1}
   For any \(b \geq 1\) and $c \leq \ln s$, if \(s \geq 2^{16} \cdot b^2\), and \(\mu \in [s/4, s/2]\). For any \(\delta > 0\),
\begin{equation}\label{eq:pretty1}
       \Prp{\size{X} < (1 - \delta)\mu} < 3 \exp\left(\frac{- \delta^2 \mu}{7}\right) + \SValueError\;. 
\end{equation}
\end{restatable}
We note again that the added error probability drops~\emph{rapidly}, even with a small choice of $b$.
The proof of this theorem appears in~\cref{sec:proof-pretty}. The concrete value of $\mathcal{F}=7$ is an artifact of our analysis, and it is likely that a more careful argument will show that $\mathcal{F}$ is even closer to the $3$ in~\eqref{eq:Chernoff0}. Additionally, we will see in the next section that the result can be bootstrapped to give $\mathcal{F}= 3$ at the cost of a constant blow-up in space.
The proof of our lower tail bound requires significantly more work than the proof of the upper tail bound
in \req{eq:BBKHT23} and several completely new ideas. The fundamental 
challenge is that in contrast to upper tail bounds, lower tail bounds must argue about the probability distribution of the keys that are \emph{not selected}. Namely,  standard proofs of the Chernoff upper-bound use the Taylor expansion, of which each term represents the probability that some fixed set of keys is selected. Thanks to local uniformity, the selection of these fixed keys can be viewed as fully random, so bounding these probabilities is straightforward. In contrast, for the lower tail, the hash values of non-selected keys are far from fully random. 

%The most general version of our lower tail bound (\Cref{thm:ugly}) can be found in~\Cref{sec:main-result}, which also includes a technical discussion of our approach and building blocks. ~\Cref{sec:layers} then contains proofs of these technical building blocks.

\subsubsection{Threshold Sampling}\label{sec:threshold-sampling}
To illustrate the power of the framework of~\cite{dahlgaard15k-partitions} combined with strong guarantees on local uniformity and concentration of selection, we here apply it to the simple but fundamental algorithmic primitive \emph{threshold sampling}. For threshold sampling, we view the hash values as numbers in $[0,1)$, and given some sampling probability $p\in [0,1]$, we sample a key $x\in A$ if $h(x)< p$. We would like the set of sampled keys $X$ to accurately represent $A$ in the sense that $|X|/p$ is a reliable estimator for $|A|$.

To apply the framework of~\cite{dahlgaard15k-partitions}, we require that $\mu=\Ep{X}=p|A|\ll s$. For the analysis we pick the smallest $t$ such that $n/2^t\leq s/2$. We then use the $t$ most significant bits as select bits, asking for all of them to be zero. It follows that the sampled keys are all selected, and we can view the threshold sampling as a supsample. By local uniformity, this subsample is fully random and we can apply the classic Chernoff bound~\eqref{eq:Chernoff0}. When $s$ is large compared to $\mu$, the deviation of~\cref{thm:pretty1} diminishes in comparison to the deviation of the fully random threshold sampling and this allows us to bootstrap the theorem to achieve $\mathcal{F}=3$. 
As in the previous section, in the theorem below, we again consider Tornado Tabulation hashing with $c+b+3$ look-up tables of size $s$.

\begin{restatable}{theorem}{subsamplingIntro}\label{subsamplingIntro}
  Let \(h : [u] \to [2^l]\) be a Tornado Tabulation hash function with \(s \geq 2^{16} b^2\) and $c\leq \ln(s)$, \(A\) a set of keys, and \(X=\{x\in A \mid h(x)<p\}\) for some \(p \in [2^l]\). Suppose that \(\mu=\Ep{X} \leq s/278\).
  Then for any \(\delta < 1\), it holds that
  \[ \Prp{\size{\size{X} - \mu} > (1+\delta)\mu} < 5 \exp\left(\frac{- \delta^2 \mu}{3}\right) + \SValueErrorPlusOne. \]
\end{restatable}
The proof of this theorem is given in~\cref{sec:proof-subsamplingIntro}. Comparing to the concentration bounds provided by past work (discussed in the previous section), our new exponent factor $\mathcal{F}=3$ is smaller by several orders of magnitude. Recall that the bounds in~\cite{aamand2020fast,HouenT22:chaos} had $\mathcal{F}=(c^2\gamma \mathcal{C})^c$ where $\mathcal{C}$ is a large unspecified constant and $c$ is such that $s^c=u$. 

The requirement that $s\geq \max(2^{16}b^2,278\mu)$ may seem disappointing but for large-scale applications it is not a big concern. The point is that we think of the Tornado Tabulation hash function $h$ as a single \emph{central} hash function used in the construction of millions of sketches. Storing each sketch requires space $\Omega(\mu)$, so the space used for storing the hash function is insignificant. The setting with many sketches also emphasizes the importance of having high probability bounds since we can then use a union bound to prove that the sketches all behave well simultaneously. Additionally, note that when $s\geq \max(2^{16}b^2,278\mu)$, then the added error probability decays very quickly with increasing \(b\), and we thus require fewer look-up tables for the hash function for a desired added error probability. Finally, the constant $278$ is again an artifact of the analysis which could likely be reduced significantly.

\subsubsection{Relationship to Highly Independent Hashing}
A natural question is how well the classic $k$-independence framework by Wegman and Carter~\cite{wegman81kwise} fits into the local uniformity framework by~\cite{dahlgaard15k-partitions}. First of all, it is clear that to obtain the property of local uniformity, we need space at least $s/2$ for the hash function. Indeed, in expectation $s/2$ keys will be fully random over the free bits. In particular, to employ the $k$-independence scheme, we would need $k\geq s/2$. If we implement the hash function as a degree $k-1$ polynomial, this becomes prohibitively slow. As an alternative, we can use the highly independent hashing introduced by Siegel~\cite{siegel04hash}.
In this setting, Thorup's~\cite{thorup13doubletab} double tabulation provides a simpler and more efficient implementation of such highly independent hashing schemes. Unfortunately, with space $O(s)$, the independence achieved by this construction is $O(s^{1/(5c)})$ which far from suffices for local uniformity. Moreover, with realistic parameters in the probabilistic guarantees, double tabulation is too slow for practical applications (see the discussions and experiments in~\cite{aamand2020fast}).

Even with a hypothetical fast highly independent hashing scheme at hand, we would run into further issues. First, to fit the local uniformity framework, we would need two~\emph{independent} such hash functions, one for the select bits and one for the free bits. This is a bigger issue than it may seem, since the select bits appear only in the analysis and are not chosen by the algorithm design. In fact, for all of the applications in~\cref{sec:applications}, the select bits depend on the input and for some of the analyses of these applications, we need to apply the local uniformity framework over several different choices of select bits. 

Finally, one may ask if the locally uniformity framework is~\emph{necessary}. For instance, for the threshold sampling in~\cref{sec:threshold-sampling}, it suffices to use a $2\ln(1/P)$-independent hash function to get additive error probability $P$ in~\cref{eq:Chernoff-k-univ}. There are two points to make in regards to this. First, in order to obtain high probability error bounds, say $n^{-\gamma}$, we must have $k=2\gamma \ln(n)$ and both $k$-independent hashing as well as Thorup's construction~\cite{thorup13doubletab} would be prohibitively slow. Secondly, for some of the applications we will discuss in~\cref{sec:applications} (e.g., the important Vector-$k$ Sample), independence below $k$ has no proven guarantees.

\subsection{Roadmap of the Paper}
In~\cref{sec:applications}, we discuss hash based sampling and estimation schemes and how they fit in the local uniformity framework. In~\cref{sec:prelim}, we present the necessary preliminaries on Tornado Tabulation hashing. In~\cref{sec:main-result}, we discuss our main technical contribution and steps we need for the proof of~\cref{thm:pretty1}. We will include a second roadmap by the end of~\cref{sec:main-result} for an overview of where we take these steps.

%\begin{center}
%\color{NavyBlue}\emph{Suggestion for the reviewer: Depending on whether you are more interested in applications or in our theoretical contribution, you can either continue reading~\cref{sec:applications} or skip straight to~\cref{sec:prelim} on preliminaries.}\color{black}
%\end{center}

\section{Applications to Hash-Based Sampling and Estimation}\label{sec:applications}

%A natural approach towards such a bound (which even gives $\mathcal{F}=3$), is to use the classic $k$-independent hashing~\cite{wegman81kwise}. However, this requires that $k=\Omega(\log (1/P))$ and evaluating a hash 
\color{black}

Below in~\cref{sec:app-hash-sample} we discuss different types of hash-based sampling and estimation schemes, starting from the most basic, and moving to those with the highest demand on the hash function. In each
case, we first assume that the hash function $h$ is fully random.
Having seen all these applications
of full randomness, we then argue in~\cref{sec:app-concentration-and-local}, in a black-box fashion, that Tornado Tabulation hashing performs almost as well as a fully random hash function.

\subsection{Hash-Based Sampling and Estimation Schemes}\label{sec:app-hash-sample}
Our starting point is the fundamental threshold sampling of~\cref{sec:threshold-sampling}. We review it again here.
\paragraph{Threshold Sampling.} The most primitive form of hash-based sampling and estimation takes a threshold probability $p\in[0,1]$ and
samples a key $x$ if $h(x)<p$. For
any key set $A\subseteq [u]$, let
$S_p(A)$ be keys sampled from $A$. 
With $X=|S_p(A)|$ and $\mu=\E[X]=|A|p$, by standard Chernoff bounds, for $\delta\leq 1$, we have that
\begin{equation}\label{eq:Chernoff}
\Pr[|X-\mu|\geq \delta\mu]<2\exp(-\mu\delta^2/3).
\end{equation}
Thus, $X/p=|S_p(A)|/p$ is a strongly concentrated estimator
of $|A|$. One advantage to storing
$S_p(A)$ is that $S_p(A)$ can then be used to estimate the intersection size $|A\cap B|$ as $|S_p(A)\cap B|/p$, given any other set $B$.
Having exponential concentration in particular
is critical if we want low error
probability bounds for a union bound over many events. For example, if we
store $S_p(A_i)$ for many sets
$A_i$ and want to estimate the maximal
intersection size with $B$, then
it is important that none of the individual intersection
estimates are too large.
So far, however,
we have not seen the advantage of
hash-based sampling over independent sampling.

\paragraph{Benefits of Hash-Based
or Coordinated Sampling.}
There are two main benefits
to using hashing to coordinate the sampling. One is if the set of keys appear in a stream where a single key may appear multiple times. Then we can easily maintain a sample of the distinct keys.
Another benefit, more critical to this paper, is that if we have sampled from two sets $A$ and $B$, then we can compute the sample of their union $A\cup B$. For threshold sampling, this is done simply as
$S_p(A\cup B)=S_p(A)\cup S_p(B)$, and
likewise for the intersection as
$S_p(A\cap B)=S_p(A)\cap S_p(B)$.
Note that if we had instead sampled independently from $A$ and $B$, then
keys from the intersection would be
too likely to be included in $S_p(A)\cup S_p(B)$
and less likely to be included in $S_p(A)\cap S_p(B)$.

\drop{
\paragraph{Bounded number of samples (might skip this one to reduce material)}
Often we have a resource bound $k$ on the number
of samples we wish to store from a set $A$ which is independent of the size of $A$, e.g., $A$ could be a stream of keys of unknown length. In this
case, we can use the smallest value integer $i$ such that $|S_{1/2^i}(A)|\leq k$, and use
this to define $S^{\leq k}(A)=S_{1/2^i}(A)$. For a stream $A$, we just start with $p=1$ and then we halve $p$ whenever
$|S_p(A)|>k$. Together with $S^{\leq k}(A)$, we also store
the threshold probability $p^{\leq k}(A)=1/2^i$ used for
the sake of estimators, e.g., estimating
$|A\cap B|$ as $|S^{\leq k}(A)\cap B|/p^{\leq k}(A)$. When analyzing the probabilistic performance, we know $|A|$, and thanks to concentration, we know that, w.h.p., we end up 
stopping with some $p=1/2^i$ such
that $\frac 32 k>p|A|>\frac 38 k$, 
leaving us at most two choices
for  $i$. We can then do a union bound over concentration bounds for each
possibility, this way at most doubling the
error probability bounds.

Recall that one of the key points of hash-based sampling is that we given
the samples of two sets, can compute the
sample of the union. This is done
as $S^{\leq k}(A\cup B)=S^{\leq k}(S^{\leq k}(A)\cup S^{\leq k}(B))$.
We cannot compute $S^{\leq k}(A\cap B)$ since a small intersection may require a larger threshold probability and involve lost samples, 
but 
we can estimate the size of the intersection as 
as $|S^{\leq k}(A\cup B)\cap (S^{\leq k}(A)\cap S^{\leq k}(B)|/p^{\leq k}(A\cap B)$.
}%end-drop

\paragraph{Bottom-$k$ Sampling and Order Statistics on Hash Values.}
With fully random hashing, the hash
values from a set $A$ are just a uniformly
distributed set $h(A)$ of $|A|$ hash values from $(0,1)$. Let 
$h_{(1)}(A),\ldots,h_{(|A|)}(A)$ denote these hash values in sorted order. Assuming $k\leq |A|$, we know
from order statistics (see, e.g., \cite{David81}) that
$\E[1/h_{(k+1)}(A)]=|A|/k$, so
we can use $k/h_{(k+1)}(A)$ as an estimator
for $|A|$. By definition, $p>h_{(k+1)}(A)\iff |S_{p}(A)|> k$.
Therefore 
\[k/h_{(k+1)}(A)<(1-\delta)|A|\iff
k/((1-\delta)|A|)<h_{(k+1)}(A)
\iff |S_{k/((1-\delta)|A|)}|>k
.\]
Here $\E[|S_{k/((1-\delta)|A|)}|=k/(1-\delta)$, so lower tail bounds for
$|S_p|$ with $p={k/((1-\delta)|A|)}$ imply similar lower tail bounds for $k/h_{(k+1)}(A)$, and likewise for
upper tail bounds. This argument for concentration of $1/h_{(k+1)}$ around $|A|/k$ is essentially taken from \cite{bar-yossef02distinct}, except
that they use $1/h_{(k)}$ which is
slightly off in the sense that its mean is
$|A|/(k-1)$.

As in \cite{CK07}, we can also define the bottom-$k$ sample
of $A$ as the subset $S^k(A)$ with
the $k$ smallest hash values. Together
with $S^k(A)$, we can store $p^k(A)=h_{(k+1)}$, and
then $S^k(A)=S_{p^k(A)}(S)$. Note that if
we also have the bottom-$k$ sample
of a set $B$, then we can easily create the bottom-$k$ sample for
their union as $S^k(A\cup B)=S^k(S^k(A)\cup S^k(B))$.
Note also that the bottom-$1$ sample is
identical to Broder's famous MinHash~\cite{broder97onthe,broder98minwise}.

\paragraph{Frequency and Similarity.}
Based on the above observations, we now discuss a very powerful analysis for frequency and similarity estimation assuming sampling based on
a fully random hash function. Very generally, given a set $A$, we 
assume that the sampling process is
given the set of (distinct) hash values $h(A)$ and selects a subset $Y$ of these hash values. It then returns the set of keys $S=\{x\in A\mid h(x)\in Y\}$ . For threshold sampling,
the selected hash values are those which hash below the threshold sampling probability, and for bottom-$k$ it is the $k$ smallest hash values that are selected.

With fully random hashing, for a given set
$A$ and a given set $h(A)$ of hash values, the set $S$ is a fully-random sample without replacement from $A$. As a consequence, if $B$ is a subset of $A$, then the frequency of $B$ can be estimated as $|B\cap S|/|S|$. For
a given sample size $|S|$, this 
estimator is the sum of negatively correlated 0-1 variables (does each sample belong to $B$ or not), and all the standard Chernoff bounds, e.g., \req{eq:Chernoff} hold in this case. For bottom-$k$ samples, for
$B\subseteq A$, we can use $|B\cap S^{k}(A)|/h^{(k+1)}(A)$ to estimate $|B|$.
This estimator is unbiased as proved
in \cite{CK07}, and it is concentrated
thanks to the above concentrations of
$|B\cap S^{k}(A)|$ and $1/h^{(k+1)}(A)$.

We are pointing out this analysis because we could do something more lossy using a union bound over different  
concentration bounds, as described in \cite{thorup13bottomk}. Assuming
fully random hashing, we get the
clean arguments presented above using the fact
that the samples are independent.
In a black-box fashion, we are going to argue that Tornado Tabulation hashing is similar to fully random hashing for frequency estimation, hence the above type of reasoning applies.

\drop{
\paragraph{Implementations}
Using $\log 1/p$ independence
we would get it, but this takes super constant to evaluate...
before, with space $s$ and $c$ such that $s^c\geq u$, for error probability $1/s^d$ we
only had the general chaos bound of
$A=(Bcd)^c$ where $B$ is large constant, but here fantastically, we show that if $\mu\leq s/2$ and $s>2^{16}d^2$, then we get a beautiful bound 
so stepping back, in these sampling contexts, if we are willing to use space linear in the number of samples, then we can get much stronger (several orders of magnitude) concentration bounds than the general Chernoff bound from Chaos.
}

\paragraph{$k$-Partition-Min and Distinct Elements.}
We now discuss a very powerful
and efficient way of creating sketches
based on $k$-partitions. We
use the first $\log k$ bits
of the hash value to partition the keys
between $k$ buckets. We refer
to these $\log k$ bits of the hash values as the bucket index, and the remaining bits as the local hash value. 
The idea is to look at the smallest local hash value within each bucket separately. We generally refer
to this approach as $k$-Partition-Min, and it dates back at least to
Flajolet and Martin \cite{Flajolet85counting} who
used it for estimating the number of
distinct elements. The more
recent popular HyperLogLog algorithm  \cite{Flajolet07hyperloglog} is a compressed version, in that it only
stores the number of leading zeros in the smallest local hash value. 

The HyperLogLog sketch is very
easy to maintain and update. When
a new key comes to the bucket, we just have to check if it has more
leading zeros than the current coordinate. This is faster than using a bottom-$k$ approach, where we would need to keep a hash table over the sampled keys in order to check if the incoming key
is new or a repeat. Likewise, given the HyperLogLog
sketches from two sets, it is easy to construct the sketch of their union: for
each bucket, we just have to find the maximal number of leading zeros from
the two input sets.

Computing the estimate of the number of distinct elements from the HyperLogLog sketch is complicated
and the analysis is involved even if we assume fully random hashing (see \cite{Flajolet07hyperloglog}). Luckily, we will be able to claim that Tornado Tabulation hashing performs similarly in a black-box 
fashion, without needing to understand the
details of the estimator. All we need
to know is that it increases monotonically in each coordinate in the HyperLogLog sketch. Indeed, with a fixed hash function, it is
clear that the coordinates of the HyperLogLog sketch can only increase
as more keys are added, and hence
so should the estimate of the number of distinct keys.

\paragraph{Vector-$k$ Sample.}
Another powerful application
of $k$-Partition-Min is when we store, for each bucket, the key with the
smallest local hash value, i.e., the ``min-key''.  For now, we assume that all buckets are non-empty.
For a set $A$, we use $S^{\vec k}$
to denote the vector of these min-keys. This is the One-Permutation Hashing
scheme of \cite{li12oneperm}. If the hash function is fully-random, then the keys in $S^{\vec k}(A)$
are sampled uniformly, without replacement, just like the samples
in the bottom-$k$ sample $S^{k}(A)$. One important difference between
the vector-$k$ and bottom-$k$ sample is that
the vector-$k$ sample is easier to update and maintain, the same as in the case of 
HyperLogLog:  when a key is added,
we only need to go to the bucket it hashes to and compare it with the current min-key. In contrast, with
bottom-$k$, we would need to maintain a priority queue.

A more fundamental difference appears 
when we want to estimate the similarity
of two sets $A$ and $B$.  Then we only have to compare $S^{\vec k}(A)$ and $S^{\vec k}(B)$ coordinate-wise: the Jaccard similarity is estimated as
$\sum_{i=0}^{k-1} \left[S^{\vec k}(A)[i]=S^{\vec k}(B)[i]\right]/k$. 
Comparing coordinate-wise
is necessary for some very important applications. As described
in \cite{li11minhash}, it implies that we can estimate the similarity between sets as a dot-product and
use this in Support Vector Machines (SVM)
in Machine Learning. To get
a standard bit-wise dot-product, \cite{li11minhash} suggest
that we hash
the min-key in each bucket uniformly into $\{{\tt 01},{\tt 10}\}$ (we could earmark the least significant bit of
the hash value of the min-key for this purpose). If the min-keys in a coordinate are different, then with probability 1/2, they remain different, so dissimilarity is expected to be halved
in this reduction.
More importantly, more similar
sets are expected to get larger dot-products, and this is all we need
for the SVM applications. Mathematically, a cleaner alternative is to use
the least significant bit to map the min-key in
a bucket to $\left\{\frac{-1}{\sqrt k},\frac{1}{\sqrt k}\right\}$.
Now, in expectation, the dot-product is exactly the Jaccard similarity.

Having a vector sample is also important for Locality Sensitive Hashing (LSH) \cite{indyk98ann} as explained
in detail in \cite{DKT17}. The point is 
that using $k$-Partition-Min to compute a $k$-vector sample replaces
the much slower approach to computing
the MinHash \cite{broder97onthe,broder98minwise} with $k$ independent hash
functions, using the min-key with
the $i$th hash function as the $i$th coordinate. With this $k\times$MinHash,
we need to compute $k$ hash values for
every key while $k$-Partition-Min
requires only one hash computation per key. This makes a big difference if $k$ is large, e.g., $k=10,000$ as suggested in \cite{Li15}.

A caveat of $k$-Partition-Min is that
if bucket $i$ is empty, then the $i$th sample is undefined. The ``error'' that some bucket is empty happens with probability at most $P$ if 
$|A|\geq k\ln (k/P)$. 
It was shown in \cite{DKT17} that we
can fill the holes with
error probability at most $P$ 
by hashing indexed keys from
$[j]\times A$ where $j|A|\geq \max\{|A|,k\ln (k/p)\}$. The total number
of hash computations are then
at most $\max\{|A|,2k\ln (k/P)\}$,
which is still much better than the
$k|A|$ hash computations needed for $k\times$MinHash.
The resulting
vector-$k$ sample becomes a mix of sampling with and without replacement.
As proved in \cite{DKT17}, assuming fully random hashing, the number of samples from any subset of
$A$ will still be exponentially concentrated as in our Chernoff
bound \req{eq:Chernoff}. 

We note that in the applications of vector-$k$
sample, we are typically comparing one set
with many sets, to find the most similar set.
Concentration is crucial to making sure
that the most similar sample is not just
similar due to noise.

The fundamental challenge in implementing
$k$-Partition-Min with a realistic hash function is that we want
the min-keys of different buckets to act as if they were independent except for
being without replacement. In the $q$-independence paradigm of Wegman and Carter \cite{wegman81kwise}, it is
not clear if any $q$ less than $|A|$ would suffice. Nevertheless, Tornado Tabulation hashing will make all the applications work similarly to fully-random hashing.

We will now discuss how we can apply local uniformity and the concentration bounds to sampling and frequency estimation. Some of the applications are taken from \cite{dahlgaard15k-partitions}, but we review them here to underline the power of Theorem \ref{thm:pretty1}.

\subsection{Applying Concentration of Selection and Local Uniformity}\label{sec:app-concentration-and-local}
We next discuss the power of the local uniformity framework by~\cite{dahlgaard15k-partitions} when employed with a hashing scheme with strong concentration for selection and local uniformity guarantees.

\paragraph{Concentration Bounds with Subsampling.}
We have already discussed the concentration bounds that we obtain for threshold sampling using the local uniformity framework in~\cref{sec:threshold-sampling} and refer the reader to~\cref{subsamplingIntro}

\paragraph{Selecting Enough Keys for Applications.}
The original paper \cite{dahlgaard15k-partitions} did not introduce any new concentration bounds, but below we review
how they used concentration bounds and local
uniformity to analyze the more complex sampling and estimation.

The basic requirement is that the selected keys, with high probability, should contain
all keys relevant to the final estimators (for threshold sampling, this was trivial).
For instance, let us consider
bottom-$k$ sampling. As for threshold
sampling, we will select keys
based on 0s in their $t$ most significant
bits. If this leads to selection of
more than $k$ keys, then we know that we
have selected all keys and hash values relevant to
the estimators.
If $s\geq 5k$ and 
$t$ is the smallest value for which
$\mu=n/2^t\leq s/2$, then $\mu\geq s/4=1.25 k$.
Thus, using our concentration bound in~\cref{thm:pretty1}, we get that, with high probability, we select
more than $k$ keys.

For $k$-partition-min, the selection is
a bit more subtle. We select keys
based on 0s in the $t$ most significant
positions of their local hash value. We
call such a hash value ``locally small'' regardless of the bucket index. With
high probability, we want the smallest
hash value in every bucket to be locally
small. If we select more than $k\ln(k/P)$ 
locally small keys, then, with probability
at least $1-P$, we get one in each bucket.
Thus we must pick $s\geq 5k\ln(k/P)$. To apply~\cref{thm:pretty1}, we of course further have to assume that $s\geq 2^{16}b^2$.

 The extra factor $\ln (k/P)$ for vector-$k$
sampling may seem disappointing, but as
explained in \cite{DKT17}, we do not
know of any other reasonable way to implement
vector-$k$ sampling if we want exponential
concentration bounds. We already mentioned
the issue in using Wegman and Carter's independence paradigm \cite{wegman81kwise}.
Another tempting approach would be to use one hash function to split the keys between buckets, and then use an independent hash function for each bucket. However, the best implementation of MinHash uses tabulation~\cite{dahlgaard14minwise},
and then we would need $k$ sets of tables yielding much worse space overall. Again our contribution
is that we get an explicit and reasonable 
constant in the exponential concentration.

We finally note that while the Tornado Tabulation hash function may dominate the space of the streaming algorithm
producing the vector-$k$ sample of a given
set, the general point is to produce 
vector-$k$ samples for many sets,
and use them as high-quality compact sketches
in support vector machines and locality
sensitive hashing.

\paragraph{Coupling for Counting Keys.}
We now discuss a stochastic dominance argument where we couple a Tornado Tabulation hashing experiment on a set $A$ with a fully random hashing experiment on a slightly different set $A'$. Let us first consider the case of
counting (distinct) keys, as in HyperLogLog
applied to $k$-partition-min. Let $h$ be the
tornado hash function and $\tilde h$ be the
fully-random hash function. Assuming that distinct keys hash to distinct values, the estimator
only depends on the set of hash values. Furthermore, as described above, we have made the selection such that with high probability, the estimator only depends on the hash values of the selected keys.

Now, for both hash functions, we first compute the select bits of the hash values and let $L$ denote the set of hash values matching the bitmask. This defines the sets of selected keys $X=\{x\in A \mid h(A) \in L\}$ and $X'=\{x\in A' \mid \tilde h(A') \in L\}$.
%$X=A\cap h^{-1}(h(A)\cap L)$ and $X'=A'{\tilde h}^{-1}({\tilde h}(A')\cap L)$.
Next, we perform a maximal matching between keys in $X$ and $X'$, thus matching all keys in the smaller set. Since the free bits of the hash values
of the selected keys are fully random in both cases, we can couple the hash values of matched pairs of keys so that matched keys have the same free bits in
both experiments. As a result, we
end up with the following relations 
\begin{align*}|X|\leq|X'|&\iff (h(A)\cap L)\subseteq (\tilde h(A')\cap L)
\implies HLL(A,h)\leq HLL(A',\tilde{h})\\
|X|\geq|X'|&\iff (h(A)\cap L)\supseteq (\tilde h(A')\cap L) \implies
 HLL(A,h)\geq HLL(A',\tilde{h}).
\end{align*}

Above $HHL$ is the HyperLogLog estimator \cite{Flajolet07hyperloglog}
applied to the $k$-Partition-Min sketch. All we need to know is that it
is increasing in the number of leading zeros of the min-key in each
bucket.  We assumed that $h(A)\cap L$ contained min-keys from each bucket. Therefore,  if $\size{X}\leq \size{X'}$, then we get that $HLL(A,h)\leq HLL(A',\tilde{h})$. If, on the other hand, $\size{X'}\leq \size{X}$, then
$(\tilde h(A')\cap L)$ could be missing keys in some bucket. Since 
$h(A)\cap L$ has keys in these buckets, $h(A)$ has at least $t$
leading zeros while $\tilde h(A')$ has $<t$ leading zeros in these
buckets. Therefore implying that $HLL(A,h)\geq HLL(A',\tilde{h})$. In other words, the estimator from $HLL(A',\tilde{h})$ would be lower because it has seen higher hash values $\tilde h(A')$.

The question now is how to set up the parameters such that the Tornado Tabulation hashing estimator is smaller than the fully random estimator with high probability. For this, we pick $A'$ so large that, with high probability, we have
$|X'|\geq |X|$.
For some target error
probability $O(P)$ and $\mu'=\E[|X'|]=|A'|/2^t$, it would be sufficient to have that
\[\mu+\sqrt{3\mu \ln(1/P)}\leq \mu'-\sqrt{2\mu'\ln(1/P)} \;.\]
This is using that we have the Chernoff upper-tail bound from \eqref{eq:pretty1} 
on $|X|$ and the classic Chernoff lower-tail bound on $|X'|$. 
Assuming $\mu'\leq 2\mu$ and $\mu\geq s/4$,
we see it suffices that the following holds
\[\mu'\geq \mu\left(1+\sqrt{12\ln(1/P)/s}+\sqrt{16\ln(1/P)/s}\right),
\]
which in turn holds if
\[
\mu'
\geq 
\mu\left(1+8\sqrt{\ln(1/P)/s}\right)
\]
Since $\mu=|A|/2^t$ and $\mu'=|A'|/2^t$, this
means that if we want $\size{X} \geq \size{X'}$ to hold with
probability $2P+\DPmaxs$, it suffices to compare Tornado Tabulation hashing on $A$ with fully-random
hashing on $A'$ with 
\[|A'|=\left\lceil
|A|\left(1+8\sqrt{\ln(1/P)/s}\right)\right\rceil.
\]
When we want $|X'|\leq|X|$,  we need to employ our new lower-tail bound from \cref{thm:pretty1} on $|X|$ in combination with the classic Chernoff upper-tail bound on $|X'|$.
Thus we want
\[\mu-\sqrt{7\mu \ln(1/P)}\geq \mu'+\sqrt{3\mu'\ln(1/P)} \;.\]
Assuming $\mu\geq s/4$, we see
it suffices that
\[\mu'\leq \mu(1-\sqrt{28\ln(1/P)/s}-\sqrt{12\ln(1/P)/s}), \]
which in turn holds if
\[
\mu'\leq \mu(1-9\sqrt{\ln(1/P)/s})
\]
Thus for $|X'|\leq|X|$ to hold with
probability $3P+\DPmaxs$, it suffices to compare Tornado Tabulation hashing on $A$ with fully-random
hashing on $A'$ where 
\[|A'|=\left\lfloor|A|\left(1-9\sqrt{\ln(1/P)/s}\right)\right\rfloor.
\]

\paragraph{Coupling for Estimating Frequency.}
Finally, we consider the problem of estimating
frequency, again using a coupling argument. This was discussed as a key example in \cite{dahlgaard15k-partitions}, but using
only $O$-notation (hiding large exponential constants). To get more precise bounds, one has to employ the same carefulness as we did above for counting keys.

Here, we have a set $A$ of red and blue keys and we want to estimate the frequency of the least
frequent color since this implies the best frequency estimate for both colors. Assume without loss of generality that red is the least frequent color. A main point in \cite{dahlgaard15k-partitions} is that we
perform selection on two different levels.
If we have $r$ red keys and $n$ keys in total
in $A$, then we let 
$a$ be the smallest number such that
$r/2^{a}\leq s/2$ and $t$ the smallest number
such that $n/2^{t}\leq s/2$. We then first select based on the first $a$ select bits (as a pre-selection). In expectation, this leads
to between $s/4$ and $s/2$ pre-selected red keys. For an upper bound on the Tornado Tabulation estimator, we want more pre-selected red keys in $A'$ (using fully random hashing) than in $A$ (using Tornado Tabulation hashing). On the
red keys, all remaining bits are fully random, so we can use the same coupling as 
we did above, just for counting. We note here that the pre-selection is essential if we want to get good bounds when the frequency of the red keys is very small.

Next, we settle the following $t-a$ select bits. At
this point, we drop all pre-selected keys that weren't also selected in this second step (in effect, we do a sub-selection). From the perspective of the red keys, this sub-selection  is fully random on and so the previous coupling ensures that
The argument for the lower bound is symmetric:  we decrease the number of red keys in $A'$ and increase the total number of keys by adding more blue keys, using the same parameters as we did in the previous subsection where we were counting distinct keys.

\section{Preliminaries: Tornado Tabulation Hashing}\label{sec:prelim}

We now review the formal definition of tornado tabulation hashing, as well as the relevant technical results from Bercea et. al~\cite{BBKHT23}. First, we recall that a \emph{simple tabulation hash function}~\cite{zobrist70hashing,wegman81kwise} is a function from some universe $\Sigma^c$ to some range of hash values $\cR = [2^r]$. Namely, we view the keys as being a concatenation of $c$ \emph{characters} from some alphabet $\Sigma$.
In fact, our space parameter from the introduction is $s=|\Sigma|$. We use $x_1, \ldots, x_c$ to denote these characters, thus $x = x_1 \dots x_c$. A simple tabulation hash function $h$ associates with each character position $i = 1 \dots c$ a table $T_i: \Sigma \longrightarrow \cR$ that maps each character to a fully-random hash value. These tables are independent across different character positions. Given a key $x\in \Sigma^c$, the final hash value of $x$ is computed as an exclusive or of all the individual character lookups: 
$$h(x) = T_1[x_1] \xor \cdots \xor T_c[x_c] \;.$$

A tornado tabulation hash function uses multiple such simple tabulation hash functions, and can be thought of as a two-step process. In the first step, it extends the original key $x\in\Sigma^c$ into a \emph{derived key} $\derive{h}(x)= \derive x \in \Sigma^{c+d}$, where $d\geq 0$ is an internal parameters that controls the final probability bounds we obtain. We refer to $d$ as the number of derived characters. Namely, the first $c-1$ characters of $\derive{x}$ match $x$: if $\derive{x}_i$ denotes the $i^{\text{th}}$ character of $\tilde{x}$, then $\derive{x}_i=x_i$ for all $i<c$. To compute $\tilde{x}_c$, we use a simple tabulation hash function $h_0:\Sigma^{c-1} \rightarrow \Sigma$ and set $\tilde{x}_c = x_c \xor h_{0}(\derive x_1 \cdots \derive x_{c-1})$. This character is often referred to as being twisted. For the remaining characters $\tilde x_{c+1},\ldots, \tilde x_{c+d}$ (the derived characters), we employ a series of simple tabulation hash functions $\derive h_i : \Sigma^{ c+i - 1} \longrightarrow \Sigma$ and set
$$ \derive x_{c+i} = \derive h_{i}\ld(\derive x_1 \cdots \derive x_{i+c-1}\rd) \text{ for } i = 1 \ldots d \;.$$

The last step in computing the hash value is to do one final round of simple tabulation hashing on the derived key. We denote this last round by $\widehat h: \Sigma^{c+d} \longrightarrow \cR$. Then $h(x) = \widehat h(\derive x)$. 

Below is the C-code  implementation of tornado tabulation for $64$-bit keys, with $\Sigma=[2^{16}]$  $c=4$, $d=3$,  and $R=[2^{64}]$. The function takes as input the key $x$, and $c+d$ fully random tables
of size $\Sigma$, containing 128-bit values.
\begin{lstlisting}[language=C,basicstyle=\footnotesize]
INT64 Tornado(INT64 x, INT128[7][65536] H) {
    INT32 i; INT128 h=0; INT16 c;
    for (i=0;i<3;i++) {
        c=x;
        x>>=16;
        h^=H[i][c];}
    h^=x;
    for (i=3;i<7;i++) {
        c=h;
        h>>16;
        h^=H[i][c];}
    return (INT64)h;}
\end{lstlisting}

\medskip\noindent
\textbf{Selection.} We consider the setting in which a key is selected based on its value and its hash value. We do not consider query keys in our selection, as in~\cite{BBKHT23}. Formally, we have a selector function $f:\Sigma^c \times \cR \rightarrow \{0,1\}$ and let $p_x := \Pr_{r \sim \cU(\cR)}\ld[f(x, r)=1\rd]$, i.e., the probability that a key is selected when its hash value is chosen uniformly at random from $\cR$. The set of selected keys is then defined as 
$$ X = \set{x \in \Sigma^c \cond f(x, h(x)) = 1} \;,$$
with $\Ep{\size{X}} = \mu = \sum_{x\in\Sigma^c} p_x$.

Local uniformity is shown for selector functions that select keys based on bitmasks. That is, we partition the bit representation of the final hash value $h(x)$ into $s$ \emph{selection bits} and $t$ \emph{free bits}, and let $h^{(s)}(x) \in [2^s]$ denote the $s$ selection bits of $h(x)$.  Then the selector function $f$ has the property that $f(x,h(x)) = f(x,h^{(s)}(x))$. The remaining $t$ bits of $h(x)$ are denoted by $h^{(t)}(x) \in [2^t]$ and are not involved in the selection process. Going back a step, we can define a similar partition on the bits of the final simple tabulation hash function $\widehat h$. That is, we let $\widehat h^{(s)}(x)$ denote the $s$ selection bits of $\widehat h(x)$ and note that: $h^{(s)} (x)= \widehat h^{(s)}(\tilde x)$. Similarly for the free bits of $\widehat h$, we have $h^{(t)} (x)= \widehat h^{(t)}(\tilde x)$.

\medskip\noindent
\textbf{Linear Independence.} A crucial ingredient in~\cite{BBKHT23} is the notion of linear independence of a set of keys. Consider some set $Y$ of keys in $\Sigma^k$, each consisting of $k$ characters. Then the set $Y$ is linearly independent if, for every subset $Y' \subseteq Y$, the keys in $Y'$ have the following property: there exists a character that occurs an odd number of times in some position $i\in \{1,\ldots, k\}$. Conversely, it cannot be that in each character position, all characters appear an even number of times across the keys in $Y'$.  We then define
$$ \calI(Y)=  \text{the event that the set } Y \text{ is linearly independent} \;.$$

For such linearly independent sets,  Thorup and Zhang  \cite{thorup12kwise} showed that a simple tabulation hash function is fully-random on the set $Y$ if and only if $Y$ is linearly independent.
In the context of tornado tabulation, we focus on sets $Y$ of derived keys and thus, linear independence is an event that depends only on the randomness of $\derive h$. Bercea et. al~\cite{BBKHT23} then show the following:

\begin{restatable}{theorem}{maintechtheorem}\label{thm:tech-random-set}
	Let $h=\widehat h\circ \derive h:\Sigma^c\to\cR$
	be a random tornado tabulation hash function with $d$ derived characters and $f$ as described above. If $\mu \leq \Sigma / 2$,
	then the event $\calI(\tilde h(X))$ fails with probability at most
\begin{equation}\label{eq:old-linear-indep}
    \DP\;.
\end{equation} 	
\end{restatable}

\noindent
They also showed the following:
\begin{theorem}\label{thm:Chernoff-upper}
	Let $h=\widehat h\circ \derive h:\Sigma^c\to\cR$
	be a random tornado tabulation hash function with $d$ derived character and $f$ as described above. Then, for any $\delta>0$,  we have that
	$$ \Pr\ld[{\size{X} \geq (1+\delta)\cdot  \mu \wedge \calI(\derive h(X))  }\rd] \leq \parentheses{\frac{e^\delta}{(1+\delta)^{1+\delta}}}^{\mu}\;.$$
\end{theorem}

We note that, from the above, one can obtain the classic Chernoff-style concentration (without $\calI(\derive h(X))$), by summing the error probabilities from~\Cref{thm:tech-random-set} and~\Cref{thm:Chernoff-upper}.

\begin{comment}
However, unlike~\cite{BBKHT23}, we will consider a slightly different set of derived keys. For a given key $x$, we let $\derive x_{<c+d}$ denote the prefix of length $c+d-1$ of the derived key $\tilde x$.  Based
on $\derive h_0,\ldots,\derive h_{d-1}$, we also
define $\tilde h_{<c+d}$ such that
$\tilde h_{<c+d}(x)=\derive x_{<c+d}$.\ioana{which function do we want, so we are consistent} In this context, we define the following events:
$$ \calI(\tilde h_{<c+d}(X))=  \text{the event that the set }  \derive{h}_{<c+d}(X)  \text{ is linearly independent}  \text{, and}\;,$$ 

$$ \calI(\tilde h(X))=  \text{the event that the set }  \derive{h}(X)  \text{ is linearly independent} \;.$$ 
\end{comment}

\section{Technical Contribution}\label{sec:main-result}
In this section, we describe the setup and techniques used in the proof of~\Cref{thm:pretty1} which we restate below. 

\prettyOne*
 %In the following sections, we will carry out the individual steps and furthermore prove~\cref{subsamplingIntro} on threshold sampling.

\subsection{The Upper-Tail Bound \texorpdfstring{(\Cref{thm:Chernoff-upper})}{Theorem \ref{thm:Chernoff-upper}}}
To appreciate our new lower-tail bound, we first briefly review the simple proof of the upper-tail bound from \cite{BBKHT23}. We will see why the same techniques breaks down for the lower-tail bound. However, it importantly turns out that we can still use some of the techniques for the upper tail bound in the proof of the lower tail bound. Namely, to get in a position to bound the \emph{lower tail}, it helps to exclude certain~\emph{upper tail} error events. We will return to this point shortly.

The upper tail bound is the
classic Chernoff bound as long as
we also ask for the selected derived keys to be
linearly independent, that is,
\[\Pr[X\geq(1+\delta)\mu \wedge \calI(\tilde h(X)) ]\leq\exp(-\mu\delta^2/3).\]
The proof of this statement is basically
the observation that in the standard proof
of Chernoff bounds, the probability bound is
a sum over different sets $Y$ of the
probability that $Y\subseteq X$, and
this probability should be bounded by
$\prod_{x\in Y}p_x$ where $p_x$ is the
marginal probability that $x$ is selected.
If the keys $\derive Y=\tilde h(Y)$ derived from $Y$ are linearly
independent (this depends only on $\derive h$), then when we pick the top simple tabulation function at random, \[\Pr_{\hat h} [Y\subseteq X\mid \calI(\derive h(Y))]=\prod_{x\in Y}p_x.\] 
However, $\cI\wedge (Y\subseteq X)$ 
implies that $\tilde Y$ is linearly independent, and therefore
\[\Pr[Y\subseteq X\wedge \calI(\derive h(Y))]\leq
\Pr[Y\subseteq X\mid \calI(\derive h(Y))]=\prod_{x\in Y}p_x.\]
We would like to do the same kind of argument for the lower tail bound, but here the Taylor expansion in the standard proof also sums over
the probabilities of events that sets $Y$ are non-selected in the sense $Y\cap X=\emptyset$.
However, the hash values of non-selected keys are very dependent. In particular, we
have no upper bound on the probability that
$\Pr[(Y\cap X=\emptyset)\wedge I]$ which would, if independent,
have been bounded by $\prod_{x\in Y}(1-p_x)$.

\subsection{High-Level Analysis for Lower-Tail Bound}
For the analysis, we first partition the selected elements  into buckets based on their last derived character, that is, we let \(X_\alpha\) be the set of selected keys that have \(\alpha\) as their final derived character, i.e., 
$$ X_{\alpha}  = \{ x \in X \mid  \derive x_{c+d}= \alpha\} \;.$$
Note that $X = \bigcup_{\alpha\in \Sigma} X_{\alpha}$ and we define $f = \Ep{\size{X_{\alpha}}} = \mu/|\Sigma|\leq 1/2$. 

With full randomness the $(|X_\alpha|)_{\alpha \in \Sigma}$ would be independent Poisson distributed random variables. With Tornado Tabulation, they are neither independent nor Poisson distributed. However, we can argue that with high probability, in a certain sense they approximate this ideal. Below, we will introduce two \emph{experiments} which describe this sense.

\medskip\noindent
\textbf{Introducing $\bar{h}$.} Key to our analysis is to break up the definition of the hash function in a new way. Specifically, we divide the process of computing $h$ differently. Let $\widehat h_{c+d}:\Sigma \rightarrow \cR$ denote the table corresponding to the last derived character in our top simple tabulation function $\hat h$. In our C-code, this is the last table we look up in before we output the final hash value $h(x)$. Everything that comes before this last table lookup, we denote by  $\bar h: \Sigma^{c} \rightarrow \cR \times \Sigma$ (this includes the computations needed to obtain the full derived key $\derive x$). Note that $\barh$ outputs two values. The first value, denoted as $\bar h[0]$, is a value in $\cR$ and is the exclusive or of the first $c+d-1$ table lookups that $\widehat h$ makes. The second value, denoted as $\bar h[1]$, is equal to the  last derived character $\derive x_{c+d}$. Under this view, the final hash value can be computed as $$h(x) = \bar h[0](x) \xor \widehat h_{c+d}(\bar h[1] (x)) \;.$$
Our tornado hash function $h$ is thus defined by the two independent random variables $\bar h$ and $\hat h_{c+d}$. 
\paragraph{The High-Level Analysis.}
Using the principle of deferred decision, we are going to make two different analyses depending on which of $\bar h$ and $\hat h_{c+d}$ is generated first. Each order provides a different understanding, and at the end, we link the two in a subtle way.
\begin{description}
    \item[Experiment 1.] Suppose we first fix $\bar h$ arbitrarily, while leaving $\hat h_{c+d}$ random. We claim that this does not
    change the expectation, that is,
    $\Epcond{\size{X}}{\bar h}=\Ep{\size{X}}$. Moreover,
    the $|X_\alpha|$ are completely independent, for when $\bar h$ is fixed, then so are all the derived keys, and then
    $X_\alpha$ only depends on the independent $\widehat h_{c+d}(\alpha)$.

The problem that remains is that the distribution of each $|X_\alpha|$ depends completely on how we fixed $\bar h$.
% \mtcom{Somewhere on pretty high level,  we need to talk about that that no $X_\alpha$ can get bigger than $i_{\max}$ thanks to our new analysis} \jonas{Added to \cref{layers-ingredients}.}
    \item[Experiment 2.] Suppose instead we first fix $\hat h_{c+d}$ arbitrarily, while leaving $\bar h$ random. In this case, expanding a lot on the upper-tail bound techniques from \cite{BBKHT23}, we
    will argue that with high probability,
    on the average, the values $|X_\alpha|$ will follow a distribution not much worse than if they were
    independent random Poisson variables.
    More precisely, we will argue that, w.h.p., for
    any $i\in \mathbb{N}$, the fraction
    of $\alpha \in  \Sigma$ for which $|X_\alpha|\geq i$ is not much
    bigger than $\frac{f^i}{i!}$.

    Compare this to a Poisson distributed variable $Y$ with $\E[Y]=f$, where
    $\Pr[Y=i]=\frac{f^i}{i!}e^{-f}$.
    Naturally, \(\Prp{Y \geq i} \geq \Prp{Y=i}\),
    and so the loss of our analysis relative to that of Poisson distributed variables is less than a factor $e^{f}$.

    \item[Linking the Experiments.] We now want to link the above experiments. For $i \in \mathbb{N}$ define
    \[ S_i = \size{\set{\alpha \in \Sigma : \size{X_\alpha} \geq i}} \]
    to be the number of characters $\alpha$ for which  \(X_\alpha\) contains at least \(i\) selected keys. Note that $\size{X} = \sum_{i\in\mathbb{N}} S_i$. In particular, 
    \begin{align}\label{eq:important-expectation}
    \E[|X|]=\E[|X|\mid \bar h]=\E[\sum_{i\in\mathbb{N}} S_i\mid\bar h]
    \end{align}
    by the discussion below Experiment 1.
    As stated
    under Experiment 2, w.h.p., $S_i$ is not much bigger than $\frac{f^i}{i!}|\Sigma|$.

    We now go back to Experiment 1 where $\bar h$ was fixed first. We would like to 
    claim that, w.h.p., $E[|S_i|\mid \bar h]$ is also not much larger than $\frac{f^i}{i!}|\Sigma|$.
    If we can prove this, we are in a good shape to employ concentration bounds, for with $\bar h$ fixed, $S_i$ is the sum of independent 0-1 variables, which are sharply concentrated by standard Chernoff bounds. Moreover, $\size{X} = \sum_{i\in\mathbb{N}} S_i$ and $\Epcond{\size{X}}{\bar h}=\mu$, so the error $X-\mu$ is the sum of the \emph{layer} errors $S_i -\E[S_i\mid\bar h]$ where each $i$ defines a layer. The rough idea then is to carefully apply concentration bounds within each layer and argue that the total sum of errors across layers is not too big.

    %If so we are essentially done, for with $\bar h$ fixed, $S_i$ is the sum of independent 0-1 variables, which are sharply concentrated by standard Chernoff bounds. Moreover, $\size{X} = \sum_{i\in\mathbb{N}} S_i$ and $\Epcond{\size{X}}{\bar h}=\mu$, so the error $X-\mu$ is the sum of the  layer errors $(S_i \mid \bar h) -\E[S_i\mid\bar h]$.

    Initially, we are only able to bound the probability that $S_i$ is not too big, but we need an upper tail bound for $\E[S_i\mid\bar h]$. To get this, we need to link the two experiments. This link between the Experiment 2 analysis yielding high probability bounds on the size of $S_i$ and high probability bounds on $\E[S_i\mid\bar h]$ as needed for Experiment 1 is \cref{lem:conditionalTranslation} below.
\end{description}

\begin{lemma}\label{lem:conditionalTranslation}
	For any \(\lambda\)
	\begin{equation}\label{eq:magic}
	\Prp{\Epcond{S_i}{\barh} \geq \lambda + 1 } \leq 2 \Prp{S_i \geq \lambda} \, .
\end{equation}
\end{lemma}
\begin{proof}
The proof of~\req{eq:magic} is simple but subtle. We know from the discussion on Experiment 1 that conditioned on $\bar h=\bar h_0$ for any fixed $\bar h_0$, the distribution of $S_i$ is a sum of independent $\{0,1\}$ variables. We wish to show that \emph{if} $\barh_0$ is such that $\Epcond{S_i}{\barh=\bar h_0}\geq \lambda+1$, \emph{then} $\Prpcond{S_i \geq \lambda}{\barh = \bar h_0} \geq 1/2$. Conditioned on $\barh=\bar h_0$, this is a question about sums of independent $\{0,1\}$ variables. We will use the following restatement of a corollary from~\cite{jogdeo1968monotone}.

\begin{claim}{\cite[Corollary 3.1]{jogdeo1968monotone}}\label{cor:jogdeo}
Let $(Z_j)_{j=1}^m$ be independent random $\{0,1\}$ variables, $Z=\sum_{j=1}^m Z_j$, and $\mu=\E[Z]$. Then $\Pr[Z\geq \E[Z]-1]>1/2$.
\end{claim}
To prove our lemma, write
	\begin{align*}
            \Prp{S_i\geq \lambda}\geq \Prp{\Epcond{S_i}{\barh} \geq \lambda + 1 }\cdot \Prp{S_i\geq \lambda\mid \Epcond{S_i}{\barh} \geq \lambda + 1 }
	\end{align*}
We thus have to show that $\Prp{S_i\geq \lambda\mid \Epcond{S_i}{\barh} \geq \lambda + 1 }\geq 1/2$. This follows more or less directly from~\cref{cor:jogdeo}, but let us write out the proof. Let $H=\{\bar h_0\mid \Epcond{S_i}{\barh=\bar h_0} \geq \lambda + 1 \}$. From the corollary
and the fact that conditioned on any $\bar h=\bar h_0$, $S_i$ is the sum of independent $\{0,1\}$ variables, it follows that for any $\bar h_0\in H$, $\Pr[S_i\geq \lambda \mid \barh=\bar h_0]\geq 1/2$. Now we may write
\begin{align*}
            \Prp{S_i\geq \lambda\mid \Epcond{S_i}{\barh} \geq \lambda + 1 }=\frac{\sum_{\bar h_0\in H}\Pr[S_i\geq \lambda \mid \barh=\bar h_0]\cdot \Pr[\barh=\bar h_0]}{\Pr[\bar h \in H]}\geq \frac{\frac{1}{2}\sum_{\bar h_0\in H}\Pr[\barh=\bar h_0]}{\Pr[\barh \in H]}=\frac{1}{2},
\end{align*}
where we used the general formula $\Pr[A\mid \bigcup_{j\in J} B_j]=\frac{\sum_{j\in J}\Pr[A\mid B_j]\Pr[B_j]}{\Pr[\bigcup_{j\in J} B_j]}$ for disjoint events $(B_j)_{j\in J}$. This completes the proof.
\end{proof}
%\begin{proof}
%The proof of \req{eq:magic} is simple but subtle. We know that $(S_i \mid \bar h)$ is a sum of independent 0-1 variables. Therefore, by \cite[Corollary 3.1]{jogdeo1968monotone},  $\Pr[(S_i \mid \bar h)<\E[S_i \mid \bar h]-1]\leq 1/2$. Then
%	\begin{align*}
		%\Prp{\Epcond{S_i}{\barh} \geq c+1} &= \frac {\Prp{S_i \geq c \land \Epcond{S_i}{\barh} \geq c+1}} {\Prpcond{S_i \geq c}{\Epcond{S_i}{\barh} \geq c +1}} \\
		%&\leq \frac {\Prp{S_i \geq c}} {\Prpcond{S_i \geq \Epcond{S_i}{\barh}-1 }{\Epcond{S_i}{\barh} \geq c+1}} \\
		%&\leq 2 \Prp{S_i \geq c} \, .
	%\end{align*}
%\end{proof}
In fact, we do not have a direct way of bounding \(\Prp{S_i \geq c}\) and we will have to  consider a more restricted event which we can bound the probability of using the local uniformity result from~\cite{BBKHT23}. This will be the content of \cref{conditionalTranslationFix} below.
%Instead, we will consider a more restricted event which makes \(S_i\) behave like a sum of independent events, see \cref{conditionalTranslationFix}.
%\begin{proof}
%The proof of \req{eq:magic} is simple but subtle. We know that conditioned on $\bar h$, $S_i$ is a sum of independent 0-1 variables. Thus, by \cite[Corollary 3.1]{jogdeo1968monotone}, we have that $\Pr[S_i<\E[S_i \mid \bar h]-1\mid \bar h]\leq 1/2$. Then
%\begin{align*}
%\Prp{S_i \geq c}&=\Prp{S_i \geq c\mid \Epcond{S_i}{\barh} \geq c + 1 }\Prp{\Epcond{S_i}{\barh} \geq c + 1 } \\
%\geq &=
%\end{align*}
	%\begin{align*}
	%	\Prp{\Epcond{S_i}{\barh} \geq c+1} &= \frac {\Prp{S_i \geq c \land \Epcond{S_i}{\barh} \geq c+1}} {\Prpcond{S_i \geq c}{\Epcond{S_i}{\barh} \geq c +1}} \\
	%	&\leq \frac {\Prp{S_i \geq c}} {\Prpcond{S_i \geq \Epcond{S_i}{\barh}-1 }{\Epcond{S_i}{\barh} \geq c+1}} \\
	%	&\leq 2 \Prp{S_i \geq c} \, .
	%\end{align*}
%\end{proof}
%Unfortunately, we don't have a direct way of bounding \(\Prp{S_i \geq c}\).
%Instead, we will consider a more restricted event which makes \(S_i\) behave like a sum of independent events, see \cref{conditionalTranslationFix}.

\medskip
Let us now reflect on what we achieved above and how it is useful for our goal. We are interested in lower tail bounds for $\E[X]=\sum_{i=1}^\infty S_i$ and we already observed in~\eqref{eq:important-expectation} that $\sum_{i=1}^\infty \E[S_i\mid \bar h]=\E[X]$.
We would like to do a high probability bound over $\bar h$ to ensure that it has certain 'good' properties. If these properties hold, we hope to provide lower tail bounds for $\sum_{i=1}^\infty \E[S_i\mid \bar h]$.
Conditioned on $\bar h$, each $S_i$ is a sum of $\{0,1\}$ variables, and we know how to prove concentration bounds for such sums. Below we list the good properties which we will show that $\bar h$ satisfies with high probability.
\begin{enumerate}
    \item $\sum_{i=1}^\infty\E[S_i\mid \bar h]=\E[X]$. \label{need1}
    \item $\E[S_i\mid \bar h]\lesssim \frac{f^i}{i!}$.
    \item $\E[S_i\mid \bar h]=0$ when $i$ is larger than some $\imax$.
\end{enumerate}
We already saw (Experiment 1) that 1.~holds with probability 1. 
For 2., we will see in~\cref{sec:experiment2} how we can use an~\emph{upper tail} bound for Tornado Tabulation hashing to bound the probability that $S_i$ is large (Experiment 2), and then the result will follow from~\cref{lem:conditionalTranslation} which links the two experiments. Finally, for 3., we need to prove a stronger version of the local uniformity theorem of~\cite{BBKHT23} appearing as~\cref{thm:new-tech-random-set}. The proof is technical but follows a similar path to the one used in~\cite{BBKHT23}. However, we do require a novel combinatorial result for bounding dependencies of \emph{simple tabulation hashing}. This result is~\cref{lem:combinatorial2ttuples} in~\cref{app:zero-sets}

\paragraph{Layers}
With the properties 1.-3. in hand, the rough idea next is to prove a bound for each \emph{layer} $i$ of the form $\Pr[S_i\leq \E[S_i\mid \bar h]-\Delta_i]\leq p_i$. Note that each $p_i$ will include the additive error probability $\Pr[\bar h \text{ not good}]$. 
As the layers are not independent, we have to union bound over each layer to bound the total error. Defining $\Delta=\sum_{i=1}^{\imax}\Delta_i$ and $p=\sum_{i=1}^{\imax}p_i$, we obtain
\[
\Prp{|X|\leq \E[|X|]-\Delta}\leq \sum_{i=1}^{\imax}\Pr[S_i\leq \E[S_i\mid \bar h]-\Delta_i] +p
\]
which will be our desired bound. 
The most technical part of our proof thus employs various concentration bounds for events of the form $[S_i < \Epcond{S_i}{\barh} - \Delta_i]$  depending on the values of $\bar\mu_i:=\E[S_i\mid \bar h]$. Namely, we partition these $S_i$'s  into two different types of layers and use different lower tail techniques depending on which kind of layer we are dealing with. The challenge lies in setting the relative deviation $\Delta_i$ of each layer so that we get the desired overall deviation for $\size{X}$ and we do not incur a large penalty in the probability by conditioning on $\barh$. We distinguish between \emph{bottom} layers (which have large $\mu_i$), \emph{regular} layers and \emph{non-regular} layers (which have excessively small $\mu_i$). %We now describe them in more detail. 

%While the idea is simple, correctly selecting the deviations $\Delta_i$ and corresponding error probabilities $t_i$ turns out to be a significant technical challenge which is the content of~\cref{sec:layers}. 

\paragraph{Bounding Bucket Sizes.} The fact that, we need that $\E[S_i\mid \bar h]=0$ for all $i\geq \imax$ comes from the fact that we need to union bound over all layers, and the error probability $t_i$ for a single layer $i$ includes the additive $\Pr[\bar h \text{ not good}]$. Without an upper bound on $i$, we might even have to union bound over $|\Sigma|^c$ layers which will come as a significant cost for our error bounds. In fact, the upper limit \(\imax\) will be such that, with high probabiliy over \(\barh\), \(\size{X_\alpha} \leq \imax\) for all \(\alpha \in \Sigma\). Our bound on $\imax$ appears in~\cref{lem:imax} with the proof appearing in~\cref{sec:imax}.

\medskip\noindent
In the next two subsections, we will zoom in on Experiment 1 and Experiment 2.
\subsection{Experiment 1}
In this experiment,  we first fixed $\bar h$ arbitrarily, noting that then
the $|X_\alpha|$ are completely independent since $X_\alpha$ now only depends on $\widehat h_{c+d}(\alpha)$. We also claimed that
the fixing of $\bar h$ did not
change the expectation of $|X|$. More specifically, we claim that conditioning on $\barh$ does not change the probability that a specific key $x$ gets selected. 
\begin{observation} Let $I_x$ be the  indicator random variable that is $1$ if $x$ gets selected and $0$ otherwise. Then, for every fixed key $x$ and every fixed value of $\bar h$, we have that:
$$\Prp{I_x} = \Prpcond{I_x}{\barh} \;.$$
\end{observation}
\noindent
This result is well-known from~\cite{BBKHT23}, but we include a proof for completeness.
%The reasoning behind the observation is rather subtle and we include it here in full:
\begin{proof} 
	We first note that $\Prp{I_x} = \sum_{r\in \cR} \Prp{h(x) = r} \cdot \Prpcond{I_x}{h(x) = r}$. By definition, once we fix the key $x$ and its hash value, then selection becomes deterministic. Therefore $$\Prpcond{I_x}{h(x) = r} = \Prpcond{I_x}{h(x)= r \wedge \barh}\;.$$
	
	The only thing left to prove therefore is that $\Prp{h(x) = r} = \Prpcond{h(x) = r}{\barh}$. On one hand, we know that $h$ hashes uniformly in the range of hash values, so $\Prp{h(x) = r}= 1/|\cR|$. On the other hand, we know that 
	
	\begin{align*}
			\Prpcond{h(x) = r}{\barh} &= \Prpcond{\bar h[0](x) \xor T_{c+d}(\bar h[1] (x)) = r}{\barh} \\
			&= \Prpcond{T_{c+d}(\bar h[1] (x)) = r\xor \bar h[0](x) }{\barh}\\
			&=1/|\cR| \;,
		\end{align*}
	where the last inequality holds because the last table lookup $T_{c+d}(\bar h[1] (x))$ picks a value uniformly at random from $\cR$.
\end{proof}

As seen in the proof above, once we condition on $\barh$, the randomness in selection only comes from the last table lookup.  That is, conditioned on $\barh$, the random variables $\{X_{\alpha}\}_{\alpha \in \Sigma}$ become independent, i.e.,  elements across different  $X_{\alpha}$'s  will be selected independently. This, however, is not enough. That is because if we condition on $\barh$, we no longer know how many of the selected keys have a particular last derived character. Thus, even though the random variables $\{X_{\alpha}\}_{\alpha \in \Sigma}$  are independent, they have different, unknown distributions. We cannot  bound their variance nor apply Chernoff on their scaled versions  and get competitive bounds.

\subsection{Experiment 2}\label{sec:experiment2}
In this experiment, we first fix $\hat h_{c+d}$ arbitrarily, while leaving $\bar h$ random. As stated, we will analyze this case expanding on the techniques from \cite{BBKHT23}.

For a given key $x$, let $\derive x_{<c+d}$ denote the derived key except the last derived character $\derive x_{c+d}$. Moreover, we define
$\tilde h_{<c+d}$ such that
$\tilde h_{<c+d}(x)=\derive x_{<c+d}$. 
We will be very focused on the event that these shortened derived selected keys are linearly independent, and for ease of notation, we define the event
\[\calJ=\calI(\derive h_{<c+d}(X)).\]
Using \Cref{thm:tech-random-set}, we prove that this event happens with high probability. More precisely,
\begin{theorem}\label{eq:local-uniformity}
\[    \Pr[\calJ ]\geq 1- (\DPmax).\]
\end{theorem}
\begin{proof}
    We will prove that
    \[    \Pr[\calJ \mid \hat h_{c+d}]\geq 1- (\DPmax).\]
    Since the bound holds for any $\hat h_{c+d}$, it also holds unconditionally.
    We consider the function $\bar{f} \colon \Sigma^c \times (\cR \times \Sigma)$ defined by $\bar{f}(x, (r, \alpha)) = f(x, r \xor \hat{h}_{c + d}(\alpha))$.
    Since $h(x) = \bar{h}[0](x) \xor \hat{h}_{c + d}(\bar{h}[1](x))$ then $f(x, h(x)) = \bar{f}(x, \bar{h}(x))$ and  $X = \{ x \in U \mid f(x, h(x)) = 1\} = \{ x \in U \mid \bar{f}(x, \bar{h}(x)) = 1\}$.
    When $\hat{h}_{c + d}$ is fixed, $\bar{f}$ is a deterministic function and since $\bar{h}$ is a tornado hash function with $d - 1$ derived characters, the result of~\Cref{thm:tech-random-set} gives the claim.
\end{proof}

\noindent
 We can now think of keys being picked independently for $\alpha$.
 In the same way, as we proved Chernoff upper bounds for events like $[X\geq(1+\delta)\mu \wedge \calI(\derive h(X))]$, we can prove
\begin{restatable}{lemma}{boundxalpha}\label{lemma:bound-xalpha}
For any $i\in \mathbb{N}$, we have that $$\Prp{\size{X_{\alpha}} \geq i \wedge \calJ } \leq \frac{f^i}{i!} \;.$$
\end{restatable}
The proof of~\cref{lemma:bound-xalpha} can be found in~\Cref{sec:prelimlayers}. Recall that $ S_i = \size{\set{\alpha \in \Sigma : \size{X_\alpha} \geq i}}$. By
~\cref{lemma:bound-xalpha} and linearity of expectation,  we have that $$\Ep{S_i \cdot [\calJ]} \leq \size{\Sigma} \cdot \frac{f^i}{i!} = \bar{\mu_i} \;.$$

\noindent
Moreover, we can show (also in~\Cref{sec:prelimlayers}) a bound on the upper tail of $\size{S_i}$ in terms of $\bar{\mu}_i$ as such:

\begin{restatable}{lemma}{bounsi}\label{lemma:bound-si}
For any $\delta>0$:
$$ \Prp{S_i \geq (1+\delta)\cdot \bar{\mu_i} \wedge \calJ} \leq \parentheses{\frac{e^\delta}  {(1+\delta)^{(1+\delta)} }}^{\bar{\mu}_i} \;.$$
\end{restatable}
With these lemmas in place, we will briefly revisit the issue of bounding the conditional expectation \(\Epcond{S_i}{\barh}\) by bounding \(S_i\), as discussed previously.
The following \lcnamecref{conditionalTranslationFix} follows directly from \cref{lem:conditionalTranslation}, bringing in event \(\calJ\) to give an expression that we can bound with \cref{lemma:bound-si}. We note that the event $\PnegJ$ can be bounded using~\cref{eq:local-uniformity}.
\begin{restatable}{corollary}{conditionalTranslationFix}
\label{conditionalTranslationFix}
For any \(\lambda\)
\[\Prp{\Epcond{S_i}{\barh} \geq \lambda+1} \leq 2 \Prp{(S_i \geq \lambda) \land \calJ} + 2 \PnegJ \, .\]
\end{restatable}

%Now, what is the advantage of writing $\size{X}$ as a sum of $S_i$ terms rather than a sum of $\size{X_{\alpha}}$ terms?  To begin with, conditioned on $\barh$, each $S_i$ becomes a sum of independent $0/1$ random variables of the form  ``$\size{X_{\alpha}} \geq i$ ''.  At this point, the intention, naturally, is to bound the probability that $\size{X}<\mu-t$ by bounding the probability of events of the form ``$S_i < \Epcond{S_i}{\barh} - t_i$'', for some carefully chosen values of $t_i$. The issue remains that these random variables now have different distributions (depending on $\barh$), and we even have that $\Epcond{S_i}{\barh} \neq \Ep{S_i}$. 

%\medskip\noindent
%\textbf{The layers.} The most technical part of our proof thus employs various concentration bounds for events of the form ``$S_i < \Epcond{S_i}{\barh} - t_i$''  depending on the values of $\bar\mu_i$. Namely, we partition these $S_i$'s  into different kinds of  \emph{layers} and use different lower tail techniques depending on which kind of layer we are dealing with. The challenge lies in setting the relative deviation $t_i$ of each layer so that we get the desired overall deviation for $\size{X}$ and we do not incur a large penalty in the probability by conditioning on $\barh$. We distinguish between \emph{bottom} layers (which have large $\mu_i$), \emph{regular} layers and \emph{non-regular} layers (which have excessively small $\mu_i$). We now describe them in more detail. 

\subsection{Roadmap of Technical Part of the Paper} In~\cref{sec:layers}, we will provide lower tail bounds for the deviation incurred in each layer. In~\cref{sec:layer-main-ideas}, we describe the main ideas in the analysis and introduce the different types of layers. We provide some preliminary tools for the analysis in~\cref{sec:prelimlayers}. Next, in~\cref{proof:bottom,proof:regular,proof:nonRegularLayers} we analyse respectively bottom layers, regular layers, and non-regular layers. Finally~\cref{sec:no-big-layer} bounds $\imax$.

In~\cref{sec:3c-linear-independent}, we prove our improved local uniformity theorem (needed for bounding $\imax$). We provide necessary definitions (from~\cite{BBKHT23}) in~\cref{sec:prelim-obstructions}. In~\cref{sec:definingObstruction,sec:unionobstruction,sec:confirming}, we show how to modify the obstructions from~\cite{BBKHT23} and how to union bound over them.

\cref{sec:proof-of-theorems} is dedicated to prove our main results. For this, we need a technical theorem which we prove in~\cref{sec:proof-ugly}. 
In~\cref{sec:proof-pretty}, we prove~\cref{thm:pretty1} and in~\cref{sec:proof-subsamplingIntro}, we prove the subsampling~\cref{subsamplingIntro}.

Finally,~\cref{app:zero-sets} contains our new combinatorial result on bounding dependencies for simple tabulation hashing.

In~\cref{app:gen-chernoff}, we include a Chernoff bound working under a slightly weaker assumption than independence. We will need this Chernoff bound in the layer analysis.

%Of some independent interest,
%we give a new bound
%on the probability that
%our selected derived keys are not independent. The new bound is better than \req{eq:old-linear-indep} if $n$ is large compared to $\size\Sigma$, and it
%will play a crucial role in our analysis.

%\noindent

%\input{layers-ingredients} - moved to the next section {in layers.tex}

% !TeX root = main.tex
\section{Layers}\label{sec:layers}

In this section we present and prove technical theorems for the layers, and the associated parameters, as used for proving the main theorems. Our main technical proof is a union bound over events of the form ``$S_i < \Epcond{S_i}{\barh} - \Delta_i$'' that will roughly hold with probability at most $p_i$. We will treat these events differently, depending on the values of \(\bar \mu_i = \ssigma \cdot f^i/i!\), which is an an upper bound on \(\Ep{S_i}\) (see \cref{lem:expected-size}).

Namely, for \(i\) large enough, when \(\bar \mu_i < p\), we can design events such that both the sum of deviations \(\sum \Delta_i\) and associated error probabilities \(\sum p_i\) form geometric series, and are thus finite (see \cref{lem:infiniteLayers-2}).
However, as will be apparent in the statements of the theorems given below, we still incur a small constant error probability for each layer handled, originating from our applications of \cref{conditionalTranslationFix}. 
This accumulated error probability will be too high in general, so we further argue that we incur it only for a (relatively) small number of layers. Namely, in \Cref{lem:imax} 
 (\cref{sec:imax}), we show that,  with high probability,  \(\Epcond{S_i}{\barh} = 0\) for all \(i\) larger than a threshold \(\imax\). Thus, the deviation for these layers will be zero.

The main technical challenge thus lies in handling the layers where \(\bar \mu_i \geq p\).
As \(\bar \mu_1 = \mu\) there will be \(\Omega\left(\ln(\mu) + \ln(1/p)\right)\) such layers before \cref{lem:infiniteLayers-2} applies.
With one event defined for each layer, we are thus dealing with a superconstant number of events, and we will need to perform some scaling of the error probabilities if we want them to sum to \(O(p)\).
Again, the method for doing this depends on the expected size of the layer.

We define quite a few symbols in the treatment of the different layer types. A reference is given in \cref{symbols} on \cpageref{symbols}.
For building intuition, we ignore symbols \(\ssec\) and \(\sall\) in the following paragraphs.
These can both be considered to equal 1 without altering the structure of the proof, which will suffice to build a theorem with error probability \(O(p)\).
The role of these parameters is covered in the final paragraph below and serves to control the constant hidden in the \(O\)-notation.

\paragraph{Roadmap.} In~\Cref{layers-ingredients}, we review and explain the main lemmas we use to bound the deviation in each layer. In particular, we partition the layers into bottom, regular and non-regular. After some preliminaries in~\Cref{sec:prelimlayers}, we then prove each of the lemmas in the following sections. Namely, the proofs  for the bottom layers are given in~\Cref{proof:bottom}. The proof for the regular layers is given in~\Cref{proof:regular}. Finally, the proof for non-regular layers is given in~\Cref{proof:nonRegularLayers}.

% !TeX root = main.tex
\subsection{Main ingredients}\label{sec:layer-main-ideas}
\label{layers-ingredients}
We here discuss the main ideas needed for carrying out the layer analysis.
\paragraph{Regular Layers.}
A layer is said to be \emph{regular} if \(\bar \mu_i \geq \ln(1/p_i)\), in which case \(\Epcond{S_i}{\barh}\) won't be significantly larger than \(\bar \mu_i\), by \cref{lem:layer-size-chernoff}.

As \(\Ep{\size{X_\alpha}} = \mu/\ssigma \leq 1/2\) most elements of \(X\) are expected to be found in the very first layers.
In particular, we handle the combined deviation of layers 1, 2, and 3 through an application of Bernstein's inequality, which gives better bounds than those obtained through individual treatment of the layers.
\begin{restatable}[3 layers with Bernstein]{theorem}{largethree} \label{lem:3layers_2}
	Assume that \(\ln(\ssec/p) \leq \bar \mu_3\). Then
	\[\Prp{S_{\leq 3} \leq \Epcond{S_{\leq 3}}{\barh} - \sqrt{\frac{7}{3} \ln(1/p) \mu \cdot(1+\varepsilon_3)} - \ln(1/p)} < (1 + 4/\ssec) \cdot p + 4 \PnegJ \,.\]
\end{restatable}
If layer 3 isn't regular, we have the following alternate theorem.
\begin{restatable}[2 layers with Bernstein]{theorem}{largetwo}\label{lem:2layers_2}
	Assume that \(\ln(\ssec/p) \leq \bar \mu_2\). Then
	\[\Prp{S_{\leq 2} \leq \Epcond{S_{\leq 2}}{\barh} - \sqrt{2 \ln(1/p) \mu \cdot(1+\varepsilon_2)} - \frac{2}{3} \ln(1/p)} < (1 + 2/\ssec) \cdot p + 2 \PnegJ \,,\]
 where \(\varepsilon_2 = \sqrt{6\ln(\ssec/p) / \ssigma} + (\frac{2}{9} \ln(1/p) + 2)/\mu\).
\end{restatable}
Note that \(\varepsilon_2 \leq \varepsilon_3\) (see \cref{symbols}), and we thus use the latter in the statement of \cref{thm:ugly}.
Both theorems are proven in \cref{proof:bottom}.

If the following layer(s) are also regular we treat these individually with \cref{lem:regularLayer_2}, which first gives an upper bound on \(\Epcond{S_i}{\barh}\) and then bounds the deviation between the conditional expectation and \(S_i\) through a second Chernoff bound as described in the previous section.

Let \(\inr\) be the first layer which is not regular. That is, \(\inr = \min\set{i : \bar \mu_i < \ln(1/p_i}\).
Then \cref{lem:regularLayer_2} is applied to \(\inr-4\) layers in total.
As \(\bar \mu_i = \mu f^{i-1}/i! \leq \mu/2 \cdot (1/6)^{i-2}\), we have \(\inr - 4 \leq \log_6\left(\frac{\mu/2}{\ln(1/p)} \right)\) when all \(p_i \leq p\).
Setting \(p_i = \preg = \frac{p}{\log_6\left(\frac{\mu/2}{\ln(1/p)}\right)} \) for each layer \(i \in \set{4, \dots, \inr-1}\) would thus ensure that the total error probability on these layers is \(O(p)\).

However, as allocating a smaller error probability incurs a larger deviation relative to the expected size of the layer, it seems unwise to set the same low error probability for all regular layers.
They will have a progressively smaller impact on the final result, after all.

Instead, we let \(p_4 = p\) and set \(p_{i+1} = \max\set{p_i / e, \preg}\)\ such that
\(\sum_{i=4}^{\inr-1} p_i \leq (1.59+ 1) \cdot p\).
Applying \cref{lem:regularLayer_2} with these values for \(p_i\) gives the following bound, the proof of which is found in \cref{proof:regular}.
\begin{restatable}{theorem}{regular}\label{lem:regularLayers_2}
	\[\Prp{ \sum_{i=4}^{\inr-1} S_i < \sum_{i=4}^{\inr-1} \Epcond{S_i}{\barh} - \Delta_{reg} } < (1.59 + 1/\sall) \cdot (1+2/\ssec) \cdot p + (\inr - 4) \cdot 2\PnegJ \, .\]
\end{restatable}

\paragraph{Non-regular layers.}
The non-regular layers are handled by three different lemmas: \Cref{lem:layeri0_2} for the single layer \(\inr\) where \(\bar\mu_i \approx \ln(1/\preg)\), \cref{lem:infiniteLayers-2} for layer \(\iinf\) and up where \(\bar \mu_i < p\), and \cref{lem:topLayers_2} for the layers between the two.

To keep the total error probability of these ``top'' layers at \(O(p)\) we set \(p_i \leq p/\ntop\) for each layer treated by \cref{lem:topLayers_2}, where \(\ntop\) is (an upper bound on) the number of such layers.
As the top layers are those where \(p < \bar \mu_i \leq \ln(1/\preg)\) and we assume that \(\inr \geq 3\) (and thus \(\bar \mu_{i+1} \leq \bar \mu_i/6\)) there will at most be \(\log_6(\ln(1/\preg) \cdot 1/p) = \ntop\) of these layers.

As an extra complication, note that \(\inr\) is defined in terms of the threshold \(\preg\) used for the regular layers, and thus the tools used for the non-regular layers only hold for \(p_i < \preg\).
This leads to the somewhat cumbersome definition of \(\ptop = \max\set{\preg, p/\ntop}\).

At this stage, our bounds on \(\Epcond{S_i}{\barh}\) will be smaller than \(\ln(1/p_i)\) and thus a Chernoff bound will no longer give a meaningful bound on the probability that \(S_i\) is smaller than \(\Epcond{S_i}{\barh}\).
Instead we use the trivial observation that \(\Epcond{S_i}{\barh} - S_i \leq \Epcond{S_i}{\barh}\) as \(S_i\) is non-negative.

The combined deviation and error probability of these non-regular layers is summarized in \cref{thm:nonRegularLayers} below, the proof of which can be found in \cref{proof:nonRegularLayers}.

\begin{restatable}{theorem}{nonRegular}\label{thm:nonRegularLayers}
  If \(\inr \geq 3\), then
  \[ \Prp{\sum_{i=\inr}^{\imax} S_i < \sum_{i=\inr}^{\imax} \Epcond{S_i}{\barh} - \Delta_{nonreg}} < p \cdot 6/\sall + (\imax-\inr) \cdot 2 \PnegJ \, .\]
\end{restatable}

\paragraph{Bounding the number of layers.}
As discussed above, we bound the number of layers under consideration to limit the accumulation of error terms from applications of \cref{conditionalTranslationFix}.
Specifically, we let \(\Jmax\) be the event that \(\sum_{i=1}^{\imax} \Epcond{S_i}{\barh} = \mu\), such that summing the deviation found in these layers represents the full deviation between \(\mu\) and \(\size{X}\).
The following \namecref{lem:imax} bounds the probability of \(\Jmax\).
\begin{restatable}{lemma}{imaxlemma}\label{lem:imax}
  Let \(\selBits \leq \size\Sigma/2\) be the number of selection bits and \(\imax = \ln(\ssigma^{d-2} 2^\selBits)\). Then
  \[\Prp{ \sum_{i=\imax+1}^\infty \Epcond{S_i}{\barh} > 0} \leq \parentheses{\frac{1}{\ssigma}}^{d-3} + 3^{c+1} \parentheses{\frac{3}{\ssigma}}^{d-1} + \parentheses{\frac{1}{\ssigma}}^{\ssigma/2 - 1} \, . \]
\end{restatable}

In order to obtain this result, we require an improvement of\cref{thm:tech-random-set}. Namely, we need a better bound on the probability that the derived keys are linearly dependent. This requires a modification of the analysis in~\cite{BBKHT23}. The result is as follows. 
\begin{restatable}{theorem}{newtechrandomset}\label{thm:new-tech-random-set}
	Let $h=\widehat h\circ \derive h:\Sigma^c\to\cR$
	be a random (simple) tornado tabulation hash function with $d$ derived characters and $f$ as described above. If $\mu \leq \Sigma / 2$, 
	then the event $\calI(\tilde h(X))$ fails with probability at most
\begin{equation}\label{eq:new-linear-indep}
   3^c\size\Sigma/n \cdot 3\mu^3 (3/\ssigma)^{d+1} +f^{\ssigma/2}   
%\frac{3^c}{n}3.07\mu^3(3/\size\Sigma)^{d}
 \end{equation} 	
\end{restatable}
We note that the first term is a factor
$3^c\size\Sigma/n$ smaller than the bound in \cref{thm:tech-random-set}. Indeed, we prove it by showing a different analysis than the one for \cref{thm:tech-random-set} from \cite{BBKHT23}. Details can be found in \Cref{sec:3c-linear-independent}.

\paragraph{Parameters \(\ssec\) and \(\sall\).}
As seen in the theorems above, parameters \(\ssec\) and \(\sall\) serve to scale the error probabilities.
For our proof, we set \(\ssec=20\) and \(\sall=180\) (as given in \cref{symbols}), but all of the theorems hold regardless of the values chosen -- as long as the same values are used across all layer types.

Specifically, \(\sall\) scales the error probability of most events, (including the threshold \(\preg\) used for the regular layers as well the events defined on non-regular layers) while \(\ssec\) alters the ratio of the error probabilities between the two events considered on each regular layer:
As a slightly looser bound on \(\Epcond{S_i}{\barh}\) has little impact on the deviation obtained from the regular layers, we opt for a more conservative bound on these, in exchange for a smaller error probability.
If we allocate error \(p_i\) for the ``primary'' event, in which we bound the absolute difference \((\Epcond{S_i}{\barh} - S_i)\) for a fixed value of the conditional expectation,
we instead spend error \(p_i/\ssec\) on the ``secondary'' event where we bound \(\Epcond{S_i}{\barh}\).

With this terminology, the regular layers consist of both primary and secondary events, while the non-regular layers exclusively consist of secondary events.
This aligns with an intuitive understanding that the regular layers lead to deviation proportional to \(\sqrt{\ln(1/p)\mu}\) while the non-regular layers contribute deviation in terms of \(\ln(1/p)\).
Note, however, that the error probability of the secondary event, at \(\preg/\ssec\) is what determines the boundary \(\inr\) between regular and non-regular layers, which is why \(\ssec\) ends up appearing in the deviation found in the non-regular layers.

For \cref{thm:ugly} we've set \(\ssec\) and \(\sall\) quite high in order to bring the total error probability down to \(3p\).
At the other extreme, setting \(\ssec=\sall=1\) would make for an equally viable theorem with fewer additive terms in the deviation.
It's total error probability would be roughly \(19p\), however.

 %just added this

%%%%%
\subsection{Preliminaries}\label{sec:prelimlayers}

We need the following general tools for bounding \(S_i\) which, together with \cref{conditionalTranslationFix}, allows us to bound \(\Epcond{S_i}{\barh}\).

\begin{restatable}[Expected size of a layer]{lemma}{lem:expected-size}\label{lem:expected-size}
	\[\Ep{S_i \cdot \indicator{\calJ}} \leq \bar \mu_i = \ssigma \cdot \frac{f^i}{i!}\]
\end{restatable}
\begin{proof}
	Let \(p_x = \Prp{x \in X}\) be the probability that each key \(x\) is selected.
	Note that \(\Prp{x \in X_\alpha} = p_x/\ssigma\) as the last derived character of \(\barh(x)\) is uniformly distributed over \(\Sigma\).
	When restricted to event \(\calJ\) we further have \(\Prp{Y \subseteq X_\alpha \land \calJ} \leq \prod_{x\in Y} \Prp{x \in X_\alpha}\).
	\begin{align*}
		\Prp{\size{X_\alpha} \geq i \land \calJ} &\leq \sum_{\set{x_1, \dots, x_i} \in \binom{n}{i}} \prod_{k=1}^i \frac{p_{x_k}}{\ssigma} \\
		&\leq \frac{1}{i! \cdot \ssigma^i} \sum_{\langle x_1, \dots, x_i\rangle \in [n]^i} \prod_{k=1}^i p_{x_k} \\
		&= \frac{1}{i! \cdot \ssigma^i} \prod_{k=1}^i \sum_{x \in [n]} p_{x} \\
		&= \frac{\mu^i}{i! \cdot \ssigma^i} = \frac{f^i}{i!} \, .
	\end{align*}
	Then
	\[\Ep{S_i \cdot \indicator{\calJ}} = \sum_{\alpha \in \Sigma} \Prp{\size{X_\alpha} \geq i \land \calJ} \leq \ssigma \cdot \frac{f^i}{i!} \, .\]
\end{proof}
\begin{restatable}[Upper tail for layer size]{lemma}{lem:layer-size-chernoff} \label{lem:layer-size-chernoff}
	\[
	\Prp{S_i > (1 + \delta) \bar \mu_i \land \calJ} \leq \left(\frac{ e^\delta}{(1+\delta)^{(1+\delta)}} \right)^{\bar \mu_i}
	\]
	where \(\bar \mu_i = \ssigma f^i/i!\) and \(\delta > 0\).
\end{restatable}
\begin{proof}
    First observe that bounding \((S_i \land \calJ)\) is equivalent to bounding the sum of indicator variables \(\sum_{\alpha \in \Sigma} \indicator{\size{X_\alpha} \geq i \land \calJ}\).
    We will show that these indicator variables satisfy the conditions of a slightly generalized Chernoff bound (\cref{lem:non-independent-chernoff} in~\cref{app:gen-chernoff}) with \(p_\alpha = f^i/i!\) for all \(\alpha \in \Sigma\) such that \(\sum_{\alpha \in \Sigma} p_\alpha = \bar \mu_i\).
    Hence we need to show that, for any set of characters \(\set{\alpha_1, \dots, \alpha_k} \subseteq \Sigma\), we have \(\Prp{\bigwedge_{l=1}^k \size{X_{\alpha_i}} \geq i \land \calJ} \leq (f^i/i!)^k\).

    For each \(x \in \Sigma^c\) let \(q_x = \Prp{x \in X}\) such that \(\sum q_x = \mu\).
    Then
    \begin{align*}
    \Prp{\bigwedge_{\ell=1}^k \size{X_{\alpha_i}} \geq i \land \calJ} &= \Prp{ \bigwedge_{\ell=1}^k \exists A_\ell \in \binom{n}{i} : A_\ell \subseteq X_{\alpha_\ell} \land \calJ} \\
		&\leq \sum_{\substack{\langle A_1, \dots, A_k \rangle \in \binom{n}{i}^k \\ \text{disjoint}}} \Prp{\bigwedge_{\ell=1}^k A_\ell \subseteq X_{\alpha_\ell} \land \calJ} \\
		&\leq \sum_{\substack{\langle A_1, \dots, A_k \rangle \in \binom{n}{i}^k \\ \text{disjoint}}} \, \prod_{x \in \bigcup_{\ell = 1}^k A_k} \frac{q_x}{\ssigma} \leq \frac{1}{\left(i! \cdot \ssigma^i\right)^k} \sum_{\calA \in [n]^{i\cdot k}} \prod_{x \in \calA} q_x \\
  &= \frac{1}{\left(i! \cdot \ssigma^i\right)^k} \prod_{\ell = 1}^{i \cdot k} \sum_{x \in [n]} q_x = \frac{1}{\left(i! \cdot \ssigma^i\right)^k} \prod_{\ell = 1}^{i \cdot k} \mu = \left(\frac{f^i}{i!}\right)^k \, .
	\end{align*}
\end{proof}

\begin{comment}
\medskip\noindent
\textbf{Bounding the conditional expectation.}
\ioana{I moved this to the intro, so we can remove it}
\begin{lemma}[\cite{jogdeo1968monotone}, Corollary 3.1]\label{lem:median}
	Let \(X\) be a sum of independent indicator variables.
	Then
	\[
	\Prp{X < \Ep{X} - 1} \leq 1/2 \, .
	\]
\end{lemma}

\begin{lemma}%\label{lem:conditionalTranslation}
	For any \(c\)
	\[
	\Prp{\Epcond{S_i}{\barh} \geq c + 1 } \leq 2 \Prp{S_i \geq c} \, .
	\]
\end{lemma}
\begin{proof}
	When conditioned on \(\barh\), \(S_i\) is a sum a of independent indicator variables \(\sum_{\alpha \in \Sigma} \indicator{\size{X_\alpha \geq i}}\).
	\Cref{lem:median} gives that
	\[
	\Prpcond{S_i \geq \Epcond{S_i}{\barh} - 1} {\Epcond{S_i}{\barh} \geq c+1} \geq 1/2 \, .
	\]
	
	Then
	\begin{align*}
		\Prp{\Epcond{S_i}{\barh} \geq c+1} &= \frac {\Prp{S_i \geq c \land \Epcond{S_i}{\barh} \geq c+1}} {\Prpcond{S_i \geq c}{\Epcond{S_i}{\barh} \geq c +1}} \\
		&\leq \frac {\Prp{S_i \geq c}} {\Prpcond{S_i \geq \Epcond{S_i}{\barh}-1 }{\Epcond{S_i}{\barh} \geq c+1}} \\
		&\leq 2 \Prp{S_i \geq c} \, .
	\end{align*}
\end{proof}
\end{comment}

% !TeX root = ../FOCS2024/main.tex
\subsection{Bottom layers: Proof of \texorpdfstring{\Cref{lem:2layers_2,lem:3layers_2}}{Theorems \ref{lem:2layers_2} and \ref{lem:3layers_2}}}
\label{proof:bottom}
The proofs of \Cref{lem:2layers_2,lem:3layers_2} are both based on an application of Bernstein's inequality.
For non-centered independent variables \(X_1, X_2, \dots\), Bernstein's inequality states that
\begin{equation} \label{eqn:Bernstein}
  \Prp{\sum X_i \leq \Ep{\sum X_i} - t} \leq \exp\left( \frac{- 0.5 t^2}{\sum \Var[X_i] + t \cdot M/3} \right)
\end{equation}
where \(M\) is a value such that \(\size{X_i} \leq M\).

For \cref{lem:2layers_2} we sum over \(\hat X_\alpha = \min\set{3, \size{X_\alpha}}\) such that \(S_{\leq 3} = S_1+S_2+S_3=\sum_{\alpha\in\Sigma} \hat X_{\alpha}\).
As the \(\hat X_\alpha\) are independent when conditioned on \(\barh\) we can apply Bernstein's when we've established a bound on \(\sum \Varcond{\bar X_\alpha}{\barh}\).

\begin{lemma}\label{lem:conditionalVariance3}
  If \(\ln(\ssec/p) \leq \ssigma f^3/6 = \bar \mu_3\),
  \[\Prp{ \sum_{\alpha \in \Sigma} \Varcond{\hat X_\alpha}{\barh} >  \left(7/6 + \varepsilon \right) \mu} < p \cdot 4/\ssec + 4 \PnegJ \]
  where \(\varepsilon = (2+\sqrt{6})\sqrt{\ln(\ssec/p)/\ssigma} + 6/\mu\).
\end{lemma}
\begin{proof}
  Define \(\hat \mu = \sum_{\alpha} \Epcond{\hat X_{\alpha}}{\barh} = \Epcond{S_{\leq 3}}{\barh}\) and \(\hat{f} = \hat{\mu} / \ssigma\).
  It then holds that
  \[\sum_{\alpha \in \Sigma} \left(\Epcond{\hat X_\alpha}{\barh}\right)^2 \geq \sum_{\alpha \in \Sigma} (\hat \mu/\ssigma)^2 = \ssigma \hat{f}^2\]
  and thus
  \begin{align*}
    \sum_{\alpha \in \Sigma} \Varcond{\hat X_\alpha}{\barh} &= \sum_{\alpha \in \Sigma} \Epcond{\hat X_\alpha^2}{\barh} - \sum_{\alpha \in \Sigma} \left(\Epcond{\hat X_\alpha}{\barh}\right)^2 \\
    &\leq \sum_{\alpha \in \Sigma} \Epcond{\hat X_\alpha^2}{\barh} - \ssigma \hat f^2 \\
    &= \sum_{\alpha\in\Sigma} \Epcond{\hat X_\alpha(\hat X_\alpha - 1)}{\barh} + \sum_{\alpha\in\Sigma} \Epcond{\hat X_\alpha}{\barh}  - \ssigma \hat f^2 \\
    &= 2 \cdot \Epcond{\size{\set{\alpha : \hat X_\alpha = 2}}}{\barh} + 6 \cdot \Epcond{\size{\set{\alpha : \hat X_\alpha = 3}}}{\barh} + \hat \mu - \ssigma \hat f^2 \\
    &= 2 \cdot \Epcond{S_2}{\barh} + 4 \cdot \Epcond{S_3}{\barh} + \hat \mu - \ssigma \hat f^2 \, .
  \end{align*}
  We will now bound \(\Epcond{S_2}{\barh}\) and \(\Epcond{S_3}{\barh}\).
  Applying \cref{lem:layer-size-chernoff} with \(\delta_2 = \sqrt{3 \ln(\ssec/p)/\bar \mu_2}\) we have
  \[
  \Prp{S_2 \geq (1 + \delta_2) \bar \mu_2 \land \calJ} \leq p/\ssec \, .
  \]
  Invoking \cref{conditionalTranslationFix}, we have for \(\mu_2^+ = (1+\delta_2)\bar \mu_2 + 1\)
  \[
  \Prp{\Epcond{S_2}{\barh} \geq \mu_2^+} \leq 2p/\ssec + 2\PnegJ \, .
  \]
  Likewise for \(\delta_3 = \sqrt{3 \ln(\ssec/p) / \hat \mu_3}\) and \(\mu_3^+ = (1+\delta_3) \bar \mu_3 + 1\),
  \[ \Prp{\Epcond{S_3}{\barh} \geq \mu_3^+ } \leq 2p/\ssec + 2\PnegJ \, .\]
  
  Hence, with probability at least \(1- p \cdot 4/\ssec - 4\PnegJ\) we have,
  \begin{align*}
    \sum_{\alpha \in \Sigma} \Varcond{\hat X_\alpha}{\barh} &< 2\mu_2^+ + 4 \mu_3^+ + \hat \mu - \ssigma \hat f^2 \\
    &= 2(1+\delta_2) \bar \mu_2 + 4(1+\delta_3) \bar \mu_3 + \hat \mu - \ssigma \hat f^2 + 6 \\
    &= 2\delta_2 \bar \mu_2 + 4\delta_3 \bar \mu_3 + 2 \ssigma f^2/2 + 4 \ssigma f^3/6 + \hat \mu - \ssigma \hat f^2 + 6 \\
    &= \delta_2 \ssigma f^2 + \frac{2}{3}\delta_3 \ssigma f^3 + \ssigma f^2 + \frac{2}{3} \ssigma f^3 + \hat \mu - \ssigma \hat f^2 + 6 \\
    &= \left(1+\delta_2 + \frac{2f}{3}\cdot(1+\delta_3)\right) \ssigma f^2 + \hat \mu - \ssigma \hat f^2 + 6 \, .
    %% &\leq 2\left((1+\delta)\ssigma f^2/2 + \sqrt{(1+\delta) \ssigma f^2/2}\right) + \bar\mu - \ssigma \bar f^2
  \end{align*}
  Note that the function \(x \mapsto x-x^2\) is increasing in \(x \in [0, 1/2]\).
  As \(\bar f \leq f \leq 1/2\) we thus have \(\bar f - \bar f^2 \leq f - f^2\) and
  \[
  \ssigma f^2 + \hat \mu - \ssigma \hat f^2 = \ssigma (f^2 + \hat f - \hat f^2) \leq \ssigma \cdot f = \mu
  \]
  giving
  \begin{align*}
    \sum_{\alpha \in \Sigma} \Varcond{\hat X_\alpha}{\barh} &\leq \left(\delta_2 + \frac{2f}{3}\cdot(1+\delta_3)\right)\ssigma f^2 + \mu + 6 \\
    &\leq \left(7/6 + \delta_2 f + \delta_3 \cdot 2f^2/3\right) \mu + 6 \, .
  \end{align*}
  Finally,
  \begin{align*}
    \delta_2 f + \delta_3 \cdot 2f^2/3 &= f \sqrt{6 \ln(\ssec/p) /(f^2 \ssigma)} + f^2\sqrt{18 \ln(\ssec/p) /(f^3 \ssigma)} \cdot 2/3 \\
    &= \sqrt{6 \ln(\ssec/p)/\ssigma} + \sqrt{18f \ln(\ssec/p) / \ssigma} \cdot 2/3 \\
    &= \sqrt{6 \ln(\ssec/p)/\ssigma} + \sqrt{8f \ln(\ssec/p) / \ssigma} \\
    &\leq (2+\sqrt{6}) \cdot \sqrt{\ln(\ssec/p) / \ssigma}
  \end{align*}
  and thus \(\sum_\alpha \Varcond{\hat X_\alpha}{\barh}\) will be roughly \(\mu \cdot 7/6\) whenever \(\mu\) is large and \(\ln(\ssec/p) \ll \mu\).
\end{proof}

With the bound on \(\sum \Varcond{\hat X_\alpha}{\barh}\) in place, we can prove \cref{lem:3layers_2}.
\largethree*
\begin{proof}
  Let \(\varepsilon = (2+\sqrt{6})\sqrt{\ln(\ssec/p)/\ssigma} + 6/\mu\) as defined in \cref{lem:conditionalVariance3}.
  As \(\size{\hat X_\alpha} \leq 3\), Bernstein's (\cref{eqn:Bernstein}) takes the following form when conditioning on \(\sum_{\alpha \in \Sigma} \Varcond{\hat X_\alpha}{\barh} \leq \mu' = \left(7/6 + \varepsilon\right) \mu\),
  \[\Prpcond{S_{\leq 3} \leq \Epcond{S_{\leq 3}}{\barh} - t}{\barh, \, \sum_{\alpha \in \Sigma} \Varcond{\hat X_\alpha}{\barh} \leq \mu'} \leq \exp\left(\frac{-0.5t^2}{\mu' + t} \right) \, .\]
  Solving for \(t\) we find
  %% Solve: https://www.wolframalpha.com/input?i2d=true&i=Divide%5B0.5Power%5Bx%2C2%5D%2Ca%2Bx%5D%3Dc
  \begin{align*}
    t &\geq \sqrt{2\ln(1/p) \left(\mu' + \frac{1}{2} \ln(1/p) \right) } + \ln(1/p)\\
    &\qquad\implies \exp\left(\frac{-0.5t^2}{\mu' + t} \right) \leq p \, .
  \end{align*}

  As \(\sqrt{\frac{7}{3} \ln(1/p) (1+\varepsilon_3) \mu} + \ln(1/p)
  \geq \sqrt{2\ln(1/p) \left(\mu' + \frac{1}{2} \ln(1/p) \right) } + \ln(1/p)\),
  the \lcnamecref{lem:3layers_2} follows when we add the probability that \(\sum_\alpha \Varcond{\hat X_\alpha}{\barh} > \mu'\).
\end{proof}

We prove \cref{lem:2layers_2} in the same way, with \(\bar X_\alpha = \min\set{2, \size{X_\alpha}}\), \(S_{\leq 2} = S_1 + S_2 = \sum_\alpha \bar X_\alpha\) and the following bound on \(\sum \Varcond{\bar X_\alpha}{\barh}\).

\begin{lemma} \label{lem:conditionalVariance2}
  If \(\ln(\ssec/p) \leq \ssigma f^2/2 = \bar \mu_2\),
  \[\Prp{ \sum_{\alpha \in \Sigma} \Varcond{\bar X_\alpha}{\barh} >  \left(1 + \varepsilon\right) \mu} < p \cdot 2/\ssec + 2 \PnegJ \]
  where \(\varepsilon = \sqrt{6 \ln(\ssec/p)/\ssigma} + 2/\mu\).
\end{lemma}
\begin{proof}
  The proof proceeds in the same way as that of \cref{lem:conditionalVariance3}.
  Define \(\bar \mu = \sum_{\alpha} \Epcond{\bar X_\alpha}{\barh} = \Epcond{S_{\leq 2}}{\barh}\) and \(\bar f = \bar \mu / \ssigma\).
  Then
  \begin{align*}
    \sum_{\alpha \in \Sigma} \Varcond{\bar X_\alpha}{\barh}
    &\leq \sum_{\alpha\in\Sigma} \Epcond{\bar X_\alpha(\bar X_\alpha - 1)}{\barh} + \sum_{\alpha\in\Sigma} \Epcond{\bar X_\alpha}{\barh}  - \ssigma \bar f^2 \\
    &= 2 \cdot \Epcond{S_2}{\barh} + \bar \mu - \ssigma \bar f^2 \, .
  \end{align*}
  By \cref{conditionalTranslationFix,lem:layer-size-chernoff}, for \(\delta = \sqrt{3 \ln(\ssec/p)/\bar \mu_2} = \sqrt{6 \ln(\ssec/p)/(f^2 \ssigma)}\) and \(\mu^+ = (1+\delta)\ssigma f^2/2 + 1\),
  \[
  \Prp{\Epcond{S_2}{\barh} > \mu^+ \land \calJ} \leq p \cdot 2/\ssec + 2 \PnegJ \, .
  \]
  
  Hence, with probability at least \(1- p \cdot 2/\ssec - 2 \PnegJ\) we have
  \begin{align*}
    \sum_{\alpha \in \Sigma} \Varcond{\bar X_\alpha}{\barh} &< 2\mu^+ + \bar \mu - \ssigma \bar f^2 \\
    %% &\leq 2\cdot \Epcond{S_2}{\barh} + \bar \mu - \ssigma \bar f^2 \\
    &= (1+\delta) \ssigma f^2 + \bar \mu - \ssigma \bar f^2 + 2 \\
    &\leq (1 + \delta f) \mu + 2 \, .
    %% &\leq 2\left((1+\delta)\ssigma f^2/2 + \sqrt{(1+\delta) \ssigma f^2/2}\right) + \bar\mu - \ssigma \bar f^2
  \end{align*}

  Finally,
  \begin{align*}
    \delta f &= f \sqrt{6 \ln(\ssec/p) /(f^2 \ssigma)} \\
    &= \sqrt{6 \ln(\ssec/p)/\ssigma}
  \end{align*}
  and the \lcnamecref{lem:conditionalVariance2} follows.
\end{proof}
\largetwo*
\begin{proof}
  Let \(\varepsilon = \sqrt{6\ln(\ssec/p)/\ssigma} + 2/\mu\) as defined in \cref{lem:conditionalVariance2}.
  As \(\size{\bar X_\alpha} \leq 2\), Bernstein's (\cref{eqn:Bernstein}) takes the following form when conditioning on \(\sum_{\alpha \in \Sigma} \Varcond{\bar X_\alpha}{\barh} \leq \mu' = \left(1 + \varepsilon\right) \mu\),
  \[\Prpcond{S_{\leq 2} \leq \Epcond{S_{\leq 2}}{\barh} - t}{\barh, \, \sum_{\alpha \in \Sigma} \Varcond{\bar X_\alpha}{\barh} \leq \mu'} \leq \exp\left(\frac{-0.5t^2}{\mu' + t \cdot 2/3} \right) \, .\]
  Solving for \(t\) we find
  %% Solve: https://www.wolframalpha.com/input?i2d=true&i=Divide%5B0.5Power%5Bx%2C2%5D%2Ca%2BDivide%5B2x%2C3%5D%5D%3Dc
  \begin{align*}
    t &\geq \sqrt{2\ln(1/p) \left(\mu' + \frac{2}{9} \ln(1/p) \right) } + \frac{2}{3} \ln(1/p) \\
    &\qquad\implies \exp\left(\frac{-0.5t^2}{\mu' + t \cdot 2/3} \right) \leq p \, .
  \end{align*}
  The \lcnamecref{lem:2layers_2} follows when we add the probability that \(\sum_\alpha \Varcond{\bar X_\alpha}{\barh} > \mu'\).
\end{proof}

% !TeX root = main.tex
\subsection{Regular Layers: Proof of \texorpdfstring{\Cref{lem:regularLayers_2}}{Theorem \ref{lem:regularLayers_2}}}
\label{proof:regular}
In order to prove \cref{lem:regularLayers_2} we first need the following lemmas bounding the difference between \(S_i\) and its conditional expectation.
\begin{lemma}\label{lem:regularCondExp}
  Let \(\delta = \sqrt{3 \ln(1/p) / \bar \mu_i}\).
  If \(\ln(1/p) \leq \bar \mu_i\) then
  \[\Prp{\Epcond{S_i}{\barh} \geq (1+\delta) \bar \mu_i + 1 } \leq 2 p + 2 \PnegJ \, .\]
\end{lemma}
\begin{proof}
  By the assumption of the theorem \(\delta \leq \sqrt{3}\) and thus \(\left(e^\delta/(1+\delta)^{(1+\delta)}\right) \leq \exp(-\delta^2/3)\).
  It follows from \cref{lem:layer-size-chernoff} that,
  \[\Prp{S_i \geq (1+\delta) \bar \mu_i \land \calJ} \leq p \, .\]
  By \cref{conditionalTranslationFix}
  \[\Prp{\Epcond{S_i}{\barh} \geq (1+\delta) \bar \mu_i + 1 } \leq 2 p + 2 \PnegJ \, .\]
\end{proof}
\begin{lemma} \label{lem:regularLayer_2}
  Assume that \(\ln(\ssec/p_i) \leq \bar \mu_i\). Then
  \[\Prp{S_i \leq \Epcond{S_i}{\barh} - \sqrt{2 \ln(1/p_i) (\bar \mu_i \cdot(1+\varepsilon_i) + 1)}} < (1 + 2/\ssec) \cdot p_i + 2 \PnegJ \]
  where \(\varepsilon_i = \sqrt{3 \ln(\ssec/p_i)/\bar \mu_i}\).
\end{lemma}
\begin{proof}
  As \(\ln(\ssec/p_i) \leq \bar \mu_i\) we can apply \cref{lem:regularCondExp} with \(p = p_i/\ssec\) and get that
  \[\Prp{\Epcond{S_i}{\barh} \geq (1+\varepsilon_i)\bar \mu_i + 1} < 2/\ssec \cdot p_i + 2 \PnegJ \, .\]
  Let \(\mu^+ = (1+\varepsilon_i) \bar \mu_i + 1\) and \(\delta' = \sqrt{2 \ln(1/p_i) / \mu^+}\).
  When conditioned on \(\barh\), \(S_i\) is a sum of independent 0/1-variables (with unknown, non-identical distributions).
  Conditioning on \(\Epcond{S_i}{\barh} < \mu^+\), we thus have
  \begin{align*}
    \Prpcond{S_i < \Epcond{S_i}{\barh} - \delta' \mu^+}{\barh, \, \Epcond{S_i}{\barh} < \mu^+} < \exp\left(-\mu^+ \delta'^2 / 2\right) = p_i \, .
  \end{align*}
  As this bound holds for all realizations of \(\barh\) where \(\Epcond{S_i}{\barh} < \mu^+\) the bound also holds without conditioning on \(\barh\):
  \begin{align*}
    \Prpcond{S_i < \Epcond{S_i}{\barh} - \delta' \mu^+}{\Epcond{S_i}{\barh} < \mu^+} < \exp\left(-\mu^+ \delta'^2 / 2\right) = p_i \, .
  \end{align*}
  
  Combining the pieces,
  \begin{align*}
    \Prp{S_i \leq \Epcond{S_i}{\barh} - \delta'\mu^+} &\leq \Prp{\left(S_i \leq \Epcond{S_i}{\barh} - \delta'\mu^+ \right) \land \left(\Epcond{S_i}{\barh} < \mu^+\right)} + \Prp{\Epcond{S_i}{\barh} \geq \mu^+} \\
    &\leq \Prpcond{S_i \leq \Epcond{S_i}{\barh} - \delta'\mu^+}{\Epcond{S_i}{\barh} < \mu^+} + 2/\ssec \cdot p_i + 2 \PnegJ \\
    &\leq (1 + 2/\ssec) \cdot p_i + 2\PnegJ \, .
  \end{align*}
  The lemma follows as \(\delta' \mu^+ = \sqrt{2 \ln(1/p_i) \cdot ((1+\varepsilon_i)\bar \mu_i + 1)}\).
\end{proof}

We are now ready to prove \cref{lem:regularLayers_2}.
\regular*
\begin{proof}
  We define \(p_4 = p\) and \(p_{i+1} = \max\set{p_i/e, \preg}\) as discussed in \cref{layers-ingredients}.
  Let \(\varepsilon_i = \sqrt{3 \ln(\ssec/p_i)/\bar\mu_i}\) and \(\Delta_i = \sqrt{2\ln(1/p_i)\cdot((1+\varepsilon_i) \bar \mu_i + 1)}\).
  By \cref{lem:regularLayer_2}, for all \(i \in \set{4, \dots, \inr-1}\),
  \[\Prp{S_i < \Epcond{S_i}{\barh} - \Delta_i} < (1 + 2/\ssec) \cdot p_i + 2 \PnegJ \, .\]
  The \lcnamecref{lem:regularLayers_2} follows when we have shown that \(\sum_{i=4}^{\inr-1} \Delta_i \leq \Delta_{reg}\) and \(\sum_{i=4}^{\inr-1} p_i \leq (1.59 + 1/\sall) \cdot p \).
  We start with the latter, and recall that \(\preg \leq \frac{p}{\sall \cdot (\inr -4)}\) as discussed in \cref{layers-ingredients}.
  Hence
  \begin{align*}
    \sum_{i=4}^{\inr-1} p_i \leq (\inr-4)\cdot \preg + \sum_{k=0}^\infty \frac{p}{e^k} \leq \left(\frac{1}{\sall} + 1.59\right) \cdot p \, .
  \end{align*}

  To bound the sum of \(\Delta_i\)'s we distinguish between three cases:

  First, assume \(\inr = 5\). We thus have to show that \(\Delta_4 \leq \Delta_{reg}\).
  As \(\varepsilon_i \leq \sqrt{3}\),
  \begin{align*}
    \Delta_4 \leq \sqrt{2 \ln(1/p) \cdot ((1+\sqrt{3}) \bar \mu_4 + 1)}
    &\leq 0.169 \sqrt{\ln(1/p) \mu} + \frac{\sqrt{2\ln(1/p)}}{2 \sqrt{\bar \mu_4}} \leq 0.169 \sqrt{\ln(1/p) \mu} + \frac{\sqrt{2}}{2} \, .
  \end{align*}

  Second, assume \(\inr = 6\).
  Then \(\ln(1/p_5) \leq \bar \mu_5 = \bar \mu_4 \cdot f/5\).
  As \(\ln(1/p_4) \leq \ln(1/p_5)\) we then have \(\varepsilon_4 = \sqrt{3 \ln(1/p_4) / \bar \mu_4} \leq \sqrt{3/10}\).

  \begin{align*}
    \Delta_4+\Delta_5 &\leq \sqrt{2\ln(1/p_4)\cdot((1+\sqrt{3/10})\bar\mu_4 + 1)} + \sqrt{2\ln(1/p_5)\cdot((1+\sqrt{3})\bar\mu_5 + 1)} \\
    &\leq 0.127 \sqrt{\ln(1/p_4) \mu} + \frac{\sqrt{2\ln(1/p_4)}}{2 \sqrt{\bar \mu_4}} + 0.054 \sqrt{\ln(1/p_5) \mu} + \frac{\sqrt{2\ln(1/p_5)}}{2 \sqrt{\bar \mu_5}} \\
    &\leq 0.181 \sqrt{\ln(1/p) \mu} + \frac{0.054 \sqrt{\mu}}{\sqrt{\ln(1/p)}} + \sqrt{2} \, .
  \end{align*}

  Finally, assume \(\inr \geq 7\).
  Then \(\ln(1/p_6) \leq \bar \mu_6 = \mu_5 \cdot f/6 = \mu_4 \cdot f^2/30\) and thus \(\varepsilon_4 \leq \sqrt{3/120}\) while \(\varepsilon_5 \leq \sqrt{3/12}\).
  First we bound \(\Delta_4+\Delta_5\), in the same way as in the previous case:
  \begin{align*}
    \Delta_4+\Delta_5 &\leq \sqrt{2\ln(1/p_4)\cdot((1+\sqrt{3/120})\bar\mu_4 + 1)} + \sqrt{2\ln(1/p_5)\cdot((1+\sqrt{3/12})\bar\mu_5 + 1)} \\
    &\leq 0.110 \sqrt{\ln(1/p_4) \mu} + \frac{\sqrt{2\ln(1/p_4)}}{2 \sqrt{\bar \mu_4}} + 0.041 \sqrt{\ln(1/p_5) \mu} + \frac{\sqrt{2\ln(1/p_5)}}{2 \sqrt{\bar \mu_5}} \\
    &\leq 0.151 \sqrt{\ln(1/p) \mu} + \frac{0.041 \sqrt{\mu}}{\sqrt{\ln(1/p)}} + \frac{\sqrt{2}}{2} \sum_{i=4}^5 \sqrt{\frac{\ln(1/p_i)}{\bar\mu_i}} \, .
  \end{align*}
  For \(i = 6, \dots, \inr-1\) we stick with the simple bound \(\varepsilon_i \leq \sqrt{3}\).
  \begin{align*}
    \sum_{i=6}^{\inr - 1} \Delta_i &\leq \sum_{i=6}^{\inr-1} \sqrt{2 \ln(1/p_i) \cdot ((1+\sqrt{3}) \bar \mu_i + 1)} \\
    &\leq \sum_{i=6}^{\inr-1} \sqrt{2 \ln(1/p) \cdot (1+\sqrt{3}) \bar \mu_i} + \frac{(i-4) \sqrt{2(1+\sqrt{3}) \bar\mu_i}}{2\sqrt{\ln(1/p)}} + \frac{\sqrt{2\ln(1/p_i)}}{2\sqrt{\bar\mu_i}}\\
    &\leq 0.0209 \sqrt{\ln(1/p) \mu} + 0.0245 \sqrt{\frac{\mu}{\ln(1/p)}} + \frac{\sqrt{2}}{2}\sum_{i=6}^{\inr-1} \sqrt{\frac{\ln(1/p_i)}{\bar\mu_i}} \\
    %% FIRST:
    %% https://www.wolframalpha.com/input?i2d=true&i=sqrt%5C%2840%292*%5C%2840%291%2Bsqrt%5C%2840%293%5C%2841%29%5C%2841%29%5C%2841%29*Sum%5BDivide%5B1%2Csqrt%5C%2840%29Power%5B2%2Ci-1%5Di%21%5C%2841%29%5D%2C%7Bi%2C6%2Cinfty%7D%5D
    %% SECOND:
    %% https://www.wolframalpha.com/input?i2d=true&i=sqrt%5C%2840%291%2Bsqrt%5C%2840%293%5C%2841%29%5C%2841%29*%5C%2840%29Sum%5BDivide%5Bi%2Csqrt%5C%2840%29i%21Power%5B2%2Ci%5D%5C%2841%29%5D%2C%7Bi%2C6%2Cinfty%7D%5D-Sum%5BDivide%5B4%2Csqrt%5C%2840%29i%21Power%5B2%2Ci%5D%5C%2841%29%5D%2C%7Bi%2C6%2Cinfty%7D%5D%5C%2841%29
    \intertext{hence}
    \sum_{i=4}^{\inr-1} \Delta_i &\leq 0.172 \sqrt{\ln(1/p)\mu} + 0.066 \sqrt{\frac{\mu}{\ln(1/p)}} + \frac{\sqrt{2}}{2}\sum_{i=4}^{\inr-1} \sqrt{\frac{\ln(1/p_i)}{\bar\mu_i}} \, .
  \end{align*}

  Using that \(\bar \mu_{\inr-1} \geq \ln(1/p_{\inr-1})\) we thus have \(\bar \mu_{\inr - 1 - k} \geq \bar \mu_{\inr-1} \cdot 10^k \geq \ln(1/p_{\inr-1}) \cdot 10^k\) when \(k \leq \inr - 5\), and we can bound the final sum:
  \begin{align*}
    \frac{\sqrt{2}}{2} \sum_{i=4}^{\inr - 1} \sqrt{\frac{\ln(1/p_i)}{\bar \mu_i}}
    &\leq \frac{\sqrt{2}}{2} \sum_{k=0}^{\inr - 5} \frac{1}{\sqrt{10^k}}
    \leq \sqrt{2} \, .
  \end{align*}
  Thus we have shown that \(\sum_i \Delta_i \leq \Delta_{reg}\) in all three cases, which together cover all outcomes.
\end{proof}

% !TeX root = main.tex
\subsection{Non-Regular Layers: Proof of \texorpdfstring{\Cref{thm:nonRegularLayers}}{Theorem \ref{thm:nonRegularLayers}}}
\label{proof:nonRegularLayers}
\Cref{thm:nonRegularLayers} covers the layers from \(\inr\) and up.
Define \(\iinf\) to be the first integer \(i\) such that \(\bar \mu_i \leq p/\sall\).
Then \cref{lem:layeri0_2,lem:topLayers_2,lem:infiniteLayers-2} below cover all of the non-regular layers, and \cref{thm:nonRegularLayers} is obtained through a union bound over the three lemmas.

Before proving the lemmas, we need the following definition and bound:
\begin{definition}[\(W\)]
The Lambert \(W\) function is the function that solves the equation \[W(x) \cdot \exp(W(x)) = x \, .\]
\end{definition}
\begin{lemma}[Theorem 2.3 of \cite{hoorfar2008inequalities}] \label{lem:W-bound}
  For \(x > 1/e\),
    \[W(x) \leq \ln\left(\frac{2x}{\ln(x) +1}\right) \, .\]
\end{lemma}

\begin{lemma} \label{lem:layeri0_2}
  \[\Prp{\Epcond{S_{\inr}}{\barh} > \Delta_{\inr}} \leq p \cdot 2/\sall + 2 \PnegJ \, .\]
\end{lemma}
\begin{proof}
  Let \(p_i = \preg/\ssec \leq p/\sall\) such that \(\ln(1/p_i) \geq \bar \mu_{\inr}\)
  and define \(\delta = 1.33 e \ln(1/p)/\bar \mu_{i} - 1 > 1.33 e-1\) such that \((1+\delta) \bar \mu_{i} = 1.33 e \ln(1/p_i)\).
  As \(1.33 > e^{W(1/e)}\),
  \[1.33 e \ln(1/p_i) > \ln(1/p_i)/\ln(1.33) = \log_{1.33}(1/p_i) \, .\]
  Then, by \cref{lem:layer-size-chernoff},
  \begin{align*}
    \Prp{S_{\inr} \geq 1.33e \ln(1/p_i) \land \calJ} &\leq \left(\frac{e^\delta}{(1+\delta)^{(1+\delta)}}\right)^{\bar \mu_{i}} \\
    &\leq \left(\frac{e}{1+\delta}\right)^{(1+\delta)\bar\mu_{i}} \\
    &\leq \left(\frac{e}{1.33e}\right)^{1.33e\ln(1/p_i)} \\
    &< \left(\frac{1}{1.33}\right)^{\log_{1.33}(1/p_i)} = p_i \, .
  \end{align*}
  By \cref{conditionalTranslationFix}, \(\Prp{\Epcond{S_{\inr}}{\barh} > \Delta_{\inr}} \leq 2 p_i + 2 \PnegJ \leq p \cdot 2/\sall + 2\PnegJ\).
\end{proof}

\begin{lemma}\label{lem:infiniteLayers-2}
  \[\Prp{\sum_{i=\iinf}^{\imax} \Epcond{S_i}{\barh} > 3} < p \cdot 2/\sall + (\imax-\iinf) \cdot 2 \PnegJ\, . \]
\end{lemma}
\begin{proof}
  Let \(i^+ = i - \iinf\) and set \(\eps_i = 2/3^{i^+}\).
  As \(\sum_{\iinf}^{\imax} \eps_i < \sum_{i=\iinf}^\infty \eps_i = 3\),
  \[\Prp{\sum_{i=\iinf}^{\imax} \Epcond{S_i}{\barh} \geq 3} \leq \sum_{i=\iinf}^{\imax} \Prp{\Epcond{S_i}{\barh} \geq \eps_i} \, .\]

  By \cref{conditionalTranslationFix} and Markov's inequality
  \begin{align*}
  \Prp{\Epcond{S_i}{\barh} \geq \eps_i} \leq 2\Prp{S_i \geq \eps_i \land \calJ} + 2\PnegJ \leq \frac{2\Ep{S_i \cdot \indicator{\calJ}}}{\eps_i} + 2\PnegJ \leq \frac{2\bar\mu_i}{\eps_i} + 2\PnegJ \, .
  \end{align*}
  Using that \(\bar\mu_i \leq \bar \mu_{\iinf}/6^{i^+}\) we thus have
  \begin{align*}
    \sum_{i=\iinf}^{\imax} \Prp{\Epcond{S_i}{\barh} \geq \frac{2}{3^{i^+}} } &\leq \bar\mu_{\iinf} \cdot \sum_{k=0}^{\imax} \left(\frac{3}{6}\right)^k + (\imax-\iinf) \cdot 2\PnegJ \\
    &\leq p \cdot 2/\sall + (\imax-\iinf) \cdot 2 \PnegJ\, .
  \end{align*}
\end{proof}

\begin{lemma} \label{lem:topLayers_2}
  \[\Prp{ \sum_{i=\max\set{\inr+1.\, 4}}^{\iinf} \Epcond{S_i}{\barh} \geq \Delta_{top} } \leq p \cdot 2/\sall + (\iinf - \inr) \cdot 2 \PnegJ \, .\]
\end{lemma}
\Cref{lem:topLayers_2} is obtained by applying the following \lcnamecref{lem:topLayer} on each layer between \(\inr\) and \(\iinf\).

\begin{lemma}\label{lem:topLayer}
  For \(i > \inr\) and \(p_i \leq \preg/\ssec\),
  \[\Prp{S_i \geq \frac{2\ln(1/p_i)}{(i - \inr) \cdot \ln(i/ef)} \land \calJ} \leq p_i \, .\]
\end{lemma}
\begin{proof}[Proof of \cref{lem:topLayer}]
  Define \(i^+ = i - \inr\) and let \(k = (1+\delta) \bar \mu_i\) for some \(\delta > 0\).
  By \cref{lem:layer-size-chernoff}
  \[\Prp{S_i > k \land \calJ} \leq \left(\frac{e^\delta}{(1+\delta)^{(1+\delta)}}\right)^{\bar\mu_i} < \left(\frac{e}{(1+\delta)}\right)^{(1+\delta)\bar\mu_i} = \left(\frac{e\bar\mu_i}{k}\right)^k \, .\]
  As \(\bar \mu_i = \bar \mu_{\inr} \cdot\frac{ f^{i^+} \inr!}{i!}\)
  %% = \bar \mu_{i_{reg}} \cdot \frac{f^{i^+}}{i^{\ul{i^+}}}\)
  and \(\bar \mu_{\inr} < \ln(\ssec/\preg) \leq \ln(1/p_i)\) we have
  \begin{align*}
    \Prp{S_i > k \land \calJ} \leq \left(\frac{e\bar\mu_i}{k}\right)^k
    %% = \left(\frac{e\cdot \bar\mu_{i_0} i_0! \cdot f^{i^+}}{k \cdot i!}\right)^k
    \leq \left(\frac{e \ln(1/p_i)}{k} \cdot \frac{f^{i^+} \inr!}{i!}\right)^k \, .
    %% \leq \left(\frac{e \ln(1/p_i)}{k} \cdot \frac{f^{i^+}}{i^{\ul{i^+}}}\right)^k \, .
  \end{align*}
  By Stirling's approximation,
  \[\frac{\inr!}{i!} \leq e^{i^+} \cdot \frac{(\inr)^{\inr}}{i^i} = e^{i^+} \cdot \left(\frac{\inr}{i}\right)^{\inr} \cdot \frac{1}{i^{i^+}} \leq \left(\frac{e}{i}\right)^{i^+}\]
  %% \[\frac{1}{i^{\ul{i^+}}} \leq e^{i^+} \cdot \frac{(i_{reg})^{i_{reg}}}{i^i} = e^{i^+} \cdot \left(\frac{i_{reg}}{i}\right)^{i_{reg}} \cdot \frac{1}{i^{i^+}} \leq \left(\frac{e}{i}\right)^{i^+}\]
  and thus
  \begin{align*}
    \Prp{S_i \geq k \land \calJ} \leq \left( \frac{e \ln(1/p_i)}{k} \cdot \left(\frac{ef}{i}\right)^{i^+} \right)^k \, .
  \end{align*}

  Define \(k = \frac{2 \ln(1/p_i)}{i^+ \ln(i/ef)}\) and
  \begin{align*}
    \eta &= \frac{2}{e} \cdot \frac{(i/ef)^{i^+}}{i^+ \ln(i/ef)}
    = \frac{2 (i/ef)^{i^+}/e}{\ln\left((i/ef)^{i^+}\right)}
    = \frac{2 (i/ef)^{i^+}/e}{\ln\left((i/ef)^{i^+}/e\right)+1}
    \geq \exp\left(W \left((i/ef)^{i^+}/e \right)  \right) \, ,
  \end{align*}
  with the final inequality due to \cref{lem:W-bound}, as \(i/ef \geq 1\) when \(i \geq 2\).
  Hence \(\eta \ln(\eta) \geq (i/ef)^{i^+}/e\) or, equivalently, \(\frac{\eta \cdot e}{(i/ef)^{i^+}} \geq 1/\ln(\eta)\).
  Observe that \(k = \ln(1/p_i) \cdot \frac{\eta \cdot e}{(i/ef)^{i^+}}\) and thus
  \begin{align*}
    \Prp{S_i \geq k \land \calJ} &\leq \left( \frac{e \ln(1/p_i)}{k} \cdot \left(\frac{ef}{i}\right)^{i^+} \right)^k \\
    &\leq \left( \frac{e}{2} \cdot i^+ \ln(i/ef) \cdot \left(\frac{ef}{i}\right)^{i^+}\right)^k \\
    &\leq \left(\frac{1}{\eta}\right)^{\frac{\ln(1/p_i)}{\ln(\eta)}} = \left(\frac{1}{\eta}\right)^{\log_\eta(1/p_i)} = p_i \, .
  \end{align*}
\end{proof}

We can now prove \cref{lem:topLayers_2}:
\begin{proof}[Proof of \cref{lem:topLayers_2}]
  First, note that \(\iinf - \inr \leq \ntop = \log_{6}(\ln(\ssec/\preg) \cdot \sall/p)\).
  This bound comes from the fact that \(\bar \mu_{\inr} \leq \ln(\ssec/\preg)\) while \(\bar \mu_{\iinf -1} \geq p/\sall\) and \(\bar \mu_{i+1} \leq \bar \mu_i / 6\) for all \(i \geq 2\).

  For the \(i\)'th layer let \(k_i = \ln(1/\ptop) \cdot \frac{2}{i^+ \ln(i/ef)} + 1\) where \(i^+ = i - \inr\).
  By \cref{lem:topLayer,conditionalTranslationFix}, \(\Prp{\Epcond{S_i}{\barh} > k_i} < 2 \ptop + 2\PnegJ\).
  As \(\ptop \leq p/(\ntop \cdot \sall)\), we have
  \[\sum_{\max\set{\inr+1, \, 4}}^{\iinf} \Prp{\Epcond{S_i}{\barh} > k_i} \leq p \cdot 2/\sall + (\iinf - \inr) \cdot 2\PnegJ\,. \]

  Left is to show that \(\sum k_i \leq \Delta_{top}\).
  As \(i \geq 4\) we have \(\ln(i/ef) > 0\) and the \(k_i\)'s are decreasing.
  Their sum can thus be bounded by a definite integral:
  \begin{align*}
    \sum_{i= \max\set{\inr+1,\, 4}}^{\iinf} \frac{2}{i^+ \ln(i/ef)} &\leq \sum_{j=1}^{\iinf -\inr} \frac{2}{j \ln((j+3)/ef)} \\
    &\leq \sum_{j=1}^4 \frac{2}{j\ln(4/ef)} + \sum_{j=5}^{\iinf - \inr} \frac{2}{j \ln(j/ef)} \\
    &\leq \frac{25/6}{\ln(4/ef)} + \int_4^{\iinf - \inr} \frac{2}{x \ln(x/ef)} \, \mathrm{d}x \\
    &= \frac{25/6}{\ln(4/ef)} + \int_{4/ef}^{(\inf-\inr)/ef} \frac{2}{u \ln(u)} \, \mathrm{d}u \\
    &= \frac{25/6}{\ln(4/ef)} + 2\left(\ln \ln ((\iinf-\inr)/ef) - \ln \ln (4/ef)\right) \\
    &\leq 3.9 + 2 \ln \ln (\ntop) \, .
  \end{align*}
  Hence
  \begin{align*}
    \sum_{i=\max\set{\inr + 1, 4}}^{\iinf} k_i \leq 2 \ln(1/\ptop) \cdot \left(2 + \ln \ln(\ntop) \right) + \ntop \, .
  \end{align*}
\end{proof}

Finally, we are now ready to prove \cref{thm:nonRegularLayers}
\nonRegular*
\begin{proof}
  Each \(S_i\) is non-negative, and thus \(\Prp{S_i < \Epcond{S_i}{\barh} - k} \leq \Prp{\Epcond{S_i}{\barh} > k}\) for any \(k\).
  By \cref{lem:layeri0_2,lem:topLayers_2,lem:infiniteLayers-2}, we thus have
  \begin{align*}
    \Prp{\sum_{i=\inr}^{\imax} S_i < \sum_{i=\inr}^{\imax} \Epcond{S_i}{\barh} - \Delta_{nonreg}}
    &\leq \Prp{\sum_{i=\inr}^{\imax} \Epcond{S_i}{\barh} > \Delta_{nonreg}} \\
    &< p \cdot 6/\sall + (\imax - \inr) \cdot 2\PnegJ 
  \end{align*}
  as \(\Delta_{nonreg} = \Delta_{\inr} + \Delta_{top} + 3\).
\end{proof}

\subsection{No Big Layers}\label{sec:no-big-layer}
\label{sec:imax}
We finally prove~\cref{lem:imax} which handles layers beyond $i_{\text{max}}$. For this, we require our new~\cref{thm:new-tech-random-set} for bounding the probability that the selected set is linearly dependent. We postpone the proof of this theorem to~\cref{sec:3c-linear-independent}.
\imaxlemma*
\begin{proof}
  Note that \(\set{S_i}_{i \in \bbN}\) is a non-increasing sequence.
  Hence \(\set{\Epcond{S_i}{\barh}}_{i\in \bbN}\) is non-increasing, and it suffices to prove that \(\Prp{\Epcond{S_i}{\barh} > 0}\) is small for \(i = \imax + 1\).

  Let \(\select : \Sigma^c \times [2]^{\selBits} \to [2]\) be the selector function defining the set \(X\) (referred to as '\(f\)' in \cref{sec:prelim}).
  For each character \(\alpha \in \Sigma\) and value \(r \in [2]^\selBits\) of the selection bits, define
  \[\mX = \set{x \in S \cond \select(x, \barh[0](x) \xor r) = 1 \land \barh[1](x) = \alpha} \, .\]
  That is, \(\mX\) is the set of keys with final derived character \(\alpha\) which will be selected if \(T_{c+d}[\alpha] = r\).
  Note that \(\mX\) is completely determined by \(\barh\). When conditioning on \(\barh\), \(X_\alpha\) is thus uniformly distributed among the values \(\set{\mX}_{r \in [2]^\selBits}\).
  We will show that, w.h.p., \(\size{\mX} \leq \imax\) for all \(r\).
  This entails that \(\size{X_{\alpha}} \leq \imax\), in turn giving that \(\Epcond{S_i}{\barh} = 0\).

  Recall that, for each element \(x \in A\), there exists exactly one value \(r'\) such that \(\select(x, r') = 1\).
  Thus the sets \(\mX\) partition \(A\), and each \(x\) is distributed uniformly among the sets.
  As the expected number of selected elements is \(\size{A} / 2^\selBits = \Ep{\size{X}} = f\ssigma\),
  the expected size of each set will be \(\Ep{\size{\mX}} = \size{A} / (\ssigma 2^\selBits) = f \leq 1/2\).

  As \(\barh\) is a tornado tabulation function (albeit with only \(d-1\) derived characters), we can invoke \cref{thm:Chernoff-upper} to bound the probability of \(\mX\) being large.
  Setting \((1+\delta) = 1/f \cdot \ln(\ssigma^{d-2} 2^\selBits) > e^2\) we obtain
  \begin{align*}
    \Prp{\size{\mX} \geq \ln\left(\ssigma^{d-2} 2^\selBits\right) \land \eventI{\mX}}
    &< \parentheses{\frac{e}{1+\delta}}^{(1+\delta)f} \\
    &= \parentheses{\frac{1}{e}}^{\ln(\ssigma^{d-2} 2^\selBits)} \\
    &= \frac{1}{\ssigma^{d-2} 2^\selBits} \, .
  \end{align*}

  By \cref{thm:new-tech-random-set} we further have, plugging in expected size \(\mu/\ssigma = f\), ratio '\(f\)' of \(1/(2\ssigma)\), and \(d-1\) derived characters,
  \[\Prp{\neg \eventI{\mX}} \leq \frac{3^c \ssigma}{n} \cdot 3 \parentheses{\frac{\mu}{\ssigma}}^3 \parentheses{\frac{3}{\ssigma}}^d + \parentheses{\frac{1}{2\ssigma}}^{\ssigma/2} \, .\]
  Adding the two error probabilites, and performing a union bound over all \(\ssigma \cdot 2^\selBits\) sets \(\mX\), the statement follows:
  \begin{align*}
    \Prp{\exists \mX : \size{\mX} \geq \ln\parentheses{\ssigma^{d-2} 2^\selBits}}
    &\leq \ssigma 2^\selBits \cdot \Prp{\size{\mX} \geq \ln\parentheses{\ssigma^{d-2} 2^\selBits}} \\
    &\leq \ssigma 2^\selBits \parentheses{\frac{1}{\ssigma^{d-2} 2^\selBits} + \frac{3^c \ssigma}{n} \cdot 3 \parentheses{\frac{\mu}{\ssigma}}^3 \parentheses{\frac{3}{\ssigma}}^d + \parentheses{\frac{1}{2\ssigma}}^{\ssigma/2}} \\
    &\leq \frac{1}{\ssigma^{d-3}} + \frac{3^c \ssigma^2}{\mu} \cdot 3 \parentheses{\frac{\mu}{\ssigma}}^3 \parentheses{\frac{3}{\ssigma}}^d + \ssigma 2^\selBits \parentheses{\frac{1}{2\ssigma}}^{\ssigma/2} \\
    &= \frac{1}{\ssigma^{d-3}} + 9 \cdot 3^c \parentheses{\frac{\mu}{\ssigma}}^2 \parentheses{\frac{3}{\ssigma}}^{d-1} + \ssigma 2^\selBits \parentheses{\frac{1}{2\ssigma}}^{\ssigma/2} \\
    &\leq \frac{1}{\ssigma^{d-3}} + 3^{c+1} \parentheses{\frac{3}{\ssigma}}^{d-1} + \parentheses{\frac{1}{\ssigma}}^{\ssigma/2 - 1} \, ,
  \end{align*}
  where we've used that \(\mu = n/2^\selBits\) and assumed that \(\selBits \leq \ssigma/2\).
\end{proof}

\section{Proof of Theorem \ref{thm:new-tech-random-set}}\label{sec:3c-linear-independent}

In this section, we describe the main ingredients needed for the proof of Theorem \ref{thm:new-tech-random-set} and how they can be combined together. 

\newtechrandomset*

We first note that~\Cref{thm:tech-random-set} holds for a simpler version  of tornado tabulation hashing, which we call \emph{simple tornado hashing}. In this version, we do not change the last  character 
of the (original) key. Formally,  for
a key $x=x_1\cdots x_c$, its corresponding derived key $\derive x = \derive x_1 \ldots \derive x_{c+d}$ is computed as
\begin{equation*}
	\derive x_i = \begin{cases}
		x_i &\text{if $i = 1, \ldots ,c$} \\
		\derive h_{i-c}\ld(\derive x_1 \dots \derive x_{i-1}\rd) &\text{otherwise}.
	\end{cases}
\end{equation*}
The main idea is to revisit the proof of Theorem \ref{thm:tech-random-set} in \cite{BBKHT23}. Note that, if the set of derived keys in $\tilde h(X)$ are linearly dependent, then their prefixes are also linearly dependent (i.e., when we consider only the first $c+d-1$ or $c+d-2$ characters). The main idea is then to argue that, if the derived keys in $\tilde h(X)$ are linearly dependent, then we can find a certain \emph{obstruction} that captures how the keys remain linearly dependent as we add derived characters one at a time. Each such obstruction in unlikely to occur. By performing a union bound over all such possible obstructions, we then get the bound in~\Cref{thm:tech-random-set}. 

\subsection{Preliminaries}\label{sec:prelim-obstructions}

\paragraph{Position Characters, Generalized Keys and Linear Independence} We view any key $x \in \Sigma^b$ as a set of $b$ \emph{position characters} $(1, x_1) \dots (b, x_b)$. We can then define the symmetric difference of two keys as being the symmetric difference of the corresponding sets of position characters. A \emph{generalized key} can be any subset of position characters $\{1 \dots b\} \times \Sigma$. For such a generalized key $x$, we can then  define
\[x[i]=\{(i,a) \in x\}\textnormal,\quad  x[<i]=\{(j,a) \in x\mid j<i\}
\]
and
\[ x[\leq i]=\{(j,a) \in x\mid j\leq i\}.\]

\noindent
This also extends naturally to any set $X$ of generalized keys, i.e., \[X[<i]=\{x[<i]\mid x\in X\}.\]

When it comes to defining linear independence over a set $Y$ of generalized keys, we can define $\bigsd Y$  to be the symmetric difference of all the subsets of position characters, i.e., the set of position characters that appear an odd number of times in the subset in $Y$. If $\bigsd Y$ is the empty set (and hence, every position character appears an even number of times), we say that the set $Y$ is a \emph{zero-set}. If $Y$ contains a zero-set, then we say that $Y$ is \emph{linearly dependent}. Otherwise, we say that the (generalized) keys in $Y$ are \emph{linearly independent}.

\paragraph{Levels and Matching} For the sake of consistency, we follow the setup in \cite{BBKHT23}. The idea is to bound the probability that the keys in $\tilde h(X)$ are dependent with respect to each derived character separately.

To this end, for $i=1,\ldots, d$, we focus on position $c+i$ of a derived key and refer to such positions at being at \emph{level}. Linear dependence in the derived keys means that, for each level, we can pair up derived keys that have the same derived character at that level. Formally, we say that a matching $M \subseteq \binom{\ssigma^c}{2}$ on the keys $\Sigma^c$ is an  \emph{$i$-matching} if for all $\{x,y\}\in M$, 
it holds that $\derive x[c+i]=\derive y[c+i]$. We further say that such a matching is an \emph{$i$-zero, $i$-dependent, or $i$-independent} if the corresponding \[\DiffKeys(M,i)=\left\{(\derive x\sd\; \derive y)[\leq c+i]\mid \{x,y\}\in 
M\right\}\] 
is a zero-set, linearly dependent, or linearly independent, respectively. Similarly, we say that a set $Z$ is of (original) keys is \emph{$i$-zero, $i$-dependent, or $i$-independent} if the set of prefixes  $\derive Z[\leq c+i]$ is a zero-set, linearly dependent, or linearly independent, respectively, where $\derive Z$ denotes the set of derived keys of keys in $Z$. The following observation from \cite{BBKHT23} connects the notions:

\begin{observation}[Observation $11$ in \cite{BBKHT23}]\label{lem:perfect-matchings-to-zero-set}
	Let $M$ be a partial matching on $\Sigma^c$ and $Z=\bigcup M$. 
	Then $M$ is an $i$-zero matching iff $Z$ is an $i$-zero set.
    Furthermore, if $M$ is $i$-dependent then $Z$ is also $i$-dependent (but not vice versa). 
\end{observation}

\medskip\noindent
Moreover, when moving from one level to the next, we will use the following observation:
\begin{observation}[Observation $12$ in \cite{BBKHT23}]\label{lem:zero-set-to-perfect-matchings}
	If $Z$ is an $i$-zero set, then there is a perfect $j$-matching on $Z$ for every 
	level $j\leq i$.
\end{observation}

The obstructions we build will consist of matching at each level. To bound the probability that such an obstruction exists, we will use the following bound repeatedly: 

\begin{lemma}[Lemma $10$ in \cite{BBKHT23}]\label{lem:independent-matchings}
	Let $M$ be a partial matching on $\Sigma^c$. Conditioning on $M$ being $(i-1)$-independent, $M$ is an $i$-matching 
	with probability $1/|\Sigma|^{|M|}$.
\end{lemma}

\subsection{Defining an Obstruction on the Top Two Levels}
\label{sec:definingObstruction}
We distinguish between the top two levels $d$ and $d-1$, and the remaining bottom levels $1,\ldots, d-2$. The obstruction on the top two levels is defined similarly to how it is defined in \cite{BBKHT23}. Namely: if a set of derived keys $\derive X$ is linearly dependent, then it must be the case that there exists a  mininal subset $Z \subseteq X$ that is a $d$-zero set (\cite{BBKHT23} had some special concerns about query keys, but these query keys are not considered here).

By~\Cref{lem:zero-set-to-perfect-matchings}, the set $\derive Z$ exhibits a (perfect) $d$-matching $M^*_{d}$ and a (perfect) $(d-1)$-matching
$M^*_{d-1}$ o(we also have perfect matchings on all the other levels). We follow the edges of these two matching in order to build our obstruction. Namely, these two matchings form alternating cycles on the keys in $Z$. For every such cycle, we choose an arbitrary start vertex $x_1$ and follow the edge from $M^*_{d-1}$. We land at some other vertex $x_2$ and the follow the edge from $M^*_{d}$, and so on an so forth. When we are done with one cycle, we continue with the next one in a similar fashion. We end up with a sequence of vertices $x_1,\ldots,x_{|Z|}$ that describe all vertices in $Z$ such that edges $\{x_{1},x_2\},\{x_3,x_4\},\ldots,\{x_{\size Z-1},x_{\size Z}\}$ describe the edges in $M^*_{d-1}$. 

Among the edges in $M^*_{d-1}$, we now identify a minimal $(d-2)$-dependent sub-matching $M_{d-1}$ by defining $w$ as the smallest value for which $\{x_{1},x_2\},\ldots,\{x_{w-1},x_{w}\}$. We let $W= \{x_1 \dots x_w\}$ be the support of this sub-matching and note that $w$ is even. We also let $M_d$ be to be
the restriction of $M^*_d$
to $\{x_1 \dots x_{w-1}\}$ (without the last vertex we visit). Note that $M_d$ how has $w/2-1$ edges. We use the following
simple lemma that was not part of \cite{BBKHT23}:
\begin{lemma}\label{lem:Ld-1} There
is exactly one submatching $L_{d-1}\subseteq M_{d-1}$ such that $L_{d-1}$ is a $(d-2)$-zero matching.    
\end{lemma}
\begin{proof}
    First, we know $L_{d-1}$ exists because $M_{d-1}$ is
    $(d-1)$-dependent. Suppose
    we had an alternative 
    $L'_{d-1}\neq L_{d-1}$. Then $L''_{d-1}=L'_{d-1}\Delta L_{d-1}$ is
    also a $(d-1)$-zero matching,
    but since $L_{d-1}$ and 
    $L_{d-1}$ both include $\{x_{w-1},x_{w}\}$, $L''_{d-1}$ does not include $\{x_{w-1},x_{w}\}$, but this means that $\{x_{1},x_2\},\ldots,\{x_{w-3},x_{w-2}\}$ is $d-2$-dependent.
\end{proof}
Finally, we define $L_{d-1}$ uniquely
as in \Cref{lem:Ld-1}, and set $Z_{d-1}=\bigcup L_{d-1}$. Then $Z_{d-1}$ is a $(d-2)$-zero set which will play a very crucial role.

\paragraph{The Obstruction on the First Two Levels} For the purposes of our argument, we distinguish between two cases depending on $w$. We define $w_{\max}$ to be the smallest even number above $0.63\size\Sigma$. In our obstruction above, we do not want $w>w_{\max}$, so if $w > w_{\max}$,
we reduce $w$ to $w_{\max}$. In this case,
$M_{d-1}=\{x_{1},x_2\},\ldots,\{x_{w-1},x_{w_{\max}}\}$ is $(d-2)$-independent with at least $w/2-1$ edges. 
If $w > w_{\max}$, we say that the obstruction was \emph{truncated}. Otherwise, we say that the obstruction is \emph{complete}. We then define the following obstruction \(\calO = (W, M_d, M_{d-1}, L_{d-1})\) for the first two levels:
\begin{itemize}
	\item A set of keys $W \subseteq \Sigma^c$ of some even size $w$. 
	\item A matching $M_d$ of size $w/2-1$ on $W$. 
	\item A perfect matching  $M_{d-1}$ on $W$. This matching also contains a $(d-2)$-independent submatching $M_{d-1}'$ with at least $w/2-1$ edges. If $w > w_{\max}$, (the truncated case), $M'_{d-1}=M_{d-1}$. Otherwise, $M'_{d-1}$ is $M_{d-1}$ minus any edge from $L_{d-1}$.
	\item If $w\leq w_{\max}$ (the complete case), we have a submatching $L_{d-1}\subseteq M_{d-1}$ with support $Z_{d-1}=\bigcup L_{d-1}$ and size less than $w_{\max}$. Here  $Z_{d-1}$ should contain at least
 one vertex not matched by $M_d$ (this corresponds to the vertex $x_w$ that we do not mention explicitly among the components). Note that in the truncated case,  we
do not store $L_{d-1}$ and $Z_{d-1}$. In this case, we are satisfied having the $(d-1)$-independent $M_d$ and the 
$(d-2)$-independent $M_{d-1}$ matching 
with a total of at least $w-2$ edges.
\end{itemize}

\subsection{Confirming an Obstruction}\label{sec:confirming}
For an obstruction \(\calO = (W, M_d, M_{d-1}, L_{d-1})\) to occur among the selected keys, the tornado tabulation hash function $h=\toptab h\circ \derive h$ must satisfy the following conditions:
\begin{enumerate}
\item The keys in $W$ are all selected, that is, $W \subseteq X^{f,h}$. 
\item Either $W$ is $d$-independent, or it is minimally $d$-dependent.
  A minimally \(d\)-dependent \(W\) corresponds to the case where $W = Z$.
\item $M_d$ is a $d$-matching.
\item $M_d$ is $(d-1)$-independent.
\item $M_{d-1}$ is a $(d-1)$-matching.
\item $M_{d-1}$ contains a $(d-2)$-independent submatching $M'_{d-1}$ with at least $w/2-1$ edges.
\item For a complete obstruction, $Z_{d-1}=\bigcup L_{d-1}$ is a $(d-2)$-zero set.
  %% In this case $M'_{d-1}$ is $M_{d-1}$ minus any edge from $L_{d-1}$.
\end{enumerate}
For a given obstruction we use \((1), (2), \dots\) to denote the event where the tornado tabulation hash function satisfies each of the conditions given above.

When a hash function $h$ satisfies the above conditions, we say that it \emph{confirms} an obstruction, and we want to prove that this happens with small probability.
Our probability bound is parameterized by $w=|W|$.

We bound the probability of satisfying all conditions as
\[
\Prp{\bigcap_{i=1}^7 (i)} \leq
\Prp{(6) \cap (7)}\cdot
\Prpcond{(5)}{(6) \cap (7)}\cdot
\Prpcond{(3)}{\bigcap_{i>3}(i)} \cdot
\Prpcond{(1)}{\bigcap_{i>1}(i)} \, .
\]

For $\Prpcond{(1)}{\bigcap_{i>1}(i)}$, by conditioning on (2) we know that at least $w-1$ derived keys are
hashed independently by $\hat h$.
As each is selected with probability $p$, we get 
\[\Prpcond{(1)}{\bigcap_{i>1}(i)} \leq p^{w-1} \,.\]
For a truncated obstruction, however, we know that \(W\) is a strict subset of \(Z\), and hence \(W\) is \(d\)-independent, giving
\[\Prpcond{(1)}{\bigcap_{i>1}(i)} \leq p^w \, .\]

For $\Prpcond{(3)}{\bigcap_{i>3}(i)}$, by conditioning on (4) we know that all $|M_d|=w/2-1$ diff-keys from $M_d$ are hashed independently by $\derive h_d$, so
\[\Prpcond{(3)}{\bigcap_{i>3}(i)} \leq 1/\ssigma^{w/2-1} \,.\]

For $\Prpcond{(5)}{(6) \cap (7)}]$, by conditioning on (6) there exists a \((d-2)\)-independent \(M'_{d-1}\) whose keys are hashed independently by $\derive h_{d-1}$, so the probability of \(M_{d-1}\) (and thus also \(M'_{d-1}\)) being a \((d-1)\)-matching is at most
\[\Prpcond{(5)}{(6) \cap (7)} \leq 1/\ssigma^{w/2-1} \,.\]

Finally, for truncated obstructions we apply the trivial bound \(\Prp{(6) \cap (7)} \leq 1\), while for complete obstructions we apply \cref{lem:zero}, given below, with \(\size{Z_{d-1}} \leq \size{W} = w\) to obtain
\[\Prp{(6) \cap (7)} \leq \Prp{(7)} \leq (3/\ssigma)^{d-2} \cdot 2^{w /4 + 1} \, .\]

Putting it all together, we obtain the following bounds on \(h\) confirming a given obstruction \(\calO = (W, M_d, M_{d-1}, L_{d-1})\) with \(\size{W} = w\)
\begin{align*}
  \Prp{h \textnormal{ confirms truncated } \calO} &\leq \ssigma^2 \cdot (p/\ssigma)^w \\
  \Prp{h \textnormal{ confirms complete } \calO} &\leq p^{w-1} \ssigma^{2-w} (3/\ssigma)^{d-2} 2^{w/4 +1} \, .
\end{align*}

\begin{lemma}[\cite{BBKHT23}]\label{lem:zero}
If $z_{d-1}=|Z_{d-1}|\leq 0.63\cdot|\Sigma|$ and $\size\Sigma\geq 256$ then 
\begin{equation}\label{eq:general-zero}
    \Pr_{\derive h_{\leq d-2}}[Z_{d-1}\textnormal{ is an $(d-2)$-zero set}]\leq (3/\size\Sigma)^{d-2} \cdot  2^{z_{d-1}/4+1}.
\end{equation}
\end{lemma}

\begin{proof}[Proof sketch]

The bound is implicitly present in~\cite{BBKHT23}, so we sketch the arguments here and refer to reader to the appropriate sections in~\cite{BBKHT23}for details. The main idea is to proceed similarly to how we defined $L_{d-1}$ (and $Z_{d-1}$) from $M_{d-1}$ (Section $3.2$ in~\cite{BBKHT23}). Namely, while going through the matching $M_i$ (which is (minimally) i-dependent), we identify the submatching $L_i$ which is an $i$-zero matching. This then gives rise to a matching $M_{i-1}$ on $Z_{i-1}$ (which is the support of $L_{i-1}$. We do this for every layer from $i=d-2$ to $i=1$. This describes a general obstruction that includes all the matching $M_i$ and supports $Z_i$(Section $3.3$ in~\cite{BBKHT23}). The probability that an obstruction is confirmed is bound in Lemma $14$ in~\cite{BBKHT23}. The difference with what we have is that the only care about the part of the obstruction that deals with levels $1, \ldots, d-2$ (excluding the top two levels). In particular, the event that $Z_{d-1}$ is a $(d-2)$-zero set corresponds to the conjunction $\bigwedge_{i=1}^{d-2 \calC^{(i)}}$ over all possible realizations of $M_{d-2}, Z_{d-2}, M_{d-3}$ etc. We get the following:
\begin{align*}
    \Pr_{\derive h_{\leq d-2}}[Z_{d-1}\textnormal{ is an $(d-2)$-zero set}]&\leq \prod_{i=1}^{d-2}\max_{Z_{i+1}}\ld(\sum_{M_{i},e_{i},L_{i},Z_{i}}|\Sigma|^{1-|M_{i}|}\rd) \\
    &\leq 2^{z_{d+1}/4-1} \cdot \prod_{i=1}^{d-2}\max_{Z_{i+1}}\ld(\sum_{M_{i},e_{i},L_{i},Z_{i}}|\Sigma|^{1-|M_{i}|}\bigg/2^{(|Z_{i+1}|-|Z_{i}|)/4}\rd)\;.
\end{align*}

\noindent
In sections $4.2$ and $5.1$ (specifically Eq ($21$)), it is shown that:
$$\max_{Z_{i+1}}\ld(\sum_{M_{i},e_{i},L_{i},Z_{i}}|\Sigma|^{1-|M_{i}|}\bigg/2^{(|Z_{i+1}|-|Z_{i}|)/4}\rd) \leq 4\cdot (3/|\Sigma|)^{d-2} \;. $$

\noindent
Since $4(3/\size\Sigma)^{d-2} 2^{z_{d-1}/4-1}=(3/\size\Sigma)^{d-2} 2^{z_{d-1}/4+1}$, we get the claim.
\end{proof}

\subsection{Union Bounds over All Obstructions}\label{sec:unionobstruction}
To obtain the bound stated in \cref{thm:new-tech-random-set} we perform a union bound over the probability of confirming each possible obstruction.
As we have defined two types of obstructions, we treat these separately.
\paragraph{Truncated obstructions}
We start with the case where the obstruction has been truncated.

As our bound on the probability of a truncated obstruction is confirmed is identical for all truncated obstructions (they are all of size \(\size{W} = \wmax\)) we just have to count the number of obstructions \(\calO = (W, M_d, M_{d-1})\) to obtain the first part of our union bound.

In the following we let \(w = \wmax\) for improved readability.
The set $W$ can be specified in $\binom{n}{w} \leq \frac{n^w}{w!}$ ways.
The matching $M_{d}$ of size $w/2-1$ over $\{1,\ldots,w\}$ can be described as a perfect matching on $W$ with one edge removed, giving $(w-1)!! \cdot w/2$ possibilities \footnote{We use the notation, $k!! = k \cdot (k-2) \cdot (k-4) \cdot \ldots \cdot 1$.}.
The matching $M_{d-1}$ is perfect, so it can be chosen in $(w-1)!!$ ways.
In total, this means that there exists at most \(((w-1)!!)^2 \cdot w/2 \cdot n^w/w!\) choices for \(\calO = (W, M_d, M_{d-1})\).

We conclude that
\begin{align*}
  \Prp{h \textnormal{ confirms a truncated obstruction}} &\leq \sum_{\textnormal{truncated } \calO} \Prp{h \textnormal{ confirms } \calO} \\
  &\leq \frac{w}{2} \cdot \frac{n^w}{w!} \cdot ((w-1)!!)^2 \cdot \parentheses{\frac{p}{\ssigma}}^w \cdot \ssigma^2\\
  &= f^w \cdot \frac{((w-1)!!)^2}{2(w-1)!} \cdot \ssigma^2\\
  &= f^w \cdot \frac{(w-1)!!}{2(w-2)!!} \cdot \ssigma^2\, ,
\end{align*}
using that \(np/\ssigma = \mu/\ssigma = f\) and \((w-1)! = (w-1)!! \cdot (w-2)!!\).
Note that \((w-1)!!/(w-2)!! \leq 3/2\) for all \(w \geq 4\).
As \(w = \wmax \geq 0.63 \ssigma\), we have \(f^w \leq f^{\ssigma/2} \cdot f^{0.13\ssigma} \leq f^{\ssigma/2} \cdot (1/2)^{0.13\ssigma}\).

%% SOLVE: https://www.wolframalpha.com/input?i2d=true&i=Divide%5B3*Power%5Bx%2C2%5D%2C4%5D*Power%5B%5C%2840%29Divide%5B1%2C2%5D%5C%2841%29%2C0.13*x%5D%3D1
For \(\ssigma > 100\) we have \(\frac{3}{2\cdot 2} \cdot \ssigma^2 \cdot (1/2)^{0.13 \ssigma} < 1\), and thus
\[\Prp{h \textnormal{ confirms a truncated obstruction}} \leq f^{\ssigma/2} \, .\]

%% For $w=w_{\max}<0.63\size\Sigma$, this is less than $f^{\size\Sigma/2}$. More specifically that $f^{0.13\size\Sigma-2}<n\mu\sqrt{w_{\max}}$.
%% \mtcom{need something about $\size\Sigma\geq 256$. This should be exactly same argument as in \cite{BBKHT23}}

\paragraph{Complete Obstructions}
%% \paragraph{Smaller $w<w_{\max}$ via ordered obstructions}
For complete obstructions where \(w \leq \wmax\) we need to be more careful, as the probability of an obstruction being confirmed depends on its size \(w = \size{W}\).
We will consider each value of \(w\) in turn, and let \(P(w) = \Prp{h \text{ confirms a complete obstruction of size } w}\).
Summing \(P(w)\) over all even \(w \in \set{4, 6, \dots, \wmax}\) we get a bound on \(\Prp{h \text{ confirms a complete obstruction}}\).

Instead of the set \(W\) we will let the first component of the obstruction be a vector \(\vec W = (x_1, \dots, x_w) \in S^w\).
Before specifying \(\vec W\), however, we will define \(M_d, M_{d-1}, L_{d-1}\) as matching on the index set \(\set{1, \dots, w}\).
%% A consequence is that the  same  set $W$ is considered by $w!$ vectors
%% $\vec W$. Therefore

We specify the obstruction in the following order:
\begin{enumerate}
\item First we choose which indices correspond to \(M_d\) and \(M_{d-1}\).
\item Next we specify which edges of \(M_{d-1}\) are contained in \(L_{d-1}\).
\item Then we describe which keys of \(S\) go into the locations of \(\vec W\) corresponding to \(Z_{d-1} = \bigcup L_{d-1}\).
\item Finally, we choose which keys go into the remaining positions of \(\vec W\).
\end{enumerate}
In this way each obstruction will be accounted for \(w!\) times.
Note that for \(Z_{d-1} = \bigcup L_{d-1}\) to be a \((d-2)\)-zero set, then \(Z_{d-1}\) must also be a zero set -- the keys themselves, before computing any derived characters.
Thus
\[
P(w)\leq\sum_{
\substack{M_d,M_{d-1},L_{d-1},\\
\vec W=(x_1,\ldots,x_w)\in S^w,\\ 
\Delta_{i\in Z_{d-1}}x_i=\emptyset}}
\frac{\Prp{h \textnormal{ confirms } \calO=(\vec W, M_d, M_{d-1}, L_{d-1})}}{w!} \, .
\]

In the following, we bound the number of ways to perform each of the four steps outlined above.
\subparagraph{1.}
As discussed in the previous section, \(M_d, M_{d-1}\) can be chosen among the \(w\) indices in \(((w-1)!!)^2 \cdot w/2\) ways.

\subparagraph{2.}
Let \(\set{i, j}\) be the two indices of \(\set{1, \dots, w}\) not covered by \(M_d\).
At least one of these indices must be covered by \(L_{d-1}\), as discussed in \cref{sec:definingObstruction}.
We distniguish between two cases: If \(\set{i, j} \in M_{d-1}\), this edge must be included in \(L_{d-1}\), giving at most \(2^{\size{M_{d-1}} -1} = 2^{w/2 - 1}\) valid submatchings of \(M_{d-1}\).

If \(\set{i, j} \not\in M_{d-1}\), then there exists edges \(\set{i, i'}\) and \(\set{j, j'}\) in \(M_{d-1}\).
We can include one, the other, or both of these in \(L_{d-1}\), along with any subset of the remaining \(\size{M_{d-1}} - 2 = w/2 - 2\) edges, giving \(3 \cdot 2^{w/2 - 2}\) options.

The choice made in step 1 decides which case applies, and thus there are at most \(3 \cdot 2^{w/2 - 2}\) ways of performing step 2.

\subparagraph{3.}
Let \(z = \size{Z_{d-1}} = 2\size{L_{d-1}}\).
As these \(z\) entries of \(\vec W\) must form a zero set, \cref{cor:3cSequences} states that the keys for these positions can be chosen in at most \(3^c \cdot n^{z - 2}\) ways.

\subparagraph{4.}
The remaining \(w - z\) entries of \(\vec W\) can be chosen in at most \(n^{w-z}\) ways.

Multiplying the number of choices for each of the four steps, the total number of complete obstructions of size \(w\) is bounded by
\begin{align*}
  3 \cdot 2^{w/2-2} \cdot ((w-1)!!)^2 \cdot w/2 \cdot n^{w-2} \cdot 3^c
\end{align*}
and hence
\begin{align*}
  P(w) &\leq 3 \cdot 2^{w/2-2} \cdot ((w-1)!!)^2 \cdot w/2 \cdot n^{w-2} \cdot 3^c \cdot \frac{p^{w-1} \ssigma^{2-w} (3/\ssigma)^{d-2} 2^{w/4 +1}}{w!} \\
  &= \frac{3}{2} \cdot 2^{3w/4 -1} \cdot \frac{((w-1)!!)^2}{(w-1)!} \cdot \parentheses{\frac{pn}{\ssigma}}^{w-1} \cdot \frac{\ssigma}{n} \cdot 3^c \cdot \parentheses{\frac{3}{\ssigma}}^{d-2} \\
  &= 9 \cdot  2^{3w/4 -2} \cdot \frac{(w-1)!!}{(w-2)!!} \cdot f^{w-1} \cdot \frac{3^c}{n} \cdot \parentheses{\frac{3}{\ssigma}}^{d-3} \\
  &= \frac{2^{3w/4 -2}}{3} \cdot \frac{(w-1)!!}{(w-2)!!} \cdot f^{w-4} \cdot \mu^3 \cdot \frac{3^c}{n} \cdot \parentheses{\frac{3}{\ssigma}}^{d} \, .
\end{align*}
What is left is to bound \(\sum_{\textnormal{even }w=4}^{w_{\max}} P(w)\).
Let \(g(w) = \frac{2^{3w/4 -2}}{3} \cdot \frac{(w-1)!!}{(w-2)!!} \cdot f^{w-4}\).
\begin{observation}
\[\sum_{\textnormal{even }w=4}^{w_{\max}} g(w) \leq 3 \]
\end{observation}
\begin{proof}
  First, observe that \(g(w+2) = \frac{w+1}{w} \cdot 2^{3/2} f^2 \cdot g(w) \leq \frac{w+1}{w} \cdot 0.71 \cdot g(w)\), when \(f \leq 1/2\).
  For any fixed \(k\) we thus have
  \[
  \sum_{\text{even } w=k}^{w_{\max}} g(w) \leq g(k) \cdot \sum_{i=0}^{\infty} \parentheses{\frac{k+1}{k} \cdot 0.71}^i \, .
  \]

\noindent  For \(k=4\) we thus have
  \begin{align*}
    \sum_{\text{even } w=4}^{w_{\max}} g(w) &\leq g(4) \cdot 9 = 9 \, ,
  \end{align*}
  as \(g(4) = \frac{2}{3} \cdot \frac{3}{2} = 1 \).
\end{proof}

\medskip\noindent
Hence \(\sum_{\textnormal{even }w=4}^{w_{\max}} P(w) \leq 9 \mu^3 (3/\ssigma)^d \cdot 3^c/n\), and the probability that \(h\) confirms \emph{any} obstruction is bounded by \(9 \mu^3 (3/\ssigma)^d \cdot 3^c/n + f^{\ssigma/2}\). Thus we get the claim in \Cref{thm:new-tech-random-set}.

\renewcommand{\arraystretch}{1.5} % Default value: 1
\begin{table}
\centering
	\begin{tabular}{| c || c || m{63mm} | }
	\hline
Symbol & Definition & Description \\ \hline\hline
\(\ssec\)		&  $20$  & Scales error probability of secondary events \\ \hline 
\(\sall\)	&  160 & Scales error probability in each layer \\ \hline \hline
\(\preg\)		&  $\frac{p}{\log_6\left(\frac{\mu/2}{\ln(\ssec/p)}\right)} \cdot \frac{1}{\sall}$ & Threshold for error probability in regular layers \\ \hline
\(\inr\)    &   \(\min\set{i : \bar \mu_i < \ln(\sall/\preg)}\) & First non-regular layer \\ \hline
\(\imax\) & \(\ln(\ssigma^{d-2} 2^\selBits)\) & Maximum number of layers with non-zero expected size (whp), where \(\selBits\) is the number of selection bits.
Defined in \cref{lem:imax}. \\ \hline
\(\ntop\)		&  $\log_6(\ln(\ssec/\preg) \cdot \sall/p)$ & Bound on the number of layers handled by \cref{lem:topLayers_2} \\ \hline
\(\ptop\) & \(\min\set{\frac{\preg}{\ssec} ,\, \frac{p}{\ntop \cdot \sall}}\) & Error probability for each layer in \cref{lem:topLayers_2} \\ \hline\hline

% $\varepsilon_2$		&    $\sqrt{6 \ln(\ssec/p)/\ssigma} + (\frac{2}{9} \ln(1/p) + 2)/\mu$ & \\ \hline
$\varepsilon_3$ & $ (2+ \sqrt{6}) \sqrt{\ln(\ssec/p)/\ssigma} + (\frac{1}{2} \ln(1/p) +6)/\mu $  & Stretch factor of multiplicative deviation of \cref{lem:3layers_2} \\ \hline

$\Delta_{reg}$			&  $  0.181\sqrt{\ln(1/p) \mu}
+ 0.066\sqrt{\frac{\mu}{\ln(1/p)}}
+ \sqrt{2}$& Total deviation of regular layers, from  4 and up (\cref{lem:regularLayers_2}) \\ \hline
$\Delta_{\inr}$	 & $ 1.33e \ln(\ssec/\preg) + 1$ & Deviation of layer \(\inr\) (\cref{lem:layeri0_2}) \\\hline

$\Delta_{top}$		&  $2 \ln(1/\ptop) \cdot (2+\ln \ln(\ntop)) + \ntop$& Deviation of layers \(\inr+1\) through \(\iinf\) (\cref{lem:topLayers_2}) \\\hline
$\Delta_{nonreg}$ &  $ \Delta_{\inr} + \Delta_{top} + 3$& Total deviation of non-regular layers, \(\inr\) and above (\cref{thm:nonRegularLayers}) \\ \hline\hline
$\gamma_1$		&  $\sqrt{7/3 \cdot (1+\varepsilon_3)} + 0.181$ & Multiplicative term of deviation in \cref{thm:ugly} \\\hline
$\gamma_2$		&  $\Delta_{reg} -  0.181\sqrt{\ln(1/p) \mu} + \Delta_{nonreg} + \ln(1/p)$ & Additive term of deviation in \cref{thm:ugly} \\\hline

	\end{tabular}
 \caption{Overview of the symbols used in the statement of \cref{thm:ugly}. See discussion in \cref{layers-ingredients}.}
 \label{symbols}
\end{table}

\section{Proof of \texorpdfstring{\Cref{thm:pretty1}}{Theorem \ref{thm:pretty1}} and \texorpdfstring{\Cref{subsamplingIntro}}{Theorem \ref{subsamplingIntro}}}
\label{sec:proof-of-theorems}
As mentioned earlier,~\cref{subsamplingIntro} will follow from~\cref{thm:pretty1} and we prove this in~\cref{sec:proof-subsamplingIntro}. In order to prove our main~\cref{thm:pretty1}, we start with the following quite technical result. In this result, $\Jmax$ denotes the event that there are no large layers of~\cref{lem:imax}, namely the event that $\Epcond{S_i}{\barh}= 0$ for all $i\geq \imax$. 

\begin{restatable}{theorem}{ugly}\label{thm:ugly}
 If  \(\ssigma \geq 2^{11}\) and \(\mu \in [\ssigma/4, \ssigma/2]\), then the following holds for any $p>0$
	\[\Prp{X < \mu - \sqrt{\ln(1/p) \mu} \cdot \gamma_1 - \gamma_2} < 3 p + \Error \;,\]
 where $\Error = \imax \cdot 2\PnegJ + \PnegJmax  $.
\end{restatable}
The definitions of \(\gamma_1\) and \(\gamma_2\) are quite involved and relies on a number of symbols that will be defined and motivated in \cref{layers-ingredients}.
For \(\ssigma \geq 2^{11}\), \(\gamma_1\) can be considered to be approximately \(2\) while \(\gamma_2\) is of order \(\widetilde O(\ln(1/p) \cdot \ln \ln(\mu))\).
The full definition of \(\gamma_1\) and \(\gamma_2\) can be found in \cref{symbols} on \cpageref{symbols}.
\subsection{Proof of \texorpdfstring{\Cref{thm:ugly}}{Theorem \ref{thm:ugly}}}
\label{sec:proof-ugly}
The proof essentially combines the analyses done for the different layers in~\cref{sec:layers}.
\begin{proof}
The proof will proceed by applying a union bound over the contribution of all layers.
We begin by noting that we can assume that $\ln(1/p) \leq  \mu/8.6$. Namely, we claim that $\ln(1/p) > \mu/8.6$ implies that $\gamma_2 > \mu$, which, in turn, makes the event in \cref{thm:ugly} trivially false. To see this, we note that  $\gamma_2 \geq 8.6 \ln(1/p)$ always, since the following hold regardless of $p$:
\begin{itemize}
    \item $\Delta_{i_{nr}} \geq 1.33\cdot e \cdot \ln(1/p) \approx 3.61 \cdot \ln(1/p)$ because $s_{sec}/p_{reg} \geq 1/p$ and  $p_{reg} \leq p$
    \item $\Delta_{top} \geq 4 \cdot \ln(1/p)$ since $p_{top} \leq p$
\item $\Delta_{reg} \geq 0.181 \sqrt{\ln(1/p)\mu}$ by definition \;.
\end{itemize}

Assuming that $\ln(1/p) \leq  \mu/8.6$, together with the assumptions in the theorem statement, gives us that $\bar \mu_2 \geq \ln(\ssec/p)$. This means that  layers $1$ and $2$ are certainly regular, i.e., $\inr \geq 3$ (recall that $\inr$ was defined as the smallest index for which the corresponding layer is \emph{not} regular). We now distinguish between whether layer $3$ is also regular or not.

If layer $3$ is regular, we use~\Cref{lem:3layers_2} to bound the contribution of the first three layers.
This, together with the contribution of the remaining regular layers from~\Cref{lem:regularLayers_2}, gives us:

\[\Prp{\sum_{i=1}^{\inr-1} S_i < \sum_{i=1}^{\inr-1} \Epcond{S_i}{\barh} - \sqrt{\ln(1/p)\mu} \cdot \gamma_1 - \gamma_2 + \Delta_{nonreg}} < 2.96 p + \inr \cdot 2\PnegJ \, .\]

The contribution of the non-regular layers is given by \cref{thm:nonRegularLayers}. If $\inr = 3$, i.e., layer $3$ is not regular, then note that no other higher index layer can be regular either. We then use \cref{lem:2layers_2} to bound the contribution of the first two layers and note that the bound is stronger than if we had used \cref{lem:3layers_2} (the one for the first three layers combined). We then proceed to consider the non-regular cases in a similar way as before. 

Finally, by \cref{lem:imax}, with probability $\PnegJmax$, we can assume that the contribution from higher layers is zero since: $$\sum_{i=\imax+1}^\infty \Epcond{S_i}{\barh} =0 \;,$$ and hence $$\sum_{i=1}^{\imax}\Epcond{S_i}{\barh} = \mu \;.$$

\noindent
At this point, we get that: 

\[\Prp{X < \mu - \sqrt{\ln(1/p) \mu} \cdot \gamma_1 - \gamma_2} < 3 p +  \imax \cdot 2\PnegJ + \PnegJmax   \;,\]
which matches the claim.
\end{proof}

We now bound the $\Error$ terms by further assuming that $c$ and $d$ are not too large. Note that similar bounds can be obtained even for bigger $c$ and $d$.

\begin{lemma}\label{lem:value-error} For $\ssigma\geq 2^{11}$ and $c \leq \ln \ssigma$, the following holds:
  $$\Error  \leq \ValueError \,.$$
\end{lemma}

\begin{proof}
\noindent
Recall that $\imax = \ln(\ssigma^{d-2}\cdot 2^\selBits)$. We now derive expressions for $\imax$, $\PnegJ$ and $\PnegJmax$.
Since $c \leq \ssigma/(2\log \ssigma)$ by assumption  and $2^\selBits \leq \ssigma^c$ (the universe), we have that
\begin{align*}
\imax &\leq \ln(\ssigma^{d-2}\cdot \ssigma^c) \leq (c+d -2)\cdot \ln\ssigma \;.
\end{align*}
From \Cref{eq:local-uniformity}, we also have that:
\begin{align*}
    \PnegJ &\leq \DPmax 
\end{align*}
Finally, we have that
\begin{align*}
   \PnegJmax &\leq \parentheses{\frac{1}{\ssigma}}^{d-3} + 3^{c+1} \parentheses{\frac{3}{\ssigma}}^{d-1} + \parentheses{\frac{1}{\ssigma}}^{\ssigma/2 - 1}
\end{align*}

\noindent
Thus
\begin{align*}
2\imax \PnegJ \leq 2(c+d-2) \ln(\ssigma) \cdot \parentheses{24\parentheses{\frac{3}{\ssigma}}^{d-3} + \parentheses{\frac{1}{2}}^{\ssigma/2}} \, .
\end{align*}
Note that \((1/\ssigma)^{d-3} +3^{c+1}/\ssigma (3/\ssigma)^{d-1}  \leq (3/\ssigma)^{d-3}\) and \((1/\ssigma)^{\ssigma/2-1} < 1/2^{\ssigma/2}\).
We have
\begin{align*}
    \Error &= \imax \cdot 2\PnegJ + \PnegJmax \\
    &\leq (c+d-2)\ln(\ssigma) \cdot \parentheses{49\parentheses{\frac{3}{\ssigma}}^{d-3} + 3\parentheses{\frac{1}{2}}^{\ssigma/2}}  \, .
\end{align*}

\end{proof}
%% Imported from layers-ingredients.tex
\subsection{Proof of \texorpdfstring{\Cref{thm:pretty1}}{Corollary \ref{thm:pretty1}}}\label{sec:proof-pretty}
With the technical~\cref{thm:ugly} in hand, we can now prove our main~\cref{thm:pretty1}. The proof 
is technical but the goal is clear: To unwind the unwieldy expression expressions of~\cref{thm:ugly}. We restate the theorem below 
\prettyOne*
\begin{proof}
Let \(b = d-3\).
  We will show that \(\gamma_1 + \frac{\gamma_2}{\sqrt{\ln(1/p)\mu}} \leq \sqrt{7}\) whenever \(p, \mu\), and \(\ssigma\) obey the stated restrictions, such that the \lcnamecref{thm:pretty1} follows from \cref{thm:ugly,lem:value-error}.
  Before tackling \(\gamma_1\) and \(\gamma_2\), however, we will bound the involved symbols in terms of parameters \(\ssigma, \mu\) and \(p\):
  \begin{align*}
    \varepsilon_3 &= (2+\sqrt{6}) \sqrt{\frac{\ln(\ssec/p)}{\ssigma}} + \frac{\frac{1}{2} \ln(1/p) + 6}{\mu} \\
    &\leq (2 + \sqrt{6}) \sqrt{\frac{\ln(1/p) + \ln(20)}{\ssigma}} + \frac{2 \ln(1/p) + 24}{\ssigma} \, , \\[1em]
    \preg &\geq \frac{p}{160 \cdot \log_6(\mu)} \, ,\\[1em]
    \ln(\ssec/\preg) &\leq \ln(1/p) + \ln\ln(\mu) + 8.1 \, \\[1em]
    \ntop &\leq \log_6(1/p) + \log_6(\ln(\ssec/\preg)) + \log_6(160) \\
    &\leq \ln(1/p) + \ln\ln\ln(\mu) + 4 \, ,\\[1em]
    \ln(1/\ptop) &\leq \ln(1/p) + \ln(160) + \max\set{\ln\ln(\mu), \ln(\ntop)} \\
    &\leq \ln(1/p) + \ln\ln(1/p) + \ln\ln(\mu) + 6.5 \, .
  \end{align*}

  Keeping \(\mu\) fixed, we see that our bound on \(\gamma_1 + \frac{\gamma_2}{\sqrt{\mu\ln(1/p)}}\) is maximized either when \(\ln(1/p)\) goes towards zero
  (where the constant terms and dependencies on \(\ln(\mu)\) dominate)
  or when \(\ln(1/p)\) goes towards infinity.
  Further, it is clear that both \(\gamma_1\) and \(\gamma_2/\sqrt{\mu}\) decreases for larger \(\mu\) as all terms of \(\gamma_2\) are of order
  \(O\left(\ln(1/p) \cdot (\ln\ln\ln(1/p) + \ln\ln(\mu)) \right)\), with the higher-order terms found in \(\Delta_{top}\).
  We will thus evaluate the expression at the extremal points given by the restrictions of the theorem.
  As the statement is trivially true when \(p > 1/3\), this is \(\ln(1/p) = \ln(3) \approx 1.09\) and \(p = 1/\ssigma^b\).

  Next, we will argue that setting \(b > 1\) will only lead to a stronger bound on \(\gamma_2/\sqrt{\mu\ln(1/p)}\) when \(p = 1/\ssigma^b\),
  due to the stronger requirement on \(\ssigma\) that follows.
  To see this, let \(\phi \geq 1\) such that \(\ssigma = \phi \cdot 2^{16} \cdot b^2\).
  Then, when \(p = 1/\ssigma^b\), the following ``atomics'' of \(\gamma_1\) and \(\gamma_2/\sqrt{\ln(1/p)\mu}\) are all maximized at \(b=1\).
  \begin{align*}
    \frac{\ln(1/p)}{\sqrt{\ln(1/p)\mu}}
    &= \sqrt{\frac{\ln(1/p)}{\mu}} \leq \sqrt{\frac{b \cdot \left(2 \ln(b) + \ln(\phi \cdot 2^{16})\right)}{\phi \cdot 2^{14} \cdot b^2}} \leq 0.0261 \\
    \frac{\ln(1/p) \cdot \ln\ln(\mu)}{\sqrt{\ln(1/p) \mu}}
    &\leq  \ln\ln(\phi \cdot 2^{14} \cdot b^2) \cdot \sqrt{\frac{b \cdot (2\ln(b) + \ln(\phi\cdot 2^{16}))}{\phi\cdot 2^{14} \cdot b^2}} \leq 0.0592 \\
    \frac{\ln(1/p) \cdot \ln\ln\ln(1/p)}{\sqrt{\ln(1/p) \mu}}
    &\leq \ln\ln(b\ln(\phi \cdot 2^{16} \cdot b^2)) \cdot \sqrt{\frac{b \cdot (2\ln(b) + \ln(\phi\cdot 2^{16}))}{\phi\cdot 2^{14} \cdot b^2}} \leq 0.0229.
  \end{align*}

  All that's left now is to evaluate \(\gamma_1 + \gamma_2/\sqrt{\mu\ln(1/p)}\) using the bounds on the underlying symbols given above, setting \(p = 1/\ssigma\) and \(p=1/3\).

  At \(p = 1/3\), \(\gamma_1 + \gamma_2/\sqrt{\mu\ln(3)} \leq 2.58 \leq \sqrt{7}\).

  At \(p = 1/\ssigma\), \(\gamma_1 + \gamma_2/\sqrt{\mu\ln(\ssigma)} \leq 2.34\).
\end{proof}

% !TeX root = main.tex
%% Imported from layers-ingredients.tex
\subsection{Subsampling and Proof of~\texorpdfstring{\cref{subsamplingIntro}}{Theorem \ref{subsamplingIntro}}}
\label{sec:proof-subsamplingIntro}
In this section, we show a method for extending the bound of \cref{thm:pretty1} to smaller \(\mu\) while also allowing us to derive stronger concentration bounds when \(\mu \ll \ssigma\).
For these results we require \(X\), the set of selected keys, to be defined as the preimage of a collection of hash values that are all contained within a dyadic interval \(I\).
Then \cref{thm:pretty1} bounds the number of keys hashed into \(I\) while \cref{eq:local-uniformity} shows that all of these keys are independently and uniformly distributed within \(I\) with high probability.
This allows us to apply a standard Chernoff bound to bound how many of the keys hashed into \(I\) are also in \(X\).

For a fixed keyset \(S\) with \(\size{S}=n\) and a subset of hash values \(I \subset [2^l]\) define \(X_I = \set{x \in S \mid h(x) \in I}\) to be the random variable that defines the preimage of \(I\) and \(\mu_I = \Ep{\size{X_I}} = \size{I}/2^l \cdot n\).

\begin{restatable}{theorem}{subsampling} \label{thm:subsampling}
  If \(I\) is contained within a dyadic interval \(I'\) such that \(\mu_{I'} \in [\ssigma/4, \ssigma/2]\) and \(\ssigma \geq 2^{16} \cdot b^2\), then for any \(p > 1/\ssigma^b\) it holds that
  \[ \Prp{\size{X_I} < \mu_I - (\sqrt{2}+\sqrt{\eps})\sqrt{\ln(1/p) \mu_I } \land \calJ'} < 4 p +\Error \]
  where \(\eps = 7 \cdot (\mu_I/\mu_{I'})\)
  and \(\calJ' = \calI(\derive h_{<{c+d}}(X_{I'}))\) is the event that the derived keys \(\derive{h}(X_{I'})_{c+d-1}\) are linearly independent.
\end{restatable}

\begin{proof}
  Let \(\mu' = \Epcond{\size{X_I}}{\size{X_{I'}}} = \size{X_{I'}} \cdot \mu_I/\mu_{I'}\).
  First, observe that
  \begin{align*}
    \left(\mu' \geq \mu_I - t_1\right) \land& \left(\size{X_I} \geq \mu' \cdot \left(1 - \frac{t_2}{\mu_I -t_1}\right)\right) \implies
    \size{X_I} \geq \mu_I - t_1 - t_2 \, ,
    \intertext{hence}
    \Prp{\size{X_I} < \mu_I - t_1 - t_2 \land \calJ'}
    &\leq \Prp{\left( \mu' < \mu_I - t_1  \lor \left(\size{X} < \mu' - t_2 \land \mu' \geq \mu - t_1 \right) \right) \land \calJ'} \\
    &\leq \Prp{\mu' < \mu_I - t_1 \land \calJ'} + \Prp{\size{X} < \mu' - t_2 \land \mu' \geq \mu - t_1 \land \calJ'} \, .
  \end{align*}

  \noindent
By \cref{thm:pretty1},
  \[\Prp{\size{X_{I'}} < \mu_{I'} \cdot \left( 1-  \sqrt{\frac{7 \ln(1/p)}{\mu_{I'}}} \right) \land \calJ' } < 3p + \Error\]
  and thus
  \[\Prp{\mu' < \mu_I - \sqrt{7 \ln(1/p) \mu_I \cdot (\mu_I/\mu_{I'})}  \land \calJ' } < 3p  + \Error \, .\]

  As \(\calJ'\) implies that the elements of \(X_{I'}\) are uniformly and independently distributed within \(I'\), \(\size{X_I}\) follows a binomial distribution with mean \(\mu'\).
  Letting \(t_1 = \sqrt{7 \ln(1/p) \mu_I \cdot (\mu_I/\mu_{I'})}\) we thus have for any \(\delta > 0\) that
  \begin{align*}
    \Prp{\size{X_I} < (1-\delta) \mu' \land \mu' \geq \mu_I - t_1 \land \calJ'} &< \exp\left(- \frac{\delta^2 (\mu_I -t_1)}{2}\right) \\
    \intertext{and hence}
    \Prp{\size{X_I} < \left(1- \frac{\sqrt{2 \ln(1/p) \mu_I}}{\mu_I - t_1}\right) \mu' \land \mu' \geq \mu_I - t_1 \land \calJ'} &< \exp\left(- \frac{\ln(1/p) \mu_I (\mu_I - t_1)}{(\mu_I - t_1)^2}\right) \leq p \, .
  \end{align*}

  Let \(t_2 = \sqrt{2\ln(1/p) \mu_I}\). Then we have shown that
  \[\Prp{\size{X_I} < \mu_I - t_1 - t_2 \land \calJ'} < 4p + \Error\,. \]
  Finally, the theorem follows by observing that
  \begin{align*}
    t_1 + t_2 &= \sqrt{\mu_I} \cdot \left(\sqrt{2\ln(1/p)} + \sqrt{7 \ln(1/p) (\mu_{I}/\mu_{I'})}  \right) \\
    &= \sqrt{\mu_I \ln(1/p)} \cdot \left(\sqrt{2} + \sqrt{\eps} \right) \, .
  \end{align*}
\end{proof}

\noindent
\Cref{subsamplingIntro} follows as a direct consequence of \cref{thm:subsampling}
\subsamplingIntro*
\begin{proof}
    Let \(s = \ssigma\) and \(b = d-3\).
    Denote the value \(\eps\) defined in \cref{thm:subsampling} by \(\eps_s\).

    Define \(\hat X\) to be the number of keys \(x \in S\) where \(h(x) \leq \hat t\) for some \(\hat t \in [2^l]\).
    Now, let \(\hat t\) be the maximal power of 2 such that \(\hat \mu = \Ep{\hat X} \leq \ssigma/2\).
    Note that \(\hat \mu \geq \ssigma/4\).

    As \(\mu \leq \ssigma/4\), it follows that \(\hat t \geq t\) and thus \([0, \hat t]\) is a dyadic interval containing \([0, t]\).
    Further, we see that \(7 \cdot \mu/\hat \mu \leq 7 \mu \cdot 4/\ssigma \leq 28 (\mu/s) \leq 28/278\).\Cref{thm:subsampling} then gives
    \begin{align*}
    \Prp{X < \mu - (\sqrt{2}+\sqrt{\eps_s})\sqrt{\ln(1/p) \mu} \land \calJ'} < 4 p + \Error
    \intertext{or, by substituting \(\delta = \left(\sqrt{2}+\sqrt{\eps_s}\right)\sqrt{\ln(1/p)/\mu}\),}
    \Prp{X < (1-\delta)\mu \land \calJ'} < 4 \exp\left(\frac{ - \delta^2 \mu_I}{\left(\sqrt{2} + \sqrt{\eps_s}\right)^2}\right) + \Error \, .
    \end{align*}
    Note that \((\sqrt{2} + \sqrt{\eps_s})^2 \leq 3\).
    
    Meanwhile, \cref{thm:Chernoff-upper} bounds the upper tail.
    Observe that \(\calJ'\) implies \(\calI(\derive h(X))\), as \(\calJ'\) requires independence on a larger keyset while only inspecting the first \(c+d-1\) characters of each.
    \[\Prp{X > \mu + \sqrt{3\ln(1/p)\mu} \land \calJ' } \leq p \, . \]

By \Cref{eq:local-uniformity}, we know that  \(\Prp{\neg \calJ'} \leq  \DPmaxs\) and, from \Cref{lem:value-error} that

$\Error \leq \SValueError$ and the \lcnamecref{subsamplingIntro} follows by adding the two probabilities together.
\end{proof}

\section{Counting Zero Sets}\label{app:zero-sets}
In this section we prove \cref{cor:3cSequences}, which allows for an efficient way of bounding the number of ordered zero sets that can be constructed from a set of \(n\) keys, each \(c\) characters long.
Trivially, if one is looking for a zero-set of size \(k\) the first \(k-1\) keys can be chosen in at most \(n^{k-1}\) ways -- and at most one choice of the final key will make them form a zero set.
\Cref{cor:3cSequences} improves this bound by a factor of \(n/3^c\).

\begin{restatable}{theorem}{combinatorial}\label{lem:combinatorial2ttuples}
  Let \(S \subseteq \Sigma^c\) be a set of \(n\) keys and let \(p\) be a generalized key.
  Then the number of \(2t\)-tuples \((x_1, \dots, x_{2t}) \in S^{2t}\) such that \(\Symmdiff_{i \in [2t]} x_i = p\) is at most \(((2t-1)!!)^c n^t\).
\end{restatable}
\begin{corollary}\label{cor:3cSequences}
  Let \(S \subseteq \Sigma^c\) with \(|S| = n\) and \(k \geq 4\).
  Then at most \(3^c \cdot n^{k-2}\) tuples from \(S^k\) are zero-sets.
\end{corollary}
\begin{proof}
  Consider any prefix \((c_1, \dots, c_{k-4}) \in S^{k-4}\).
  By \cref{lem:combinatorial2ttuples} at most \(3^c \cdot n^2\) tuples \((t_1, t_2, t_3, t_4) \in S^4\) satisfy \(\Symmdiff_{i\in[k-4]} c_i = \Symmdiff_{i\in[4]} t_i\), making the tuple a zero-set.
  Summing over all \(n^{k-4}\) prefixes, the result follows.
\end{proof}

\Cref{lem:combinatorial2ttuples} generalizes \cite[Lemma 2]{dahlgaard15k-partitions} which only applies to the case $p=0$, that is, for zero sets, but this entails that we cannot use it to prove a statement like our \cref{cor:3cSequences} which keeps the dependency on \(c\), the number of characters, at \(3^c\) regardless of the size of zero sets considered.
If \(c\) and \(n\) are known, one could construct a stronger version of \cref{cor:3cSequences} by applying \cref{lem:combinatorial2ttuples} with a value of \(t\) minimizing \(((2t-1)!!)^c/n^t\).

We prove \cref{lem:combinatorial2ttuples} through the following, more general, lemma. \Cref{lem:combinatorial2ttuples} follows by setting all \(A_k = S\).
\begin{lemma}\label{lem:tech-combinatorial2ttuples}
  Let \(A_1, A_2, \dots, A_{2t} \subseteq \Sigma^c\) be sets of keys and \(p \subseteq [c] \times \Sigma\) a generalized key.
  Then the number of \(2t\)-tuples \((x_1, \dots, x_{2t}) \in A_1 \times A_2 \times \dots \times A_{2t}\) such that
  \[
  \Symmdiff_{k \in [2t]} x_k = p
  \]
  is at most \(((2t-1)!!)^c \prod_{k =1}^{2t} \sqrt{\size{A_k}}\).
\end{lemma}

\begin{proof}[Proof of \cref{lem:tech-combinatorial2ttuples}]
The proof proceeds by induction over \(c\).
For \(c=1\) we consider all ways of partitioning the \(2t\) coordinates of the tuple \(\mathbf x = (x_1, \dots, x_{2t})\) into an ordered list of pairs \(((x_{i_k}, x_{j_k}))_{k=1}^t\) with \(i_k < j_k\).
The pairs can be chosen in \((2t-1)!! = (2t-1)\cdot(2t-3)\cdots1\) ways and can be ordered in \(t!\) ways.
All mentions of pairs being before/after each other will be with reference to the chosen ordering, not the natural ordering of the coordinates in the pair.
For \(k \in [t]\) let \(\mathbf x_{<k} = \bigcup_{l<k} \set{x_{i_l}, x_{j_l}}\) be the characters appearing in the first \(k-1\) pairs of coordinates.
We partition the characters of \(p\) into an arbitrary set of pairs \(p = \set{\set{\alpha, \alpha'}, \set{\beta, \beta'}, \dots}\), and we will use the notation \(\alpha'\) to denote the "neighbour" of \(\alpha\) in this pairing, letting \(\alpha'' = \alpha\).

We fix the relationship between characters \((x_{i_k}, x_{j_k})\) in each pair by defining the permutation \(\pi_{\mathcal A}\) on \(\Sigma\) parameterized by a set of characters \(\mathcal A \subseteq \Sigma\).
For the \(k\)'th pair \((x_{i_k}, x_{j_k})\) we require that \(x_{j_k} = \pi_{\mathbf x_{<k}}(x_{i_k})\), where
\[
\pi_{\mathcal A}(\alpha) =
\begin{cases}
  \alpha' & \text{ if } \alpha \in p \setminus \mathcal A \\
  \alpha & \text{ otherwise.}
\end{cases}
\]
In this way each pair of coordinates will contain two copies of the same character, except at most \(|p|/2\) pairs which contain two distinct characters from \(p\), thus ensuring that these characters appear an odd number of times overall.
A pair with two copies of a character \(\alpha \in p\) can only occur when a prior pair has provided the odd copy of \(\alpha\).
Although this may appear to be a crucial limitation of the process we will later show that any tuple \(\mathbf x\) with \(\Symmdiff \mathbf x = p\) can be constructed from several choices of ordered pairings.
Note that all tuples \(\mathbf x\) generated by this procedure will have a subset of \(p\) as symmetric difference, with the symmetric difference being exactly \(p\) iff all \(|p|/2\) "mixed" pairs \((\alpha, \alpha') \in p\) occur.

As \(\pi_{\mathcal A}\) is a permutation on \(\Sigma\) (for any fixed set \(\mathcal A\)) the number of possible assignments \((x_{i_k}, x_{j_k}) \in A_{i_k} \times A_{j_k}\) with \(x_{j_k} = \pi_{\mathcal A}(x_{i_k})\) is at most \(\min\set{\size{A_{i_k}}, \size{A_{j_k}}} \leq \sqrt{\size{A_{i_k}} \size{A_{j_k}}}\).
Counting all tuples in \(A_1 \times \dots \times A_{2t}\) adhering to these restrictions, and summing over all \(t!\cdot (2t-1)!!\) ordered pairings, will yield at most \(t! \cdot (2t-1)!! \prod_{k=1}^{2t} \sqrt{\size{A_k}}\) tuples, of which many will be duplicates counted from several ordered pairings.

We now prove that each tuple with symmetric difference \(p\) can be produced from at least \(t!\) distinct ordered pairings, thus proving the existence of at most \((2t-1)!! \prod_{k=1}^{2t} \sqrt{\size{A_k}}\) distinct such tuples.
Consider a tuple \(\mathbf x\) with \(\Symmdiff \mathbf x = p\), and note that \(\mathbf x\) can be produced from an ordered pairing iff:
\begin{enumerate}
\item[(1)] Each pair \(\set{\alpha, \alpha'} \in p\) appears in a paired set of coordinates \(\set{x_{i_k}, x_{j_k}}\).
\item[(2)] The remaining \(t-|p|/2\) pairs each contain two identical characters.
\item[(3)] Each pair mentioned in (1) precedes all other pairs containing \(\alpha\) or \(\alpha'\).
\end{enumerate}
First, we count the number of ways that positions containing a character from \(p\) can be partitioned into pairs satisfying (1). Given any such pairing, and the assumption that \(\Symmdiff \mathbf x = p\), at least one way of pairing the remaining coordinates of \(\textbf x\) to satisfy (2) exists.
Second, we find the number of ways that a pairing satisfying (1) and (2) can be ordered to satisfy (3).

For \(\alpha \in p\) let \(\#(\alpha)\) be the number of occurrences of \(\alpha\) in \(\mathbf x\).
For each set of neighbours \(\set{\alpha, \alpha'} \in p\) the pair of coordinates mentioned in (1) can be chosen in \(\#(\alpha) \cdot \#(\alpha')\) ways.
There is thus at least \(\prod_{\set{\alpha, \alpha'} \in p} \#(\alpha) \cdot \#(\alpha')\) valid pairings satisfying (1) and (2).

Disregarding (3), the \(t\) pairs can be ordered in \(t!\) ways. We will now compute the fraction of these permutations satisfying (3).
For each pair of neighbours \(\set{\alpha, \alpha'} \in p\) let \(\#(\alpha\alpha')\) be the number of pairs containing \(\alpha\) or \(\alpha'\).
Consider the following procedure for generating all permutations:
First, for each pair \(\set{\alpha, \alpha'} \in p\), choose \(\#(\alpha\alpha')\) \emph{priorities} from \(\set{1, \dots, t}\) which will be distributed amongst pairs of coordinates containing \(\alpha\) and/or \(\alpha'\).
Next, for each pair \(\set{\alpha, \alpha'}\in p\), the paired coordinates containing \(\set{\alpha, \alpha'}\) is assigned one of the \(\#(\alpha\alpha')\) priorities reserved for \(\alpha/\alpha'\).
Afterwards the remaining \((\alpha,\alpha)\)- and \((\alpha', \alpha')\)-pairs are assigned the remaining priorities.
This procedure will generate an ordering on the pairs satisfying (3) exactly when, in the second step, each mixed pair is given the highest priority amongst the \(\#(\alpha\alpha')\) choices.
Thus one in \(\prod_{\set{\alpha, \alpha'} \in p} \#(\alpha \alpha')\) permutations satisfies (3).

Finally, observe that \(\#(\alpha \alpha') = (\#(\alpha) + \#(\alpha'))/2\).
Combining the two counting arguments it is seen that \(\textbf x\) can be created from at least
\[
 t! \prod_{\set{\alpha, \alpha'} \in p} \frac{\#(\alpha) \cdot \#(\alpha') }{\#(\alpha \alpha')} \geq t!
\]
ordered pairings, showing that the lemma holds for \(c=1\).

For \(c > 1\) we assume the lemma to be true for shorter keys, and proceed in a manner similar to that for \(c=1\).
We will, at first, look at the \emph{final} position of all involved keys (that is, the position characters \((c, \alpha)\) for \(\alpha \in \Sigma\)), and will consider the set \(p^c\) of characters from \(p\) appearing in the final position, i.e. \(p^c = \set{\alpha \in \Sigma \, | \, (c, \alpha) \in p}\), which we again consider to be partitioned into pairs \(p^c = \set{(\alpha, \alpha'), (\beta, \beta'), \dots}\).
Generally, we will use \(a^c\) to refer to the character at the \(c\)'th position of the key \(a\) and \(\tilde a = a \setminus a^c\) for the preceding \(c-1\) characters.
We use the same notation for tuples and sets of keys where the operation is applied to each key, \(A^c = \cup_{a\in A} a^c\) and \(\tilde A = \cup_{a\in A} \tilde a\).
For a character \(\beta \in \Sigma\) and set of keys \(A\) let \(A[\beta] = \set{a \in A \, | \, a^c = \beta}\) be the keys of \(A\) having \(\beta\) as their last character.

We consider all ways of partitioning the \(2t\) positions into an ordered set of pairs, this time requiring that the characters in the final position of each key satisfy \(x^c_{j_k} = \pi_{\textbf x^c_{<k}}(x^c_{i_k})\) where, for any \(\mathcal A \subseteq \Sigma\),
\[
\pi_{\mathcal A}(\alpha) =
\begin{cases}
  \alpha' & \text{ if } \alpha \in p^c \setminus \mathcal A \\
  \alpha & \text{ otherwise.}
\end{cases}
\]

For a given sequence \((\alpha_1, \dots, \alpha_t) \in \Sigma^t\) of \(t\) character consider the number of \(2t\)-tuples \(\mathbf x \in A_1 \times \ldots \times A_{2t}\) where \((x^c_{i_k}, x^c_{j_k}) = (\alpha_{i_k}, \pi_{\mathbf x^c_{<k}}(\alpha_{i_k}))\) for all \(k \in [t]\) and \(\Symmdiff \mathbf x = p\).
If \((\alpha_1, \dots, \alpha_t)\) gives \(\Symmdiff \mathbf x^c = p^c\) then counting these tuples is equivalent to counting in how many ways each pair \((\tilde x_{i_k}, \tilde x_{j_k})\) can be chosen from \(\tilde A_{i_k}[\alpha_k] \times \tilde A_{j_k}[\pi_{\mathbf x_{<k}^c}(\alpha_k)]\) such that \(\Symmdiff \mathbf{\tilde x} = \tilde p\).
By the induction hypothesis the number of such \(2t\)-tuples is bounded by
\[
((2t-1)!!)^{c-1} \prod_{k=1}^t \sqrt{\size{A_{i_k}[\alpha_k]} \cdot \big| A_{j_k}[\pi_{ \mathbf x^c_{<k}}(\alpha_k)] \big| } \, .
\]
Summing over all tuples \((\alpha_1, \dots, \alpha_t)\) thus gives an upper bound on the number of \(2t\)-tuples with symmetric difference \(p\) while also adhering to the restrictions imposed by \(\pi\) on the characters in the last position of each key, which in turn is determined by the ordered pairing of the \(2t\) coordinates.
For ease of notation we use the shorthand \(\pi_k\) for \(\pi_{\mathbf x^c_{<k}}\) where \(\mathbf x_{<k}^c\) is understood as the characters \(\alpha_1, \dots, \alpha_{k-1}\) along with their neighbours as determined by \(\pi\).
Thus \(\pi_k\) is dependent on \(\alpha_1, \dots, \alpha_{k-1}\), even if this is not apparent from the notation.

To sum over all choices of \(\alpha\)'s we repeatedly apply Cauchy-Schwarz on the innermost term of the sum.
For any \(k \in [t]\) and fixed \(\alpha_1, \dots, \alpha_{k-1}\) we have
\[
\sum_{\alpha_k \in \Sigma} \sqrt{\size{A_{i_k}[\alpha_k]} \size{A_{j_k}[\pi_k[\alpha_k]]}} \leq \sqrt{\sum_{\beta\in\Sigma} \size{A_{i_k}[\beta]} \sum_{\beta\in\Sigma}\size{A_{j_k}[\pi_k(\beta)]}} = \sqrt{\size{A_{i_k}}\size{A_{j_k}}} \, .
\]
To see that this is an application of Cauchy-Schwarz observe that
\[
\sum_{\alpha \in \Sigma} \sqrt{\size{A_{i_k}[\alpha]} \size{A_{j_k}[\pi_k(\alpha)]}}
\]
is the inner product of the two vectors \(\left(\sqrt{\size{A_{i_k}[\alpha]}}\right)_{\alpha \in \Sigma}\) and \(\left(\sqrt{\size{A_{j_k}[\pi_k(\alpha)]}}\right)_{\alpha \in \Sigma}\) while
\[\sqrt{\sum_{\alpha \in \Sigma} \size{A_{i_k}[\alpha]} \sum_{\alpha \in \Sigma} \size{A_{j_k}[\pi_k(\alpha)]}}\]
is the product of their norms.

Applying the above inequality \(t\) times the total number of tuples for each ordered pairing becomes
\begin{align*}
  &((2t-1)!!)^{c-1} \sum_{(\alpha_1, \dots, \alpha_t) \in \Sigma^t} \prod_{k=1}^t \sqrt{\size{A_{i_k}[\alpha_k]} \size{A_{j_k}[\pi_k(\alpha_k)]}} \\
  &\leq \sqrt{\size{A_{i_t}} \size{A_{j_t}}} \cdot ((2t-1)!!)^{c-1} \sum_{(\alpha_1, \dots, \alpha_{t-1} ) \in \Sigma^{t-1}} \prod_{k=1}^{t-1} \sqrt{\size{A_{i_k}[\alpha_k]}} \cdot \sqrt{\size{A_{j_k}[\pi_k(\alpha_k)]}} \\
  &\leq \prod_{l=t-1}^t \sqrt{\size{A_{i_l}} \size{A_{j_l}}} \cdot ((2t-1)!!)^{c-1} \sum_{(\alpha_1, \dots, \alpha_{t-2}) \in \Sigma^{t-2}} \prod_{k=1}^{t-2} \sqrt{\size{A_{i_k}[\alpha_k]}} \cdot \sqrt{\size{A_{j_k}[\pi_k(\alpha_k)]}} \\
  &\leq \prod_{l=t-2}^t \sqrt{\size{A_{i_l}} \size{A_{j_l}}} \cdot ((2t-1)!!)^{c-1} \sum_{(\alpha_1, \dots, \alpha_{t-3}) \in \Sigma^{t-3}} \prod_{k=1}^{t-3} \sqrt{\size{A_{i_k}[\alpha_k]}} \cdot \sqrt{\size{A_{j_k}[\pi_k(\alpha_k)]}} \\
  & \vdots \\
  &\leq ((2t-1)!!)^{c-1} \prod_{l = 1}^t \sqrt{\size{A_{i_k}} \size{A_{j_k}}} \\
  &= ((2t-1)!!)^{c-1} \prod_{k=1}^{2t} \sqrt{\size{A_k}} \, .
\end{align*}

Summing over all \(t!(2t-1)!!\) ways of partitioning the coordinates into an ordered list of pairs we thus find at most \(t!((2t-1)!!)^c \prod_{k=1}^{2t} \sqrt{\size{A_k}}\) tuples.
By the same counting argument as presented for single-character keys each \(2t\)-tuple \(\mathbf x^c\in\Sigma^c\) with \(\Symmdiff \mathbf x^c = p^c\) and which complies with \(\pi\) is produced from at least \(t!\) distinct ordered pairings of its coordinates. Thus each sequence of applications of the induction hypothesis is repeated at least \(t!\) times.
Hence at most \(((2t-1)!!)^c \prod_{k=1}^{2t} \sqrt{\size{A_k}}\) distinct \(2t\)-tuples \(\mathbf x\) satisfy \(\Symmdiff \mathbf x = p\), proving \cref{lem:tech-combinatorial2ttuples}.
\end{proof}

\bibliographystyle{alpha}
\bibliography{general}

%%%% APPENDIX
\pagebreak
\appendix

\section{An Generalized Chernoff Bounds}\label{app:gen-chernoff}
The following is the standard Chernoff bound, here shown to apply to variables that are not independent, but whose sum is dominated by that of independent indicator variables.
This result is known from \cite{dubhashi1998balls} but we include a proof for completeness.
\begin{lemma}\label{lem:non-independent-chernoff}
  Let \(X_1, \dots, X_n\) be 0/1-variables and \(p_1, \dots, p_n\) be reals such that \(\mu = \sum_i p_i\) and, for all \(I \subseteq [n]\), \(\Prp{\prod_{i\in I} X_i = 1} \leq \prod_{i \in I} p_i\) (implying \(\Ep{X_i} \leq p_i\)) then
  \[
  \Prp{\sum_{i = 1}^n X_i > (1+\delta) \mu} \leq \left( \frac{e^\delta}{(1+\delta)^{1+\delta}}\right)^\mu
  \]
  for any \(\delta > 0\).
\end{lemma}
\begin{proof}
  Let \(a = 1+\delta\), \(X = \sum_{i=1}^n X_i\) and \(s > 0\).
  Let \(Z_1, \dots, Z_n\) be independent 0/1-variables with \(\Ep{Z_i} = p_i\) and define \(Z = \sum_{i=1}^n Z_i\).
  For any \(I \subseteq [n]\) we then have \[\Ep{\prod_{i \in I} X_i} \leq \prod_{i\in I} p_i \leq \Ep{\prod_{i \in I} Z_i} \, .\]

  Let \(i\) be a positive integer. For \(V \in [n]^i\), which may contain duplicate entries, let \(I\) be the distinct elements of \(V\). We similarly have
  \[\Ep{\prod_{v \in V} X_v} = \Ep{\prod_{l \in I} X_l} \leq \Ep{\prod_{l \in I} Z_v} = \Ep{\prod_{v \in V} Z_v}\]
  hence
  \[\Ep{X^i} = \sum_{V \in [n]^i} \Ep{\prod_{v \in V} X_v} \leq \sum_{V \in [n]^i} \Ep{\prod_{v \in V} Z_v} = \Ep{Z^i} \, . \]

  \begin{align*}
    \Prp{X \geq a\cdot \mu}
    &= \Prp{ e^{sX } \geq e^{ s a \cdot \mu} }\\
    &\leq \frac{\Ep{\exp(sX) }}{e^{sa \cdot \mu}} \, . \\
  \end{align*}
  Note that
  \[\Ep{\exp(sX)} = \sum_{i=0}^\infty \frac{s^i \Ep{X^i}}{i!} \leq \sum_{i=0}^\infty \frac{s^i \Ep{Z^i}}{i!} = \Ep{\exp(sZ)} \, .\]
  Due to the independence of \(Z_1, \dots, Z_n\) we further have
  \begin{align*}
    \Ep{\exp(sX)} &\leq \prod_{i=1}^n \Ep{\exp(sZ_i)} \\
    &= \prod_{i=1}^n (p_i \cdot e^s + (1 - p_i)) \\
    &\leq \prod_{i=1}^n (\exp( p_i (e^s - 1)) ) \\
    &\leq \exp(\mu(e^s - 1))
  \end{align*}
  and thus
  \begin{align*}
    \Prp{X \geq a\cdot \mu} &\leq \frac{\exp(\mu(e^s-1))}{\exp(sa \mu)} \\
    &= \left( \frac{e^\delta}{(1+\delta)^{1+\delta}}\right)^\mu
  \end{align*}
  when setting \(s = \ln(a) = \ln(1+\delta)\).
\end{proof}

%\input{locally-uniform-preliminaries} - not needed
%\input{locally-uniform-obstruction} - not needed
%\input{locally-uniform-analysis} - moved it in main

%% Working document for figuring out weighted sampling
%% \input{Notes/poisson-notes}

\end{document}